\documentclass[10pt]{article} 
\usepackage[preprint]{tmlr}

\usepackage{amsmath,amsfonts,bm}

\def\eqref#1{equation~\ref{#1}}

\def\1{\bm{1}}

\DeclareMathAlphabet{\mathsfit}{\encodingdefault}{\sfdefault}{m}{sl}
\SetMathAlphabet{\mathsfit}{bold}{\encodingdefault}{\sfdefault}{bx}{n}

\usepackage{subcaption}
\usepackage{hyperref}
\usepackage{url}
\usepackage[table,xcdraw]{xcolor}
\usepackage{mathtools}
\usepackage{amsthm}
\usepackage{amsmath}
\usepackage{amssymb}
\usepackage{xfrac}
\usepackage{fontawesome}
\usepackage{bm}
\usepackage{pgfplots}
\usetikzlibrary{calc,external,arrows,arrows.meta,shapes,chains,patterns, fadings,positioning,shadows,decorations.pathreplacing,calligraphy,decorations.text}
\usetikzlibrary{plotmarks}

\newtheorem{assumption}{Assumption}
\newtheorem{lemma}{Lemma}
\newtheorem{proposition}{Proposition}

\title{Closing the Loop in Learning with Missing Data}

\author{\name Dimitrios Pylorof \email Dimitrios.ParsinasPylorof@inl.gov \\
      \addr U.S. Department of Energy Idaho National Laboratory
      \AND
      \name Humberto E. Garcia \email Humberto.Garcia@inl.gov \\
      \addr U.S. Department of Energy Idaho National Laboratory}

\begin{document}

\maketitle

\begin{abstract} What should a machine learning model learn when data is missing during training? We look at the learning process from a dynamical systems perspective, cast data missingness as a structured loss of actuation that limits controllability of the parameter error dynamics, and ultimately derive adaptation mechanisms with Lyapunov stability characteristics that throttle model updates in ways that preserve learning coherence under partial, intermittent observability. Under recurrent excitation, our analysis provides ISS-type residual-to-state bounds with respect to a bounded closed-loop mismatch between the loss residual and the preconditioned update geometry. We evaluate the efficacy of our directional observability-aware adaptive learning approach on multimodal contexts, reinforcing its premise in promoting learning coherence and stability even in pathologically sparse domains and problems.
\end{abstract}

\section{Introduction} \label{sec.introduction}

Machine learning models are commonly trained under the implicit assumption that all relevant input and output components are, consistently, fully observed. In operational settings, however, data are frequently missing, censored, or intermittently unavailable due to sensing failures, communication constraints, or data integrity considerations. Existing approaches typically address such missingness through imputation, masking, or marginalization, focusing on statistical consistency or likelihood maximization. These methods offer limited insight into the dynamical behavior of learning algorithms under partial observability, and generally provide no guarantees on the stability or boundedness of the parameter evolution when data are absent.

\begin{figure*}[hb!]
    \centering
    \footnotesize{\input{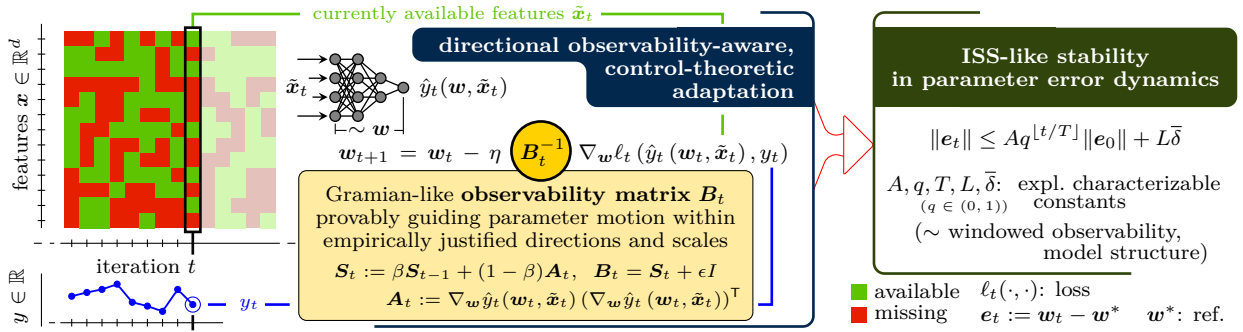}} \vspace{-3mm}
    \caption{A temporal view of learning with missing data, where intermittent observability drives our directionally aware adaptation which, in turn, yields stable and coherent learning dynamics.}
    \label{fig.intro}
\end{figure*}

In this work, we recast supervised learning with missing data as a closed-loop adaptive dynamical system subject to time-varying losses of observability, to provably guide parameter motion within empirically justified directions and margins and promote stability, reliability, and coherence of the learning process. Rather than treating missing components as stochastic noise, we model missingness as a structured, possibly adversarial, reduction in the excitation available to the learning dynamics. The learning parameters are viewed as the system state, while the learning rate, update direction, and regularization strength constitute control inputs that govern how the model adapts in response to available information. Missingness changes both the feedback signal and the actuation geometry; our stability analysis treats the resulting closed-loop residual mismatch explicitly, rather than presuming that missingness enters as an independent exogenous input.

Building on tools from adaptive control and Lyapunov stability theory, we introduce an observability-aware learning framework that explicitly reshapes parameter updates according to directions supported by the currently observed data. By tracking the evolving actuation geometry induced by the data stream, we derive controlled update laws and finite-window conditions under which the parameter error contracts while remaining bounded in the presence of a bounded closed-loop residual mismatch. The resulting learning dynamics admit residual-to-state bounds when this mismatch remains bounded and the effective update geometry is recurrently excited, without requiring imputation, probabilistic modeling of the missingness mechanism, or distributional assumptions on the data.

We derive Lyapunov-style guarantees establishing boundedness of the parameter trajectories and controlled evolution of extrapolative components, even when persistent excitation is intermittently lost. The proposed formulation applies to general supervised learning models, including linear regression and nonlinear function approximators trained via gradient-based methods, and naturally extends to settings with partial output supervision and intermittent labels. From a control perspective, the framework can be interpreted as gain-scheduled adaptive learning under time-varying observability constraints.

By framing learning under missing data as a stabilization problem rather than a purely statistical inference task, this work provides a complementary notion of robustness that addresses failure modes such as parameter drift, hallucination, and instability that are not captured by standard generalization or regret analyses. The proposed perspective bridges adaptive control and modern machine learning, offering principled mechanisms for enforcing stability and bounded behavior in partially observable learning environments.

In Figure \ref{fig.intro}, we illustrate the landscape of learning under missing data for which our methods are developed to eventually yield coherent, reliable learning subject to ISS-like stability, whereas in Figure \ref{fig.intro_directions}, we sketch the potentially detrimental effect of missing data in the actual update directions, compared to full-view learning, and also compared to the update directions that our two methods ---\textit{scalar observability-aware adaptation} and\textit{ directional observability-aware adaptation}--- yield.

\begin{figure}[ht!]
    \centering
    \footnotesize{\begin{tikzpicture}

\node[] at (0,0.7cm){\textbf{full-view SGD}};

 \node[] at (-0.35,-0.1cm){$\bm{w}_t$};
 \node[] at (3.65,-0.1cm){$\bm{w}_t$};
 \node[] at (7.65,-0.1cm){$\bm{w}_t$};
 \node[] at (11.65,-0.1cm){$\bm{w}_t$};

 \node[] at (-0.2,-1.5cm){$\nabla_{\bm{w}} \ell_t(\hat{y}_t(\bm{w}_t, {\bm{x}}_t), y_t)$};
 \draw[rounded corners = 2] (-0.2,-1.2cm) -- (-0.,-0.8cm) -- (0.4,-0.8cm) -- (0.45,-0.7cm);

\draw[-stealth] (0,0cm) -- (2,-2cm);
\draw[black, fill = white] (0,0cm) circle (0.1);

\node[] at (4,0.7cm){\textbf{masked SGD}};

 \draw[red!70!black, rounded corners = 2] (2.9,-0.8cm) -- (3.1,-0.6cm) -- (4.1,-0.6cm);

 \node[red!70!black] at (2.9,-1.1cm){$\nabla_{\bm{w}} \ell_t(\hat{y}_t(\bm{w}_t, \tilde{\bm{x}}_t), y_t)$};
 
 \node[red!70!black] at (5.1,-2.4cm){potentially hallucinated};
  \node[red!70!black] at (5.1,-2.7cm){or incoherent update direction};

\draw[red!70!black,-stealth] (4,0cm) -- (4.8,-2cm);
\draw[dashed,-stealth] (4,0cm) -- (6,-2cm);
\draw[black, fill = white] (4,0cm) circle (0.1);

\node[] at (8,1.0cm){\textbf{scalar}};
\node[] at (8,0.7cm){\textbf{observability-aware}};
\node[] at (8,0.4cm){\textbf{learning}};

\draw[black!40,-stealth] (8,0cm) -- (8.8,-2cm);
\draw[blue!50!black, line width=1,-stealth] (8,0cm) -- (8.4,-1cm);
\draw[dashed,-stealth] (8,0cm) -- (10,-2cm);

 \node[blue!50!black] at (9.0,-2.4cm){moderated update};

\node[blue!50!black] at (7.8cm,-1cm){$\sim$(\ref{eq.linear_regression_adaptive_step})};

\draw[black, fill = white] (8,0cm) circle (0.1);

\node[] at (12,1.0cm){\textbf{directional}};
\node[] at (12,0.7cm){\textbf{observability-aware}};
\node[] at (12,0.4cm){\textbf{learning}};

\draw[black!40,-stealth] (12,0cm) -- (12.8,-2cm);
\draw[black!60, line width=1,-stealth] (12,0cm) -- (12.4,-1cm);
\draw[green!50!black, line width=1.5, -stealth] (12,0cm) -- (13.4,-0.85cm);
\draw[dashed,-stealth] (12,0cm) -- (14,-2cm);

\node[green!50!black] at (13.3cm,-0.4cm){$\sim$(\ref{eq.adaptive_directional_control_law})};

 \node[green!50!black] at (13.2,-2.4cm){empirically justified};
 \node[green!50!black] at (13.2,-2.7cm){\& scaled update};

\draw[black, fill = white] (12,0cm) circle (0.1);

\end{tikzpicture}}
    \caption{Illustration of how missing data (resulting in the masked observation $\tilde{x}_t$) can distort gradient-based learning updates, compared to the full-view SGD update direction (shown on the left, and then repeated with dashed lines), and how our two methods, \textit{scalar} and \textit{directional}-observability-aware adaptation progressively moderate and reshape parameter motion, to promote learning stability and coherence.}
    \label{fig.intro_directions}
\end{figure}
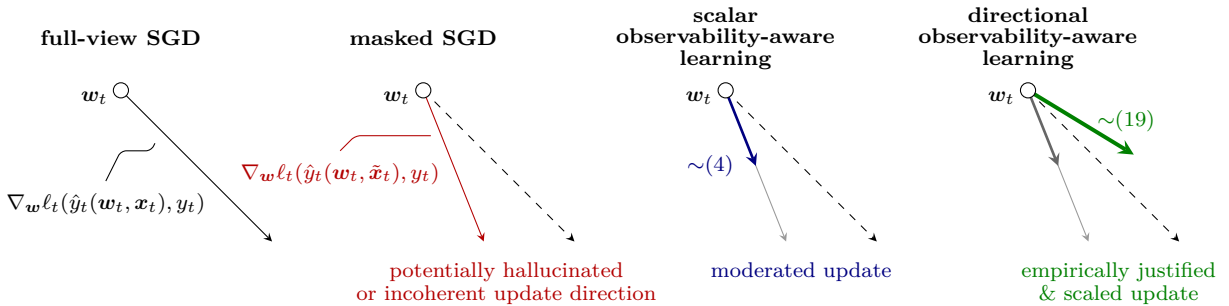

\section{A canonical failure under missing data, and a prototypic defense} \label{sec.prototypic_scalar}

We introduce a minimal supervised learning example illustrating how standard gradient-based training can exhibit unstable or unjustified parameter evolution under partial observability. We emphasize that the failure is dynamical rather than statistical, and arises once realistic learning dynamics such as intermittent observability, feature correlations, regularization, or momentum are taken into account. The purpose of this section is not to solve the general problem, but to isolate the core mechanism motivating our framework that we develop in subsequent sections.

\subsection{Linear regression with masked inputs}
Consider the linear regression model $y_t = \langle \bm{w}^*,\bm{x}_t \rangle$, where $\bm{x}_t \in \mathbb{R}^d$ are input features and $\bm{w}^* \in \mathbb{R}^d$ is a fixed but unknown parameter vector. At each time step $t$, only a subset of the components of $\bm{x}_t$ is observed. We encode this using the diagonal masking matrix $\bm{M}_t \in \mathbb{R}^{d \times d}$ with entries $\{0,1\}$, and define the \textit{observed input} $\tilde{\bm{x}}_t = \bm{M}_t \bm{x}_t$. The learner trains using the \textit{masked} squared loss $\ell_t(\bm{w})  = \frac{1}{2} \left( y_t - \langle \bm{w}, \tilde{\bm{x}}_t \rangle \right)^2$, and updates parameters via stochastic gradient descent:
\begin{equation}
    \bm{w}_{t+1} = \bm{w}_t + \eta \tilde{\bm{x}}_t \left( y_t - \langle \bm{w}_t, \tilde{\bm{x}}_t \rangle \right) \label{eq.linear_regression_sgd_step}
\end{equation}
where $\eta > 0$ is a fixed step size. This update corresponds to standard practice in supervised learning with missing inputs: gradients are computed using only observed components, while missing features are ignored.

\subsection{Residual decomposition under partial observability}
Define the parameter error $\bm{e}_t := \bm{w}_t - \bm{w}^*$ and the residual $r_t := y_t - \langle \bm{w}_t , \tilde{\bm{x}}_t \rangle$. Substituting $ y_t = \langle \bm{w}^*,\bm{x}_t \rangle$ and $\bm{w}_t = \bm{w}^* + \bm{e}_t$ into the residual definition yields:
\begin{align*}
    r_t &= \langle \bm{w}^*, \bm{x}_t \rangle - \langle \bm{w}^* + \bm{e}_t, \bm{M}_t \bm{x}_t \rangle, \nonumber \\ & = \underbrace{- \langle \bm{e}_t , \tilde{\bm{x}}_t \rangle}_{\text{error on observable subspace}} \; + \;  \underbrace{\langle \bm{w}^* , \left( I - \bm{M}_t \right) \bm{x}_t \rangle.}_{=: d_t\; (\text{unobserved contribution})} \label{eq.linear_Regression_residual_decomposition}
\end{align*}
This decomposition separates two distinct contributions: the first term reflects error on the currently observed subspace; the second term is an unknown, unobservable disturbance arising from components of the signal that are not currently sensed. Crucially, the learner cannot distinguish these two contributions using the masked loss alone.

\subsection{Idealized behavior under strict masking}
If a coordinate of $\bm{x}_t$ is masked at time $t$, then the corresponding component of $\tilde{\bm{x}}_t$ is zero, and the associated parameter coordinate in $\bm{w}_t$ receives no gradient update at that time step. In a fully idealized setting (i.e., absent of correlations, momentum, regularization, or noise) strict masking therefore freezes unobserved coordinates. However, rather critically, freezing alone does \textit{not} imply stability of the learning dynamics. To the contrary, freezing \textit{amplifies instability}.

Under partial observability, the parameter vector decomposes into: coordinates that are currently actuated by the gradient, and coordinates that are temporarily frozen. Standard SGD continues to evolve the actuated coordinates as if the entire parameter vector were observable. The resulting updates are computed with respect to a reduced observation model and are therefore inconsistent with the frozen components (and whatever the value of these components is in reality unbeknownst to the learner). This asymmetry induces relative misalignment within the parameter vector: some coordinates evolve rapidly while others remain fixed, even though the loss function implicitly couples them through the residual.

\subsection{Partial actuation and coherence loss}
From a dynamical systems perspective, masking induces a loss of actuation authority over a subset of the parameter state. Standard SGD does not account for this loss and therefore applies control inputs (gradient updates) \textit{designed for the fully actuated} system to a \textit{partially actuated} one. In physical systems, such behavior leads to instability: when some actuators fail, aggressive inputs from the remaining actuators cause distortion rather than coherent motion. The analogous effect in learning is \textit{internal inconsistency} between parameter components trained under different observation regimes. This effect is exacerbated in realistic learning settings where observability is intermittent rather than permanent, features are correlated, momentum or weight decay induces motion even in unobserved directions or otherwise influences training, and/or stochastic noise accumulates over time.

The resulting failure is \textit{dynamical} rather than statistical: the learning algorithm operates \textit{open-loop} and under \textit{partial actuation}. A concrete numerical illustration of the fault mode that arises under partial actuation of the learning dynamics is given in Appendix \ref{sec.app.demo}.

\subsection{A prototypic observability-aware update} \label{sec.prototypic_update}
To restore coherence, we introduce a simple modification inspired by adaptive control. Define the instantaneous observed-feature energy:
\begin{equation}
    \gamma_t := \| \tilde{\bm{x}}_t \|^2. \label{eq.prototypic.observed_feature_energy}
\end{equation}

We replace the fixed step size $\eta$ in the method (\ref{eq.linear_regression_sgd_step}) with the \textit{observability-dependent} gain:
\begin{equation}
    \alpha_t := \frac{\eta}{\epsilon + \gamma_t}, \label{eq.linear_regression_adaptive_step_alpha}
\end{equation}
where $\epsilon >0$ is a constant, to obtain the following iterative update:
\begin{equation}
    \bm{w}_{t+1} = \bm{w}_t + \alpha_t \tilde{\bm{x}}_t \left( y_t - \langle \bm{w}_t, \tilde{\bm{x}}_t \rangle \right). \label{eq.linear_regression_adaptive_step}
\end{equation}

Although the scalar gain increases as $\gamma_t$ decreases, the effective update magnitude satisfies
\begin{equation}
    \| \Delta \bm{w}_t \| = \| \bm{w}_{t+1} - \bm{w}_t \| = \alpha_t \| \tilde{\bm{x}}_t \| \vert r_t \vert = \frac{\eta}{\epsilon + \gamma_t} \|\tilde{\bm{x}}_t\| \vert r_t \vert = \eta \vert r_t \vert \frac{\| \tilde{\bm{x}}_t \|}{ \epsilon + \| \tilde{\bm{x}}_t \|^2},
\end{equation}
which vanishes as observability is lost. Thus, when only a small fraction of the parameter vector is actuated, \textit{all} parameter updates are proportionally throttled. We note that this update does not aim to mobilize masked coordinates or in any way guess or impute their missing values. Instead, it \textit{limits the motion of unmasked coordinates} so that learning proceeds more coherently under partial actuation.

\subsection{Lyapunov interpretation and ISS-style bound} \label{sec.prototypic_ISS}

Consider the Lyapunov function $V_t = \| \bm{e}_t \|^2$, and let $a_t := \langle \bm{e}_t, \tilde{\bm{x}}_t \rangle$. A direct calculation gives
\begin{equation}
    V_{t+1} - V_t = - \left( 2 \alpha_t - \alpha_t^2 \gamma_t \right) a_t^2 + 2\alpha_t \left(1-\alpha_t\gamma_t\right)a_td_t + \alpha_t^2\gamma_td_t^2, \label{eq.prototypic.V_exact}
\end{equation}
where $\gamma_t = \|\tilde{\bm{x}}_t \|^2$ and $d_t = \langle \bm{w}^* , \left( I - \bm{M}_t \right) \bm{x}_t \rangle$. If $0<\eta\leq 1$, then $\alpha_t\gamma_t\leq 1$, and Young's inequality yields the dissipation-disturbance bound
\begin{equation}
    V_{t+1} - V_t \leq -\frac{\alpha_t}{2}a_t^2 + \left(2\alpha_t+\alpha_t^2\gamma_t\right)d_t^2. \label{eq.prototypic.V_ineq}
\end{equation}
Inequality (\ref{eq.prototypic.V_ineq}) exhibits the standard structure familiar from input-to-state stability (ISS) analyses. In particular, under reasonable finite-window excitation and bounded-feature conditions, it implies contraction of the parameter error up to a disturbance-dependent neighborhood, yielding an ISS-style bound:
\begin{equation}
    \| \bm{e}_t \| \leq \rho^t \| \bm{e}_0 \| + C \sup_{\tau < t} \vert d_\tau \vert, \qquad \rho \in (0,1),  \quad C\geq 0. \label{eq.prototypic.ISS}
\end{equation}

\subsection{Beyond linear regression: observability-driven adaptation in implicit Hilbert spaces and deep neural networks} \label{sec.prototypic_nonlinear_intro}
Take $\epsilon-$SVR (Support Vector Regression) with a kernel $k(\bm{x}, \bm{x}')$. The training dynamics, whether we view them via primal SGD on $\bm{w}$ in Reproducing Kernel Hilbert Space (RKHS), or dual variable ascent on $\alpha_i$, are driven by inner products of feature maps: $k(\bm{x},\bm{x}') = \langle \phi(\bm{x}) , \phi(\bm{x}') \rangle_\mathcal{H}$. Missing data leading to masked inputs, i.e., $\bm{x} \mapsto \tilde{\bm{x}} = \bm{M} \bm{x}$, means that the feature maps $\phi(\bm{x})$ now operate on the masked inputs instead, i.e., $\phi(\tilde{\bm{x}})$. This is a crucial differentiation compared to the linear regression case, as the learner does not know which directions in $\mathcal{H}$ have lost excitation. The problem is structurally the same, however, the parameter $\bm{w}$ here lives in a possibly infinite-dimensional space, actuation happens through kernel evaluations, and correlations between directions are baked-in by construction. Given these complications we view the failure mode in $\epsilon-$SVR as actually more dangerous than the linear counterpart: missingness can silently destroy excitation along entire RKHS directions without freezing anything explicitly and overtly. Our defensive posture in Hilbert spaces remains the same on the basis of using observability metrics to throttle the learning dynamics. Replacing the method's fixed step size $\eta$ with the gain $\alpha_t = \eta / \left(\epsilon + k(\tilde{\bm{x}}_t,\tilde{\bm{x}}_t)\right)$ preserves the same observability-aware conditioning principle. A finite-window contraction guarantee, however, can hold only for a finite-dimensional kernel realization or on a finite-dimensional active subspace containing the analyzed error trajectory: a finite sum of rank-one sensitivity operators cannot dominate the identity on a genuinely infinite-dimensional RKHS. We therefore make no claim of full-space finite-window contraction in an infinite-dimensional Hilbert space, as we also state in Section \ref{sec.special.kernel}.

The same observability-aware adaptation principle applies to (deep) neural networks, with suitable measures of observability/excitation that can extend to learned, time-varying representations.

\subsection{Scope and limitations of scalar adaptation}

The prototypic controller introduced here is intentionally myopic. It reacts only to instantaneous observability and applies a uniform gain across all parameter directions. Its role is not to recover missing information, but to enforce a basic stability principle: when only part of the system can be actuated, learning should slow down globally to preserve internal coherence.

In the following sections we generalize this idea by introducing directionally-aware controllers that track excitation history and selectively constrain parameter evolution based on data-supported observability.

\subsection{Summary and motivation for directional awareness}
\begin{itemize}
\item Masking does not destabilize learning by moving unobserved parameters.
\item Instability arises because observed parameters move as if unobserved ones were also considered and moved by the method.
\item A minimal, $d = 2$ numerical walkthrough that makes the partial actuation mechanism concrete is provided in Appendix \ref{sec.app.demo}.
\item The toy defense restores coherence by throttling learning under partial actuation. Yet, isotropic throttling may still miss opportunities to learn in well-observed subsets of a model \textit{and} lacks the discriminating ability to more assertively slow down learning in persistently or increasingly unobserved subsets, in models and architectures of any level of complexity.
\item This control perspective and the pursuit of directional awareness motivate the general framework that we develop next.
\end{itemize}

\section{Learning coherence and stability metrics} \label{sec.metrics}

When we invoke \textit{learning coherence} and \textit{stability}, our main target is preventing parameter oscillation and hallucination at training time, that is, parameter motion that is unjustified by currently available --and intermittently-varying-- data streams. Even though our stability analysis machinery will operate mostly in the parameter error space, it is still useful to introduce a plurality of metrics besides $\| \bm{e}_t \|$ that can illustrate, simply and in a way that is decoupled from our mathematical constructs such as the Gramian-like directional observability matrix $\bm{B}_t$, what we refer to as learning coherence and stability, and help us evaluate the progression of a learning process under missing data. 

\begin{itemize}
    \item \textbf{Update coherence}, evaluated at time $t$ by averaging cosine similarity of the last two update directions (only for pairs of nonzero updates, or after applying a small positive denominator floor):
    \begin{equation}
        \mathrm{Coherence}(t) = \frac{1}{W} \sum_{k=0}^{W-1}\frac{\langle \Delta \bm{w}_{t-k}, \Delta  \bm{w}_{t-k-1} \rangle}{ \| \Delta  \bm{w}_{t-k} \| \| \Delta  \bm{w}_{t-k-1} \|}. \label{eq.metrics.coherence}
    \end{equation}
    One may consider update coherence over an end-to-end-run by averaging $\mathrm{Coherence}(t)$ for $t = 1,\dots, N$, or, as we do in Experiment \ref{sec.exp1}, in an \textit{occupation time sense} (i.e., \textit{how much effort / time does the method spend undoing itself vs.\ moving in a similar direction?}). Our update coherence metric is measuring directional consistency of consecutive updates. Higher is better; in general: positive values imply self-alignment, negative values indicate inconsistent motion. 
    \item \textbf{Trajectory smoothness}, evaluated at time $t$ by looking at second-order differences:
    \begin{equation}
        \mathrm{Smoothness}(t) = \frac{1}{W} \sum_{k=0}^{W-1} \| \bm{w}_{t-k} - 2\bm{w}_{t-k-1} + \bm{w}_{t-k-2} \|.
    \end{equation}
    Our smoothness metric indicates acceleration in the parameter space. Lower values indicate  smoother parameter motion. Under the challenging, pathological conditions induced by missing data, we regard smooth trajectories that do not steer aggressively whenever modalities appear and disappear in the data stream as safer.
    \item \textbf{Final stability}. Even if we are mostly motivated by enduring (in time) learning processes and we scrutinize the ensuing parameter trajectory rather than just the final parameter state, we believe it is still important to consider the stability of the learned parameter late in the process. Assuming a finite learning run of length $N$, we evaluate final stability by looking at the standard deviation of the parameter vector norm in the last $W$ learning iterations:
    \begin{equation}
        \mathrm{Stability} = \mathrm{std} \left( \| \bm{w}_t \| : t \in [N - W, N] \right).
    \end{equation}
    Lower is more stable.
    \item \textbf{Lyapunov contraction rate}. Motivated by the fact that, nominally (and under no missing data conditions), we expect gradient descent to reduce the parameter error vector exponentially, over relatively large temporal windows or in the middle of a finite learning run (to isolate initialization or late-stage oscillation/convergence effects) we fit $\log \left( \| \bm{e}_t \|^2 \right) \sim -\rho t + c$ to estimate the contraction rate $\rho \approx - \frac{d}{dt} \log V_t = -\dot{V}_t/V_t$. Higher $\rho$ implies faster (exponential-like) convergence.
\end{itemize}
The quantities introduced in this section are not used to define or prove stability, but to operationalize the qualitative notions of coherence and (in)stability discussed in Sections \ref{sec.prototypic_scalar} and \ref{sec.directional}, and to make them empirically visible in finite learning runs.

\section{Directional observability-aware learning} \label{sec.directional}
Across linear models, kernels, and gradient networks, the key question that needs to be asked at every update under missing data is:
\begin{equation*}
    \textit{How much control authority is currently available to move the model in directions that matter for the loss?}
\end{equation*}
We argue that, in a model-agnostic and unifying way, this authority is not a property of the parameter space. Rather, it is a property of the Jacobian of the model output with respect to parameters, evaluated at the current input. Concretely, consider a scalar prediction:
\begin{equation}
    \hat{y}_t = f(\bm{w}_t, \tilde{\bm{x}}_t) \label{eq.directional.scalar_prediction}
\end{equation}
where $\bm{w}_t$ are the model parameters and $\tilde{\bm{x}}_t$ is the currently observed --and possibly masked-- input. The local sensitivity of the prediction to parameter motion is given by
\begin{equation}
    \bm{g}_t := \nabla_{\bm{w}} \hat{y}(\bm{w}_t, \tilde{\bm{x}}_t), \label{eq.directional.gt}
\end{equation}
which maps infinitesimal parameter changes $\delta \bm{w}$ to first-order output changes: $\delta \hat{y}_t \approx \langle \bm{g}_t , \delta \bm{w} \rangle$. The sensitivity vector induces a rank-one actuation operator
\begin{equation}
    \bm{A}_t := \bm{g}_t \bm{g}_t^\mathsf{T}, \label{eq.directional.At}
\end{equation}
which captures which directions of the parameter state are instantaneously actuated by the current data. 
Equivalently, writing the sensitivity in Jacobian form,
\begin{equation}
    \bm{J}_t := \frac{\partial f(\bm{w}_t, \tilde{\bm{x}}_t)}{\partial \bm{w}}, \quad \bm{g}_t = \bm{J}_t^\mathsf{T}, \label{eq.directional.Jt}
\end{equation}
we obtain $\bm{A}_t = \bm{J}^\mathsf{T}_t \bm{J}_t$. The operator $\bm{A}_t$ captures both the total energy of the current actuation level, as induced by the masked input and captured by the Jacobian, and, more crucially as we will discuss and further develop next, the directions of varying actuation, as evidenced by current and recent data. Indicatively:
\begin{itemize}
    \item In the case of linear regression, with $f(\bm{w},\bm{x}) = \bm{w}^\mathsf{T} \bm{x}$, the sensitivity is $\bm{g}_t = \tilde{\bm{x}}_t$, so $\bm{J}_t = \tilde{\bm{x}}_t^\mathsf{T}$, and the actuation energy collapses to $\mathrm{tr}(\bm{A}_t) = \|\tilde{\bm{x}}_t \|^2$ whereas directional information is encoded in the entries of $\bm{A}_t$.
    \item For kernel methods with $f(w, \bm{x}) = \langle \bm{w}_t, \bm{\phi}(\tilde{\bm{x}}_t) \rangle_\mathcal{H}$, the sensitivity is $\bm{g}_t = \phi(\tilde{\bm{x}}_t)$, so $\bm{J}_t = \phi(\tilde{\bm{x}}_t)^\mathsf{T}$, and the actuation energy is weighed by $\mathrm{tr}(\bm{A}_t) = \| \phi(\tilde{\bm{x}}_t) \|_\mathcal{H}^2 = k(\tilde{\bm{x}}_t,\tilde{\bm{x}}_t)$, and directional information is similarly encoded in $\bm{A}_t$. We notice that we obtain the same object over a different basis.
    \item For neural networks, transformers, and, in general, a differentiable model parameterized by $\bm{w}_t$, we evaluate $\bm{J}_t = \partial f(\bm{w}_t,\tilde{\bm{x}}_t)/\partial \bm{w}$, and measure the instantaneous actuation energy via  $\| \bm{J}_t \|_F^2$ or even the loss-scaled variant $\| \bm{J}^\mathsf{T}_t \nabla_{\hat{y}} \ell_t \|^2$ (weighing how much actuation authority the learner has along the signal that the loss function provides for parameter motion).
\end{itemize}

To preserve stability and coherence, our thesis is that learning across diverse bases and models should \textit{adapt} to the \textit{time}- and \textit{directionally}-\textit{varying controllability of the parameter dynamics}, provably \textit{and} empirically, as \textit{measured} through the gradient channel induced by the (intermittently) observed data.

\subsection{Motivation: Why scalar throttling is insufficient}
In Section \ref{sec.prototypic_scalar}, we restored learning coherence and stability by modulating the learning rate using a scalar measure of actuation, effectively treating the gradient channel as isotropic. In the notation introduced above, this corresponds to collapsing the instantaneous actuation operator $\bm{A}_t = \bm{g}_t \bm{g}_t^\mathsf{T}$ into the single scalar quantity $\mathrm{tr}(\bm{A}_t) = \| \bm{g}_t \|^2$, and regulating updates uniformly across all parameter directions. Such a reduction is appropriate when loss of observability affects all directions similarly. In such cases, the rank-one structure of $\bm{A}_t$ carries little additional information beyond its overall magnitude, and scalar throttling suffices to prevent incoherent motion. 

In general, however, the actuation induced by the data stream can be directionally heterogeneous. Missing inputs, partial outputs (i.e., beyond the scalar prediction case prototyped here), and masked modalities (e.g., in large models learning on text and pictures, where part or the entirety of either modality may be unavailable) suppress excitation along specific parameter directions while leaving others well actuated. This anisotropy is already encoded in $\bm{A}_t$: although $\bm{A}_t$ is rank-one at each instant, its orientation varies over time, and different directions may be persistently excited or persistently silent through iterations. Moreover, this anisotropy can be amplified in more complex, nonlinear models --from kernel methods to neural networks-- as individual input coordinates or features propagate through shared representations and influence broad subsets of parameters. As a result, excitation patterns that are localized or intermittent in input space can induce highly anisotropic actuation in parameter space, further widening the gap between well-supported and weakly supported directions.

A scalar gain discards such directional information by design. As a result, it cannot distinguish between directions that remain reliably actuated (where learning could safely proceed) and directions that are effectively blind (where motion risks being unjustified). Scalar throttling is therefore overly conservative in some directions and insufficiently discriminating in others. We posit that to preserve both learning progress and coherence under partial observability, the adaptation mechanism must be sensitive not only to the amount of actuation present (as is the scalar adaptation approach) but also to how that actuation is distributed across the parameter space. This motivates retaining and building on the directional structure of the actuation operator and aggregating it over time, leading naturally to the observability matrix $\bm{B}_t$ that we introduce next and the associated directionally-aware adaptive update law.

\subsection{Directional excitation and observability history}
The rank-one operator $\bm{A}_t$ is, by definition, instantaneous and, as such, relatively volatile. To obtain a smoother, more informative measure of directional excitation, we maintain the unregularized exponential moving average state $\bm{S}_t$ and form the regularized observability matrix $\bm{B}_t$ according to:
\begin{equation}
    \begin{aligned}
        \bm{S}_t &:= \beta \bm{S}_{t-1} + (1-\beta) \bm{A}_t, \qquad \bm{S}_{-1}=0,\\
        \bm{B}_t &:= \bm{S}_t + \epsilon I, \qquad \beta \in (0,1),\quad \epsilon>0.
    \end{aligned}
    \label{eq.directional.B_recursive}
\end{equation}
The state $\bm{S}_t$ encodes which parameter directions have recently been actuated by the data stream, whereas the ridge term in $\bm{B}_t$ provides a fixed conditioning floor without recursively accumulating an isotropic background. By unrolling the recursion, we obtain
\begin{equation}
    \bm{B}_t = \epsilon I + (1-\beta)\sum_{k=0}^{t}\beta^k\bm{g}_{t-k}\bm{g}_{t-k}^\mathsf{T}. \label{eq.directional.B_unrolled}
\end{equation}
The matrix $\bm{B}_t$ is reminiscent of the concept of a finite-horizon or discounted Gramian in control. In analogy with observability Gramians in dynamical systems, $\bm{B}_t$ here quantifies which directions of the parameter state have been excited by the data stream under partial observability as we look at the learning dynamics from a dynamical and control systems viewpoint.

\subsection{Directionally-aware update law} \label{sec.directional.update_law}
Let $\ell_t(\hat{\bm{y}}_t , \bm{y}_t)$ denote a scalar loss. For any differentiable scalar-output model, the parameter gradient can be written as:
\begin{align}
    \nabla_{\bm{w}} \ell_t (\hat{y}_t(\bm{w}_t, \tilde{\bm{x}}_t) , y_t) =\frac{\partial \ell_t (\hat{y}_t(\bm{w}_t, \tilde{\bm{x}}_t) , y_t) }{\partial \hat{y}_t} \underbrace{\nabla_{\bm{w}} \hat{y}_t(\bm{w}_t, \tilde{\bm{x}}_t)}_{=:\bm{g}_t} = \bm{g}_t  \frac{\partial \ell_t (\hat{y}_t(\bm{w}_t, \tilde{\bm{x}}_t) , y_t) }{\partial \hat{y}_t} \label{eq.Beta_EMA_dynamics}
\end{align}
so the gradient lies entirely in the span of sensitivity directions induced by the currently observed data. 

We propose the following state-dependent adaptive control law:
\begin{equation}
    \boxed{\bm{w}_{t+1} = \bm{w}_t - \eta \bm{B}_t^{-1} \nabla_{\bm{w}} \ell_t}  \label{eq.adaptive_directional_control_law}
\end{equation}
where $\bm{B}_t$ evolves according to (\ref{eq.directional.B_recursive}).

At first glance, the role of $\bm{B}_t^{-1}$ may appear counterintuitive: directions that are well observed correspond to large eigenvalues of $\bm{B}_t$, which become small when inverting to obtain $\bm{B}_t^{-1}$. If the gradient magnitude scaled cleanly with observability, this inversion would indeed appear redundant as its effect on the gain would be undone by terms in $\nabla_{\bm{w}} \ell_t$. However, there's much more going on in the adaptive control law (\ref{eq.adaptive_directional_control_law}) that yields the sought-after directional observability-aware adaptation.

\subsubsection{Gradient structure under partial observability}

Under partial observability, gradients are computed from incomplete information. As a result, the instantaneous gradient $\nabla_{\bm{w}} \ell_t$ should be treated as an uncertain measurement conditioned on what the learner observes (without making at this point any assumption on the underlying distribution). Let $\mathcal{I}_t$ denote the information available through the observed input-output pair at time $t$. We may write, by definition of conditional expectation, $\nabla_{\bm{w}} \ell_t = \mathbb{E} \left[ \nabla_{\bm{w}} \ell_t \vert \mathcal{I}_t \right] + \bm{\xi}_t$. Solving for $\bm{\xi}_t$ yields $\bm{\xi}_t := \nabla_{\bm{w}} \ell_t - \mathbb{E} \left[ \nabla_{\bm{w}} \ell_t \vert \mathcal{I}_t \right]$. By construction, $\mathbb{E} \left[ \bm{\xi}_t \vert \mathcal{I}_t \right] = 0$, though the variance of $\bm{\xi}_t$ may be large, time-varying, and structured under partial observability. No assumption is made that the gradients are unbiased with respect to any latent ground truth; the instantaneous decomposition separates what is predictable from observed data from what is not.

\subsubsection{Directional scaling of gradient uncertainty \& effect of the inverse observability metric} \label{sec.directional.scaling_inv} 

Let $\bm{v}_i$ be an eigenvector of the empirical EMA state $\bm{S}_t$ with eigenvalue $\lambda_i \geq 0$. The same vector is an eigenvector of $\bm{B}_t = \bm{S}_t+\epsilon I$ with eigenvalue $\lambda_i+\epsilon$. Define the scalar gradient component
\begin{equation*}
    G_{i,t} := \langle \nabla_{\bm{w}} \ell_t , \bm{v}_i \rangle = \langle \bm{g}_t , \bm{v}_i \rangle \frac{\partial \ell_t}{\partial \hat{y}_t}. \label{eq.directional.G_it}
\end{equation*}
For $\lambda_i>0$, we introduce the normalized excitation $\zeta_{i,t} := \langle \bm{g}_t, \bm{v}_i \rangle /\sqrt{\lambda_i}$, which allows us to write $\langle \bm{g}_t , \bm{v}_i \rangle = \sqrt{\lambda_i}\zeta_{i,t}$. Recall that the direction indicated by $\bm{v}_i$ and the corresponding empirical eigenvalue $\lambda_i$ encode recent, evidence-driven information on directional observability and strength. Therefore, our normalization expresses instantaneous excitation along $\bm{v}_i$ in units of its recent empirical scale. When $\lambda_i=0$, the exact EMA carries no recent excitation in that direction and the corresponding instantaneous component is also zero. The update component along direction $\bm{v}_i$ is
\begin{equation}
    \Delta w_i = -\eta \frac{1}{\lambda_i + \epsilon} G_{i,t} = -\eta \frac{\sqrt{\lambda_i}}{\lambda_i + \epsilon} \zeta_{i,t} \frac{\partial \ell_t}{\partial \hat{y}_t}. \label{eq.directional.component_update}
\end{equation}
This expression makes the conditioning mechanism explicit:
\begin{itemize}
    \item in strongly and repeatedly excited directions $(\lambda_i \gg \epsilon)$, the update is normalized by the recent excitation scale, preventing overly aggressive repeated motion;
    \item as empirical excitation vanishes $(\lambda_i \to 0)$, the sensitivity component itself vanishes, and so does the corresponding update, while the ridge $\epsilon I$ prevents numerical singularity.
\end{itemize}
Crucially, the inverse does not act on gradient magnitude alone. It regulates parameter motion through the interaction between current sensitivity and its recent empirical scale, ensuring that learning proceeds through directions that remain supported by the observed data.

\subsubsection{Geometric interpretation}
The update law admits a natural variational interpretation. At each iteration, the learner seeks a parameter increment $\delta \bm{w}$ that reduces the loss while regulating motion according to the recent actuation geometry. This trade-off can be expressed by the optimization problem:
\begin{equation}
    \delta \bm{w}_t = \arg \min_{\delta \bm{w}} \left( \langle \nabla_{\bm{w}} \ell_t, \delta \bm{w} \rangle + \frac{1}{2\eta} \langle \delta \bm{w}, \bm{B}_t \delta \bm{w} \rangle \right), \label{eq.directional.update_geometric}
\end{equation}
where the first term in the cost corresponds to linearized loss reduction, whereas the second term measures parameter motion in the empirical observability geometry. Solving this problem yields:
\begin{equation}
    \delta \bm{w}_t = -\eta \bm{B}_t^{-1} \nabla_{\bm{w}} \ell_t, \label{eq.optim_cost}
\end{equation}
which coincides with the proposed update law. The matrix $\bm{B}_t$ therefore defines a time-varying geometry on parameter space: it induces an inner product $\langle \bm{u}, \bm{v} \rangle_{\bm{B}_t} =  \langle \bm{u}, \bm{B}_t \bm{v} \rangle$ and a corresponding norm $\| \bm{u} \|^2_{\bm{B}_t} =  \langle \bm{u}, \bm{B}_t \bm{u} \rangle$. The quadratic term alone assigns larger cost to directions with larger empirical eigenvalues; the resulting update, however, is determined jointly by this geometry and by a gradient whose components are themselves generated by the current sensitivity. Accordingly, the controller should be interpreted as normalizing current actuation relative to its recent directional scale, rather than as imposing a standalone hard penalty on weakly observed coordinates. The resulting learning dynamics reshape parameter evolution under partial actuation without introducing externally scripted constraints or heuristic scheduling.

\subsubsection{Scalar adaptive throttling as a special case}
The directionally-aware adaptive law introduced in this section generalizes the scalar observability-aware throttling mechanism introduced in Section \ref{sec.prototypic_scalar} and prototyped for the case of linear regression. To see this, let us approximate the observability matrix as $\bm{B}_t \approx (\epsilon + \gamma_t) I$, where  $\gamma_t := \| \bm{g}_t \|^2$, that is, ignoring directional structure and retaining only a measure of the instantaneous actuation energy. Substituting this approximation into the update law (\ref{eq.adaptive_directional_control_law}) yields:
\begin{equation*}
    \bm{w}_{t+1} = \bm{w}_t - \frac{\eta}{\epsilon + \gamma_t} \nabla_{\bm{w}} \ell_t,
\end{equation*}
which coincides with the scalar adaptive throttling rule (\ref{eq.linear_regression_adaptive_step}) when $\bm{g}_t = \tilde{\bm{x}}_t$ (as is the case for the linear regression example used in our prototyping in Section \ref{sec.prototypic_scalar}). Therefore, scalar throttling can be interpreted as an isotropic approximation of the full directional controller developed in the present section, in which all parameter directions are treated as equally observable. The directional formulation developed here recovers the same stability-preserving behavior while allowing the learner to explore anisotropy in the excitation induced by the data. In this sense, the proposed framework does not replace scalar observability-aware adaptation, but embeds it as a limiting case, and significantly builds on its premise to specialize direction-specific adaptation based on empirically-observed evidence in the data stream.

\subsection{Interpretation as closed-loop control}
From a control-theoretic perspective, the learning process defined by (\ref{eq.adaptive_directional_control_law}) constitutes a closed-loop adaptive system in which the parameter vector $\bm{w}_t$ assumes the role of the system state, the gradient $\nabla_{\bm{w}} \ell_t$ provides feedback, and the matrix $\bm{B}_t^{-1}$ defines a state- and data-dependent control law. It is important to note that $\bm{B}_t^{-1}$ is not scheduled externally, nor scripted in any way, but, rather, evolves with the learning process, driven by the method's understanding of time-varying directional observability as empirically seen in the data stream. Loss in said observability therefore enters the dynamics as an actionable reduction in actuation authority, rather than as stochastic noise.

This perspective allows us to draw a fundamental distinction between standard gradient descent and our proposed method. Standard SGD applies control inputs designed for a fully actuated system, even when only a subset of the parameter directions can be reliably influenced by the data. The resulting dynamics are effectively open-loop under partial observability. In contrast, our update law (\ref{eq.adaptive_directional_control_law}) continuously reshapes the control action to match the current actuation structure, as influenced by intermittent feature observability, ensuring that parameter motion remains coherent with what the data can actually support. 

From a broader perspective, our directional observability-aware learning enforces a basic control principle: when actuation is lost along some directions, the control law must adapt to avoid destabilizing the remaining degrees of freedom. The attenuation mechanism derived in Section \ref{sec.directional.scaling_inv} is precisely the manifestation of this principle at the level of parameter updates.

\subsection{Practical approximations} \label{sec.practical_approximations}
Maintaining and inverting the full observability matrix $\bm{B}_t \in \mathbb{R}^{d \times d}$ may be impractical for large-scale models. Nevertheless, our framework admits several approximations that preserve its core stabilizing behavior. 

\begin{itemize}
    \item \textbf{Diagonal approximation:} A simple and scalable option is to retain only the diagonal values of $\bm{B}_t$, leading to elementwise, decoupled adaptation. While this approximation discards cross-coordinate coupling, it preserves the essential conditioning mechanism by scaling each coordinate of the current sensitivity against its recent squared-sensitivity history. Such an update can be reminiscent of RMSProp or Adam, but carries a fundamentally different interpretation: the diagonal entries estimate coordinatewise actuation induced by the model sensitivity, whereas RMSProp and Adam accumulate uncentered second moments of loss gradients. 
    \item \textbf{Block-structured approximation:} In many modern architectures, parameters are naturally organized into groups whose effects on the model output are strongly coupled, even if the corresponding units do not interact directly in the forward computation. Typical examples include weight matrices associated with individual layers in multilayer perceptrons or attention heads in transformer architectures. Within such groups, parameter perturbations influence the output through shared intermediate representations and common downstream signals. As a result, the corresponding columns of the output-parameter Jacobian are repeatedly excited together by the data, leading to strong intra-group coupling in the observability matrix $\bm{B}_t$. This structure can be exploited by approximating $\bm{B}_t$ with a block-diagonal matrix, where each block corresponds to a predefined parameter group. Each block then tracks the aggregate observability of its associated component, allowing the adaptive law to regulate learning at the level of layers or heads rather than individual parameters. Cross-group interactions, which are typically weaker in the learning geometry, are neglected for tractability.
    \item \textbf{Low-rank approximation:} In settings where excitation concentrates along a small number of dominant directions, $\bm{B}_t$ may be approximated using low-rank factorizations. This approach can preserve the most informative excitation modes while treating the complementary directions through the ridge term. Any finite-window contraction guarantee then applies on the retained finite-dimensional subspace under its corresponding effective-excitation condition; behavior in the complementary subspace requires a separate approximation or regularization argument.
\end{itemize}
Across such approximations, the overarching conditioning principle remains unchanged: parameter motion is reshaped according to empirically observed actuation without reconstructing missing data. Their exact stability guarantees depend on the corresponding block/subspace excitation and coupling properties; empirically, we evaluate whether these scalable realizations retain the coherence and excursion-suppression behavior motivated by the ideal controller.

\subsection{Takeaway}
Directional observability determines the uncertainty scale of gradient information. The inverse observability metric $\bm{B}_t^{-1}$ is therefore not a heuristic rescaling, but a principled conditioning mechanism that suppresses noise-dominated parameter motion, while preserving learning where the data provide stable excitation.

This state-dependent adaptive control law maintains coherence of the learning dynamics under partial and intermittent observability, and directly enables the stability results developed in Section \ref{sec.stability}.

\section{Stability analyses} \label{sec.stability}

We analyze the learning dynamics as a time-varying, closed-loop system governed by the directionally-aware adaptive law  introduced in Section \ref{sec.directional}, in which missing data induces a structured loss of actuation over the parameter space. Our objective is to show boundedness and controlled contraction of the parameter error dynamics under such intermittent and possibly anisotropic observability. By wrapping around whatever model a masked SGD method was anyhow training, our resulting guarantees are largely model-agnostic, and are later instantiated explicitly for linear models, kernel methods, and gradient-based neural network cases. 

\subsection{General closed-loop error dynamics} \label{sec.stability.dynamics}
Under the directionally aware adaptive update law, the parameter evolves according to:
\begin{equation}
    \bm{w}_{t+1} = \bm{w}_t - \eta \bm{B}_t^{-1} \nabla_{\bm{w}} \ell_t, \label{eq.analysis.w}
\end{equation}
in parallel with the observability matrix $\bm{B}_t$, updated, in turn, recursively following
\begin{equation}
    \begin{aligned}
        \bm{S}_t &= \beta \bm{S}_{t-1} + (1-\beta)\bm{A}_t, \qquad \bm{A}_t := \bm{g}_t\bm{g}_t^\mathsf{T},\\
        \bm{B}_t &= \bm{S}_t+\epsilon I.
    \end{aligned}
    \label{eq.analysis.B}
\end{equation}
where $\bm{g}_t = \nabla_{\bm{w}} \hat{y}_t (\bm{w}_t, \tilde{\bm{x}}_t)$ denotes the instantaneous sensitivity induced by the observed (and possibly masked) input. Let $\bm{w}^*$ denote a reference parameter (e.g., a minimizer of the fully-observed risk, when it exists), and define the parameter error
\begin{equation}
    \bm{e}_t := \bm{w}_t - \bm{w}^*. \label{eq.analysis.error}
\end{equation}
Under partial observability, gradients are computed through a reduced observation operator and, therefore, reflect only the information available at time $t$. As discussed in Section \ref{sec.directional.update_law}, we may write the instantaneous gradient as
\begin{equation}
    \nabla_{\bm{w}} \ell_t = \mathbb{E} \left[  \nabla_{\bm{w}} \ell_t \mid \mathcal{I}_t \right] + \bm{\xi}_t, \label{eq.analysis_gradient_with_ksi}
\end{equation}
where $\mathcal{I}_t$ denotes the information revealed by the observed data and $\bm{\xi}_t$ captures the component that is unpredictable from that information. We use this decomposition only to illuminate how partial observations affect the feedback channel; the deterministic stability result below does not assume that $\bm{\xi}_t$ is independent, bounded, or exogenous, and is formulated instead through a scalar closed-loop residual mismatch.

By leveraging (\ref{eq.analysis.w}) and (\ref{eq.analysis_gradient_with_ksi}), we obtain an illuminating form of the closed loop error dynamics:
\begin{equation}
    \bm{e}_{t+1} = \bm{e}_t - \eta \bm{B}_t^{-1} \mathbb{E} \left[ \nabla_{\bm{w}} \ell_t \mid \mathcal{I}_t \right] - \eta \bm{B}_t^{-1} \bm{\xi}_t. \label{eq.analysis.error_dynamics_full}
\end{equation}
This expression illuminates fundamental components:
\begin{enumerate}
    \item A controlled descent term, shaped by the inverse observability metric $\bm{B}_t^{-1}$, which governs motion along empirically supported parameter directions.
    \item An unpredictable-feedback term, reshaped directionally by $\bm{B}_t^{-1}$, representing the time-varying and anisotropic effect of information not resolved by the observed data.
\end{enumerate}
It is critical to note that missing data does not merely increase stochasticity, but, rather, induces a directionally structured loss of actuation reflected both in the conditional gradient and in the evolution of $\bm{B}_t$. Standard SGD applies its updates blindly, as if this loss of actuation were absent. In contrast, our present formulation explicitly closes the loop around the evolving actuation geometry, with the latter continuously shaped by the data stream. 

\subsection{Residual-to-state boundedness under intermittent observability}

As the model output $y$ is scalar, the gradient of the loss always factors as
\begin{equation}
    \nabla_{\bm{w}} \ell_t = \frac{\partial \ell_t \left( \hat{y}_t , y_t \right)}{\partial \hat{y}_t} \nabla_{\bm{w}} \hat{y} \left( \bm{w}_t , \tilde{\bm{x}}_t \right) = s_t \bm{g}_t \label{eq.lossgrad_factorization}
\end{equation}
where
\begin{equation}
    s_t := \frac{\partial \ell_t \left( \hat{y}_t, y_t \right)}{\partial \hat{y}_t} \qquad \text{and} \qquad \bm{g}_t := \nabla_{\bm{w}} \hat{y}_t \left( \bm{w}_t, \tilde{\bm{x}}_t \right).
\end{equation}
The learning dynamics are therefore driven by a scalar feedback signal $s_t$, acting through a directionally structured actuation vector $\bm{g}_t$. Missing data affects learning through both means: it alters the available actuation directions $\bm{g}_t$, and corrupts the scalar feedback $s_t$ by injecting unobservable contributions into the prediction error. Substituting factorization (\ref{eq.lossgrad_factorization}) into the update law (\ref{eq.analysis.w}) yields the closed-loop dynamics:
\begin{equation}
    \bm{e}_{t+1} = \bm{e}_t - \eta \bm{B}_t^{-1} \bm{g}_t s_t. \label{eq.closed_loop_error_dynamics}
\end{equation}
This representation makes explicit that learning proceeds through a single scalar feedback loop whose effect on the parameter state is shaped by the observability-aware metric $\bm{B}_t^{-1}$.

Next, we introduce assumptions that will enable our analysis of the closed loop error dynamics (\ref{eq.closed_loop_error_dynamics}) under the update rule (\ref{eq.analysis.w})-(\ref{eq.analysis.B}).

\subsection{Effective sensitivity and prediction error representation}

To relate scalar prediction errors to the parameter error $\bm{e}_t$ under the preconditioned learning dynamics, we distinguish between the residual seen by the loss and the effective error signal seen through the update geometry. Define
\begin{equation}
    \bm{\psi}_t := \bm{B}_t^{-1}\bm{g}_t, \qquad \chi_t := \bm{\psi}_t^\mathsf{T}\bm{e}_t = \bm{g}_t^\mathsf{T}\bm{B}_t^{-1}\bm{e}_t. \label{eq.analysis.effective_sensitivity}
\end{equation}
The vector $\bm{\psi}_t$ is the effective actuation direction through which the parameter state actually moves, whereas $\chi_t$ is the associated update-geometry error signal.

First, for bookkeeping reasons, we bound the instantaneous model sensitivity under masked input.
\begin{assumption} \label{assumption_G} \textbf{(bounded instantaneous sensitivity)}
    There exists $G>0$ such that $\| \bm{g}_t \| \leq G$ for all $t$.
\end{assumption}

\begin{assumption} \label{assumption.sectors} \textbf{(sector-bounded scalar feedback)}
There exists a function $\phi : \mathbb{R} \to \mathbb{R}$ such that
\begin{equation}
    s_t = \phi \left( \tilde{r}_t \right),
\end{equation}
where $\tilde{r}_t := \hat{y}_t\left( \bm{w}_t, \tilde{\bm{x}}_t \right) - y_t$ is the residual seen by the loss, and
\begin{equation}
    m z^2 \leq z \phi(z) \leq M z^2, \qquad\vert \phi(z)\vert \leq M \vert z \vert, \label{eq.assumption.sector_bounds}
\end{equation}
for some constants $m, M > 0$ and all $z \in \mathbb{R}$.
\end{assumption}
For squared loss, $\phi(z) = z$, and $m = M = 1$.

\begin{assumption} \label{assumption_d} \textbf{(a priori bounded closed-loop residual mismatch)}
    The residual seen by the loss decomposes as
    \begin{equation}
        \tilde{r}_t = \chi_t + \delta_t, \label{eq.assumption.residual_decomposition}
    \end{equation}
    where $\delta_t$ satisfies
    \begin{equation}
        \vert \delta_t \vert \leq \overline{\delta} < \infty. \label{eq.assumption.disturbance_bound}
    \end{equation}
\end{assumption}
The boundedness in (\ref{eq.assumption.disturbance_bound}) captures both the contribution of information unavailable to the learner and the mismatch between the finite prediction residual and the local, preconditioned update geometry. Unlike a classical external ISS input, $\delta_t$ need not be independent of $\bm{e}_t$, and it need not vanish when every input component is observed. In the unpreconditioned linear case, it reduces, up to the sign induced by our residual convention, to the unobserved contribution isolated in Section \ref{sec.prototypic_scalar}. The result developed below is therefore a residual-to-state, ISS-type bound with respect to this bounded closed-loop mismatch. It specializes to conventional input-to-state stability only when $\delta_t$ is supplied or modeled as an independently specified exogenous input.

The decomposition in (\ref{eq.assumption.residual_decomposition}) is an identity; the substantive contribution of Assumption \ref{assumption_d} is the existence of the uniform bound $\overline{\delta}$. That bound must be certified independently of Proposition \ref{proposition.ISS}, for example through known output, target, and parameter-domain bounds; projection or saturation; an invariant operating region established by a separate argument; or a direct model-specific mismatch estimate. If only a state-dependent estimate such as $|\delta_t|\leq c_\delta\|\bm{e}_t\|+\overline{d}$ is available, Proposition \ref{proposition.ISS} does not by itself close that dependence; an additional small-gain argument is required.

\begin{assumption} \textbf{(raw windowed/eventual observability)} \label{assumption_window}
    There exists a window length $T \in \mathbb{Z}$, with $T \geq 1$, and $\lambda_g >0$, such that for all $t \geq T-1$:
    \begin{equation}
        \sum_{k = 0}^{T-1} {\bm{g}}_{t-k} {\bm{g}}_{t-k}^\mathsf{T} \succeq \lambda_g I. \label{eq.assumption_windowed_obs}
    \end{equation}
\end{assumption}
Assumption \ref{assumption_window} enables excitation to be intermittent, provided it recurs sufficiently over finite windows.

\begin{lemma} \label{lemma.Bt_eigenvalue_bounds} \textbf{(eigenvalue bounds for the discounted observability matrix)}
Under Assumption \ref{assumption_G}, the matrix $\bm{B}_t$ satisfies for all $t$:
\begin{equation}
    \underline{b} I \preceq \bm{B}_t \preceq \overline{b} I, \qquad \underline{b}:=\epsilon,\qquad \overline{b}:=\epsilon+G^2. \label{eq.B_uniform_bounds}
\end{equation}
Moreover, under Assumption \ref{assumption_window}, for all $t\geq T-1$ the lower bound can be sharpened to
\begin{equation}
    \bm{B}_t \succeq \left(\epsilon+(1-\beta)\beta^{T-1}\lambda_g\right)I. \label{eq.B_sharper_lower}
\end{equation}
Consequently,
\begin{equation}
    \frac{1}{\overline{b}}I \preceq \bm{B}_t^{-1} \preceq \frac{1}{\underline{b}}I,\qquad \|\bm{\psi}_t\|\leq \frac{G}{\underline{b}}. \label{eq.P_uniform_bounds}
\end{equation}
\end{lemma}
\begin{proof}
    Refer to Appendix \ref{app.proof.lemma_Bt_eigenvaluePbounds}.
\end{proof}
For global statements, $\underline{b}=\epsilon$ is always admissible. When the analysis is initialized after one complete observability window, the sharper lower bound in (\ref{eq.B_sharper_lower}) may be substituted to obtain less conservative constants.

The raw sensitivity directions in Assumption \ref{assumption_window} are not, in general, transformed by a common matrix over a full window. We therefore state separately the excitation property that the stability proof requires in the actual update geometry.

\begin{assumption} \textbf{(windowed observability in the update geometry)} \label{assumption_effective_window}
    There exists $\mu>0$ such that for all $t\geq T-1$:
    \begin{equation}
        \sum_{k=0}^{T-1}\bm{\psi}_{t-k}\bm{\psi}_{t-k}^\mathsf{T}
        =
        \sum_{k=0}^{T-1}\bm{B}_{t-k}^{-1}\bm{g}_{t-k}\bm{g}_{t-k}^\mathsf{T}\bm{B}_{t-k}^{-1}
        \succeq \mu I. \label{eq.assumption_effective_window}
    \end{equation}
\end{assumption}
Assumption \ref{assumption_effective_window} asserts that the directions through which the parameter state actually moves continue to span the analyzed parameter space over finite windows. Because each summand in (\ref{eq.assumption_effective_window}) has rank one, the left-hand side has rank at most $T$. Consequently, the full-space conditions in Assumptions \ref{assumption_window} and \ref{assumption_effective_window} necessarily require $T\geq d$ in $\mathbb{R}^d$, and no finite $T$ can make either condition hold against the identity on a genuinely infinite-dimensional Hilbert space. Proposition \ref{proposition.ISS} is therefore a full-state result for finite-dimensional parameter spaces, or for a finite-dimensional invariant active subspace $\mathcal{U}$ containing the analyzed error trajectory, with all operators and inequalities restricted to $\mathcal{U}$. In that case, the condition may equivalently be written as
\begin{equation}
    \sum_{k=0}^{T-1}\bm{\psi}_{t-k}\bm{\psi}_{t-k}^{\mathsf{T}}\succeq\mu\Pi_{\mathcal{U}},
\end{equation}
together with invariance of the error dynamics and $\bm{e}_0\in\mathcal{U}$. Without such invariance, a projected excitation inequality alone does not close the projected error dynamics, and Proposition \ref{proposition.ISS} makes no claim for either the projected or orthogonal-complement trajectories without additional coupling bounds. The condition can be verified directly. The following result also provides a sufficient bridge from raw excitation when the inverse metric evolves coherently over the same window.

\begin{lemma} \label{lemma.raw_to_effective_excitation} \textbf{(transfer of excitation through a slowly varying metric)}
Let $\bm{P}_t:=\bm{B}_t^{-1}$ and define the inverse-metric variation budget
\begin{equation}
    \nu_T := \sup_{t\geq T-1}
    \left(
        \sum_{k=0}^{T-1}\|\bm{P}_{t-k}-\bm{P}_t\|^2
    \right)^{1/2}. \label{eq.metric_variation_budget}
\end{equation}
Under Assumptions \ref{assumption_G} and \ref{assumption_window}, and the uniform bounds of Lemma \ref{lemma.Bt_eigenvalue_bounds}, if
\begin{equation}
    G\nu_T < \frac{\sqrt{\lambda_g}}{\overline{b}}, \label{eq.metric_variation_condition}
\end{equation}
then Assumption \ref{assumption_effective_window} holds with
\begin{equation}
    \mu = \left(\frac{\sqrt{\lambda_g}}{\overline{b}}-G\nu_T\right)^2>0. \label{eq.mu_transfer}
\end{equation}
\end{lemma}
\begin{proof}
    Refer to Appendix \ref{app.proof.raw_to_effective_excitation}.
\end{proof}

The metric-coherence condition above is naturally supported by the EMA construction. Indeed, the resolvent identity and recursion (\ref{eq.directional.B_recursive}) imply
\begin{equation}
    \|\bm{P}_{j+1}-\bm{P}_j\|
    \leq \frac{2(1-\beta)G^2}{\underline{b}^2}, \label{eq.metric_increment_bound}
\end{equation}
and therefore
\begin{equation}
    \nu_T
    \leq
    \frac{2(1-\beta)G^2}{\underline{b}^2}
    \sqrt{\frac{T(T-1)(2T-1)}{6}}. \label{eq.metric_variation_ema_bound}
\end{equation}
This sufficient bound is conservative, but it makes the role of $\beta$ explicit: raw observability is transferred to the effective update geometry when the preconditioner does not deform too violently over the same window.

\begin{lemma} \label{lemma.dissipation} \textbf{(one-step dissipation and bounded motion)}
Let $V_t := \| \bm{e}_t \|^2$. Under Assumptions \ref{assumption_G}-\ref{assumption_d}, if
\begin{equation}
    0 < \eta \leq \frac{m \underline{b}^2}{2 M^2 G^2}, \label{eq.Lyapunov_lemma_step_size}
\end{equation}
then, for all $t$,
\begin{align}
    V_{t+1}
    &\leq V_t-\alpha_0\tilde{r}_t^2+\gamma_0\delta_t^2, \label{eq.one_step_Lyapunov_decr}\\
    V_{t+1}
    &\leq V_t-\alpha\chi_t^2+\gamma \delta_t^2, \label{eq.one_step_effective_decr}\\
    \|\bm{e}_{t+1}-\bm{e}_t\|
    &\leq \eta\frac{GM}{\underline{b}}|\tilde{r}_t|, \label{eq.one_step_motion_bound}
\end{align}
where
\begin{equation}
    \alpha_0:=\frac{\eta m}{2},\qquad
    \gamma_0:=\frac{\eta M^2}{m},\qquad
    \alpha:=\frac{\alpha_0}{2}=\frac{\eta m}{4},\qquad
    \gamma:=\alpha_0+\gamma_0.
    \label{eq.dissipation_constants}
\end{equation}
\end{lemma}
\begin{proof}
Refer to Appendix \ref{app.proof.one_step_dissipation}.
\end{proof}

\begin{lemma} \textbf{(windowed coercivity in the update geometry)} \label{lemma.windowed_obs_in_update_geometry}
Suppose Assumption \ref{assumption_effective_window} and the conclusions of Lemma \ref{lemma.dissipation} hold. Then, for all $t\geq T-1$,
\begin{equation}
    \sum_{k=0}^{T-1}\chi_{t-k}^2
    \geq
    \theta V_{t-T+1}-\zeta\overline{\delta}^2, \label{eq.windowed_coercivity}
\end{equation}
where an admissible pair is
\begin{align}
    \theta
    &:=
    \frac{\mu}{2}
    -
    \frac{\eta^2G^4M^2T^2}{\alpha_0\underline{b}^4}
    =
    \frac{\mu}{2}
    -
    \frac{2\eta G^4M^2T^2}{m\underline{b}^4},
    \label{eq.theta_effective}\\
    \zeta
    &:=
    \frac{\eta^2G^4M^2\gamma_0T^3}{\alpha_0\underline{b}^4}
    =
    \frac{2\eta^2G^4M^4T^3}{m^2\underline{b}^4}.
    \label{eq.zeta_effective}
\end{align}
In particular, $\theta>0$ whenever
\begin{equation}
    \eta < \frac{\mu m\underline{b}^4}{4G^4M^2T^2}. \label{eq.lemma_windowed_step_size_cond}
\end{equation}
\end{lemma}
\begin{proof}
Refer to Appendix \ref{app.proof.update_geometry_lemma}.
\end{proof}

\begin{proposition} \label{proposition.ISS} \textbf{(windowed residual-to-state bound under effective eventual observability)}
Suppose Assumptions \ref{assumption_G}-\ref{assumption_d} and \ref{assumption_effective_window} hold. Alternatively, Assumption \ref{assumption_effective_window} is ensured by Assumption \ref{assumption_window} and condition (\ref{eq.metric_variation_condition}), via Lemma \ref{lemma.raw_to_effective_excitation}. If
\begin{equation}
    0<\eta<
    \min\left\{
        \frac{m\underline{b}^2}{2M^2G^2},
        \frac{\mu m\underline{b}^4}{4G^4M^2T^2},
        \frac{8}{m\mu}
    \right\}, \label{eq.proposition_step_size}
\end{equation}
then there exist constants $q\in(0,1)$ and $A,L>0$ such that, for all $t\geq0$,
\begin{equation}
    \|\bm{e}_t\|
    \leq
    A q^{\lfloor t/T\rfloor}\|\bm{e}_0\|
    +
    L\overline{\delta}. \label{eq.proposition_ISS_bound}
\end{equation}
One explicit admissible choice is
\begin{equation}
    q=\sqrt{1-\alpha\theta},
\end{equation}
with $A$ and $L$ as provided in the proof.
\end{proposition}
\begin{proof}
Refer to Appendix \ref{app.proof.ISS}.
\end{proof}

The displayed constants are \textit{sufficient} and generally \textit{conservative}. Their purpose is to certify the contraction mechanism and expose its directional dependencies, rather than provide numerically tight trajectory estimates or literal tuning rules. In Section \ref{section.experiments} we evaluate the practical behavior of the resulting adaptation outside this worst-case analytical envelope.

\subsection{Special cases}
Equipped with the full form and analysis of our directionally-aware learning adaptation mechanism, we explore next how it specializes in linear regression, kernel methods, and general gradient networks.

\subsubsection{Linear regression} \label{sec.special.linear_regression}
For the linear model $\hat{y}_t = \langle \bm{w}_t , \tilde{\bm{x}}_t \rangle$, the sensitivity collapses to $\bm{g}_t = \tilde{\bm{x}}_t$, so $\bm{A}_t = \tilde{\bm{x}}_t \tilde{\bm{x}}_t^\mathsf{T}$. The empirical state $\bm{S}_t$ therefore becomes the discounted covariance of the \textit{observed} feature vectors, while $\bm{B}_t=\bm{S}_t+\epsilon I$ supplies the regularized geometry used by the update. The directional update law (\ref{eq.adaptive_directional_control_law}) specializes to a natural-gradient-like form over the time-varying observed subspace. When excitation is approximately isotropic, $\bm{B}_t \approx (\epsilon + \gamma_t)I$ and (\ref{eq.adaptive_directional_control_law}) recovers the scalar throttling structure of (\ref{eq.linear_regression_adaptive_step}).

In this case, Assumption \ref{assumption_effective_window} becomes
\begin{equation}
    \sum_{k=0}^{T-1}
    \bm{B}_{t-k}^{-1}\tilde{\bm{x}}_{t-k}\tilde{\bm{x}}_{t-k}^\mathsf{T}\bm{B}_{t-k}^{-1}
    \succeq \mu I.
\end{equation}
Lemma \ref{lemma.raw_to_effective_excitation} shows that this effective excitation follows from raw masked-feature excitation together with sufficiently coherent evolution of the inverse covariance geometry. Under these conditions, Proposition \ref{proposition.ISS} takes the same qualitative form as the prototype bound (\ref{eq.prototypic.ISS}), with explicit constants now depending on the effective excitation level $\mu$, or, through (\ref{eq.mu_transfer}), on the raw lower bound $\lambda_g$ and the inverse-metric variation budget $\nu_T$. Our general theory, therefore, remains a natural and formalized extension of the toy controller of Section \ref{sec.prototypic_scalar}, while making explicit the additional geometry introduced by directional adaptation.

\subsubsection{Kernel methods} \label{sec.special.kernel}
In the RKHS setting of Section \ref{sec.prototypic_nonlinear_intro}, the model is $\hat{y}_t = \langle \bm{w}_t, \phi(\tilde{\bm{x}}_t) \rangle_\mathcal{H}$. The sensitivity $\bm{g}_t = \phi(\tilde{\bm{x}}_t)$ lives in the feature space, and $\bm{A}_t$ is the rank-one operator whose trace is $k(\tilde{\bm{x}}_t,\tilde{\bm{x}}_t)$. The same effective directions $\bm{\psi}_t=\bm{B}_t^{-1}\phi(\tilde{\bm{x}}_t)$ determine the windowed excitation condition.

For a finite-dimensional kernel realization---for example, a fixed dictionary, Nystr\"om approximation, random-feature model, or another finite basis---Proposition \ref{proposition.ISS} applies on that realized parameter space whenever the sector, bounded-residual-mismatch, and effective windowed-observability conditions hold. It also applies on a finite-dimensional invariant active subspace $\mathcal{U}\subset\mathcal{H}$ containing the analyzed error trajectory, with all operators and excitation inequalities restricted to $\mathcal{U}$. The constants are then driven by bounded kernel sensitivities and the effective Gramian lower bound $\mu$ on that finite-dimensional space, with Lemma \ref{lemma.raw_to_effective_excitation} providing a sufficient raw-to-effective bridge on the same space.

In a genuinely infinite-dimensional RKHS, however, the finite-window operator
\begin{equation}
    \sum_{k=0}^{T-1}\bm{\psi}_{t-k}\otimes\bm{\psi}_{t-k},
\end{equation}
where $\bm{u}\otimes\bm{u}$ denotes the rank-one operator $\bm{v}\mapsto\bm{u}\langle\bm{u},\bm{v}\rangle_{\mathcal{H}}$, has rank at most $T$ and therefore cannot dominate $\mu I_{\mathcal{H}}$ for any $\mu>0$. We consequently make no claim of full-RKHS finite-window contraction. If the dynamics remain in a finite-dimensional invariant active subspace, Proposition \ref{proposition.ISS} applies after restricting all quantities to that subspace. If excitation is known only after projection but the dynamics are not invariant, a separate coupling, approximation, regularization, or identifiability argument is required. We see this limitation as particularly important in kernel spaces, where missingness can silently impair excitation along entire Hilbert-space directions without any coordinate being explicitly frozen.

\subsubsection{Gradient networks}
For any differentiable model (e.g., multi-layer perceptron, autoencoder, transformer), the instantaneous sensitivity is the full Jacobian row $\bm{J}_t = \partial f(\bm{w}_t, \tilde{\bm{x}}_t) / \partial \bm{w}$, $\bm{g}_t = \bm{J}_t^\mathsf{T}$. The matrix $\bm{A}_t = \bm{g}_t \bm{g}_t^\mathsf{T}$ now aggregates and weighs excitation \textit{across all parameters} that influence the current output. Even if the large scale of contemporary architectures can give rise to computational and scalability considerations, modern architectures naturally group parameters into layers or attention heads whose Jacobian columns are excited together; a block-diagonal approximation for $\bm{B}_t$, as mentioned in Section \ref{sec.practical_approximations}, therefore becomes both cheap and semantically meaningful. Each block tracks the aggregate observability of its layer/head and scales group updates against aggregate recent actuation.

We make no claim of global optimality in the nonconvex landscape, and further elucidate our perspective in Section \ref{sec.nonconvex_perspectives}. Rather, Proposition \ref{proposition.ISS} guarantees bounded and contracting parameter trajectories under its stated closed-loop sector, bounded-residual-mismatch, and effective-excitation conditions. In deep models, these conditions make explicit what must remain true for representation drift to be suppressed under masking: the residual mismatch must remain bounded and the preconditioned Jacobian directions must recurrently span the relevant parameter subspace. Our controller therefore supplies a closed-loop stability mechanism without replacing the descent signal indicated by the loss gradient -- only the effective magnitude and geometry through which it acts.

\subsection{Failure modes, limits, and scope boundaries}

If \textbf{observability} in a subset of the parameter space \textbf{never returns}, learning around that regime stalls. We consider this a limit case that actually illustrates a feature, and not a failure or a limitation of our approach, as learning \textit{should not} continue when not supported by empirical data.

The finite-window observability conditions also carry an unavoidable dimensional scope. A window containing $T$ scalar-output sensitivity directions can certify full-space contraction only on a parameter space of dimension at most $T$. In larger or infinite-dimensional spaces, the guarantee must be interpreted blockwise or on a finite-dimensional active subspace; directions outside that analyzed space are not covered by Proposition \ref{proposition.ISS} unless additional structure controls them.

When \textbf{data and/or the respective gradients that enter the method are not truthful}, canonical failures emerge. These are well studied in the field (e.g., Byzantine resilient learning) and fall outside of our nominal scope in this work. Nevertheless, our adaptation becomes relevant when a subset of the learning data, potentially, cannot be trusted (as determined by some other process), and hence needs to be treated as missing. Instead of dropping the entire iteration, our method allows learning to continue safely, with the data that is still available.

Any learning perturbation and/or adaptation gives rise to \textbf{learnability} and \textbf{identifiability} questions: \textit{will the method learn the same thing with and without the perturbation/adaptation}, or \textit{can the method learn at all (i.e., in a Vapnik-Chervonenkis sense)?} Our guarantees are algebraic: they bound the Euclidean distance of the parameter trajectory from a reference, and ensure that this distance contracts under the conditions of Proposition \ref{proposition.ISS}. Whether small parameter-space distance implies semantic or functional proximity is a property of the model class and its parameter-to-output geometry, not of the training dynamics. In models with rich symmetries, degeneracies, or flat regions (e.g., permutation invariance in transformers, redundant parameterizations in overparameterized networks), Euclidean proximity may carry limited functional meaning. This is a well-known and independently studied phenomenon that no training-time adaptation mechanism can (or, perhaps, should) resolve: it is more related to the architecture and representation design, not to the control of the learning process. Our contributions operate at the level of trajectory controllability --ensuring that the learner moves only where and as much as data justify-- and are therefore compositional with, but categorically distinct from, any analysis of the regularity of the learned map.

\subsection{Perspective on nonconvex settings} \label{sec.nonconvex_perspectives}
Many learning problems take place in nonconvex landscapes. Therein, stochastic gradient descent is understood to leverage noise to explore the parameter space effectively and converge to satisfactory, robust, and well-performing (from a utility perspective, perhaps regardless of global optimality with regards to empirical loss functions) solutions. Our goal here is not to alter gradient descent mechanics or interfere with its ability to explore nonconvex landscapes. Rather, we are driven by a largely orthogonal failure mode: when data are missing, and regardless of the classical remedy employed (imputation, masking, etc.), the gradient feedback is computed by a reduced \textit{or} partially hypothesized observation operator, which can induce open-loop drift and incoherent parameter motion in directions unsupported by actual, empirically-available evidence. Seeing missing data as a potential  ``benefit'' for SGD (i.e., akin to exploration-promoting noise) does not resolve this structural loss of actuation; in fact, it can exacerbate instability by injecting variance along poorly observed subspaces (especially as missing data patterns may be correlated in time and not spontaneously random).

In either convex or nonconvex settings, and for whatever expected outcome an SGD method is used for (exact or empirical), our contribution is about trajectory sanity under partial actuation. We attenuate unjustified parameter motion when the gradient channel is information-poor or structurally censored. Even more so when either symptom persists and accumulates, as captured by the composition of $\bm{B}_t$. Ultimately, our adaptive control formulation aims to render closed loop learning dynamics reliable and coherent under the finite-window residual-to-state conditions made explicit in Proposition \ref{proposition.ISS}, while also preventing parameter hallucinations, as we demonstrate next.

\begin{figure*}[t!]
    \includegraphics[width=\linewidth]{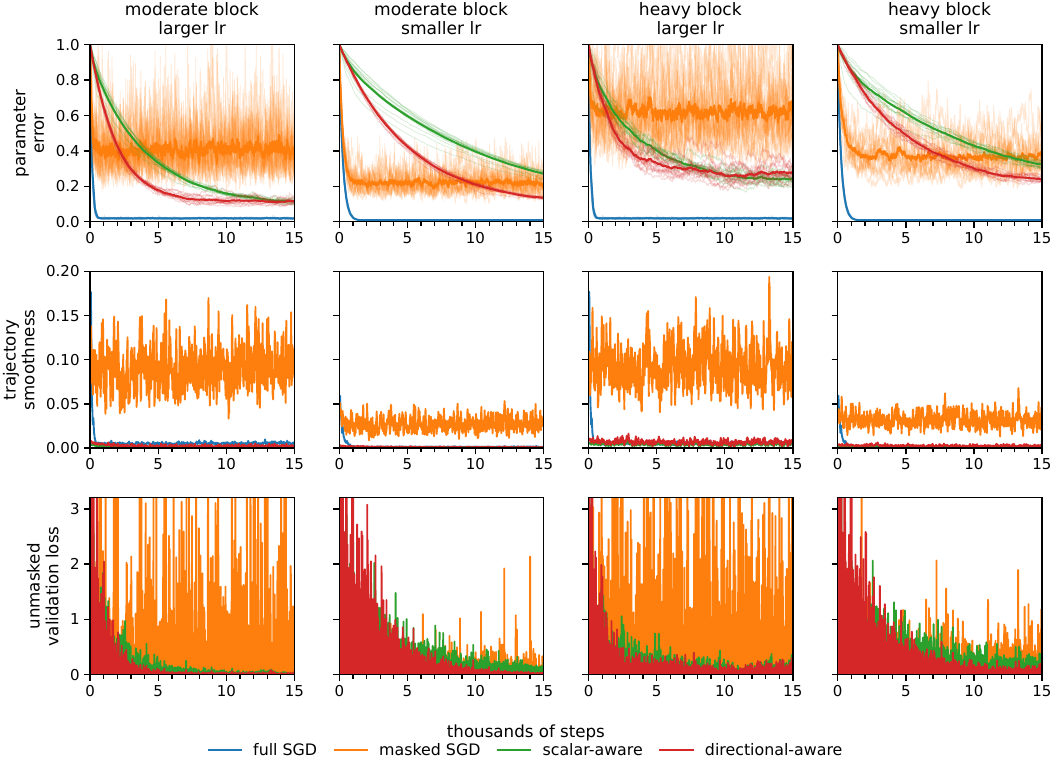}
    \caption{Numerical results from \textbf{Experiment \ref{sec.exp1}}: \textit{linear regression under missing blocks and individual features} in $d = 40$ dimensions. In the plotted results we consider \textbf{block missingness} of two different severities (moderate / heavy), that we sketch in Figure \ref{fig.patterns.linear.block}, and step sizes $\eta = \{ 0.020, 0.008\}$. For the parameter error, we consider $15$ different runs as well as their pointwise-in-time mean (drawn more heavily against the distribution), whereas for the rest of the plots we examine a single run for each algorithm (full / masked / scalar-adaptive / directionally-adaptive) to better illustrate their individual behavior.}
    \label{fig.linear.block}
\end{figure*}

\begin{figure*}[t!]
    \includegraphics[width=\linewidth]{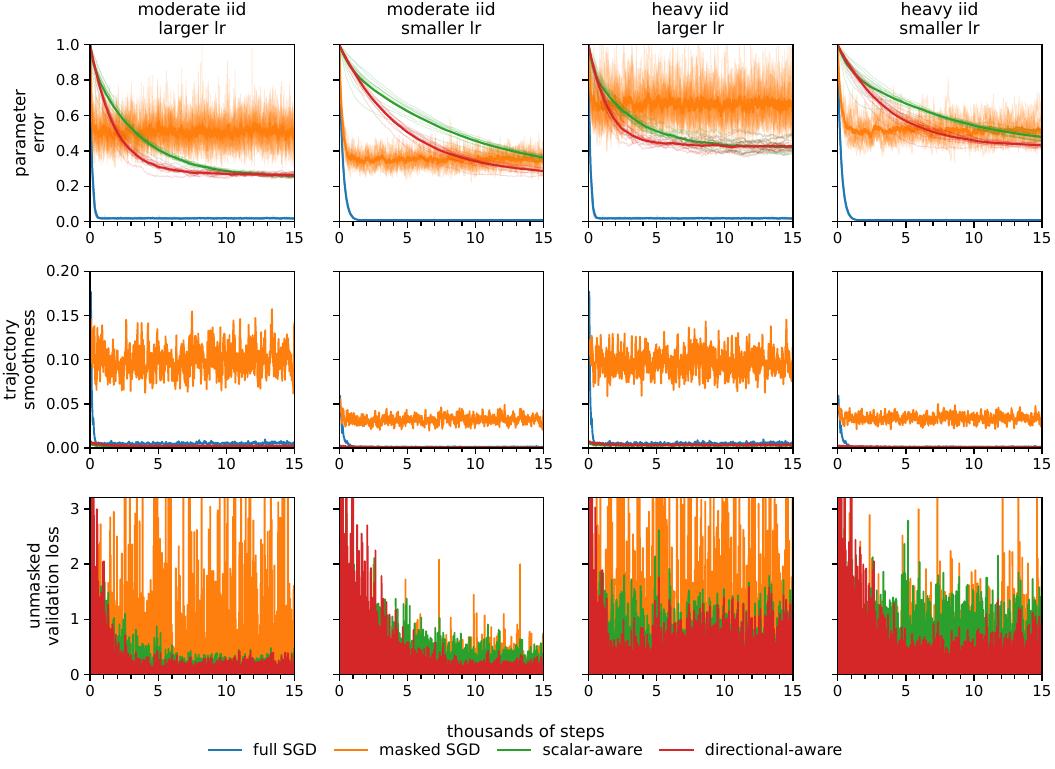}
    \caption{Numerical results from \textbf{Experiment \ref{sec.exp1}}: \textit{linear regression under missing blocks and individual features} in $d = 40$ dimensions. In the plotted results we consider \textbf{iid missingness} of two different severities (moderate / heavy), that we sketch in Figure \ref{fig.patterns.linear.iid}, and step sizes $\eta = \{ 0.020, 0.008\}$. For the parameter error, we consider $15$ different runs as well as their pointwise-in-time mean (drawn more heavily against the distribution), whereas for the rest of the plots we examine a single run for each algorithm (full / masked / scalar-adaptive / directionally-adaptive) to better illustrate their individual behavior.}
    \label{fig.linear.iid}
\end{figure*}

\begin{figure}[t]
  \centering
  \begin{subfigure}[t]{0.48\textwidth}  
    \centering
    \includegraphics[width=\textwidth]{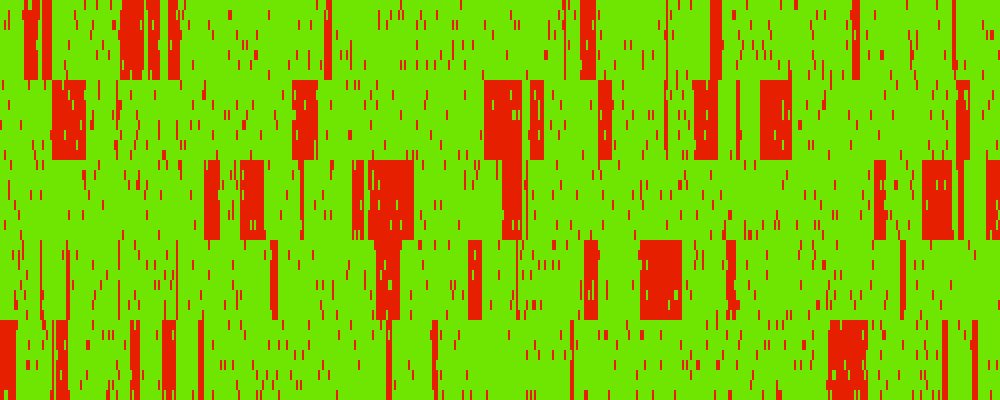}
    \caption{500 instances of moderate block missingness}
    \label{fig.patterns.linear.block.moderate}
  \end{subfigure}
  \hfill
  \begin{subfigure}[t]{0.48\textwidth}
    \centering
    \includegraphics[width=\textwidth]{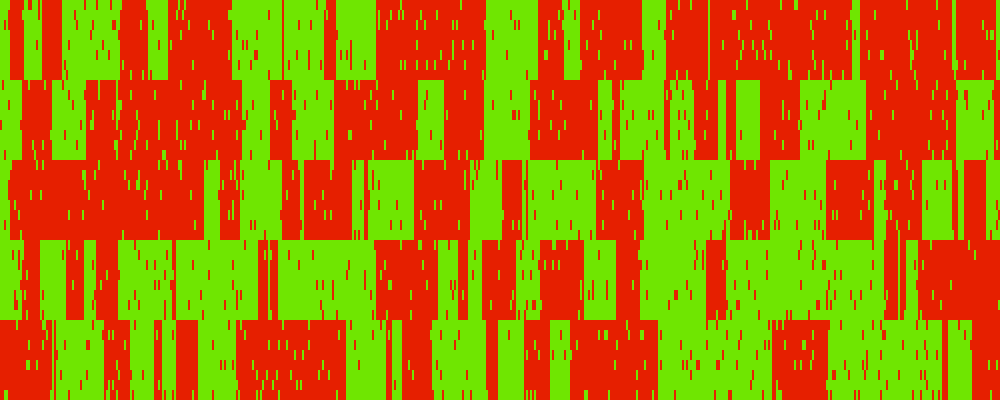}
    \caption{500 instances of heavy block missingness}
    \label{fig.patterns.linear.block.severe}
  \end{subfigure}
    \caption{Representative realizations of the \textbf{block missingness patterns} used in \textbf{Experiment \ref{sec.exp1}}. We show the first $500$ training instances under moderate and heavy block missingness. The horizontal axes show time; the vertical axes, the 40 features. Compared to iid masking, these patterns induce temporally persistent losses of observability, making the same partial actuation mechanism that we discussed in the description of Experiment \ref{sec.exp1} more pronounced and easier to visualize.}    \label{fig.patterns.linear.block}
\end{figure}

\begin{figure}[t]
  \centering
  \begin{subfigure}[t]{0.48\textwidth} 
    \centering
    \includegraphics[width=\textwidth]{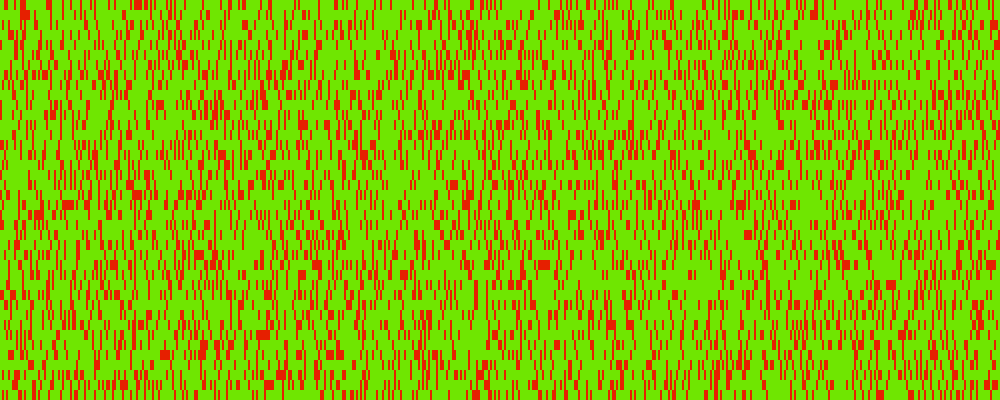}
    \caption{500 instances of moderate iid missingness}
    \label{fig.patterns.linear.iid.moderate}
  \end{subfigure}
  \hfill
  \begin{subfigure}[t]{0.48\textwidth}
    \centering
    \includegraphics[width=\textwidth]{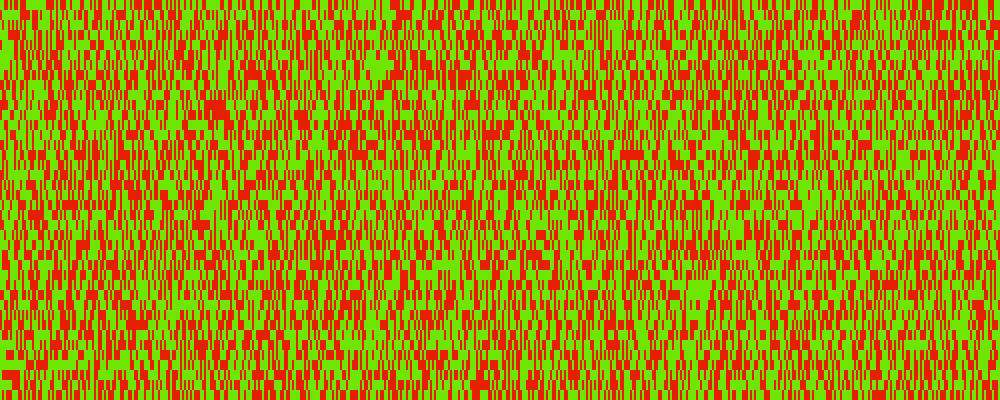}
    \caption{500 instances of heavy iid missingness}
    \label{fig.patterns.linear.iid.severe}
  \end{subfigure}
    \caption{Representative realizations of the \textbf{iid missingness patterns} used in \textbf{Experiment \ref{sec.exp1}}. We show the first $500$ training instances under moderate and heavy iid missingness. The horizontal axes show time; the vertical axes, the 40 features. Although visually less persistent than the block case, iid masking still induces the same partial-actuation failure mode: at each iteration, some coordinates are updated while others are not, even though all coordinates remain coupled through the shared residual. Moreover, with increasing missingness severity, missing feature blocks in consecutive samples emerge naturally.}  \label{fig.patterns.linear.iid}
\end{figure}

\section{Experiments} \label{section.experiments}

We consider three learning problems that allow us to illuminate the mechanics of our proposed adaptation and study its efficacy even in rather complex architectures. In line with our theoretical development, we consider scalar output problems. Moreover, in all experiments, our goal is to learn the mapping that operates on the \textit{complete} (that is, unmasked, with nothing missing) feature space.

\subsection{Linear regression under missing blocks and individual features} \label{sec.exp1}

In our simplest experiment, we consider a synthetic linear regression problem in $d = 40$ dimensions, with data generated according to $y_t = \langle \bm{w}^*, \bm{x}_t \rangle + \epsilon_t$, where $\bm{x}_t$ is Gaussian with block-structured covariance composed of 5 feature groups, within-block correlation 0.75, and diagonal jitter 0.25$I$, while $\epsilon_t$ is zero-mean Gaussian noise with standard deviation $0.02$. The ground-truth parameter $\bm{w}^*$ is group-sparse. At training time, the learner only observes the masked input $\tilde{\bm{x}}_t = \bm{M}_t \bm{x}_t$, yet the target remains the complete-space predictor $\bm{w}^*$. We optimize the quadratic loss $\ell_t(\bm{w}) = (y_t - \langle \bm{w}, \tilde{\bm{x}}_t\rangle)^2 /2$ for the masked learners, with the full-data baseline evaluated analogously on the unmasked input. We study both structured block missingness and iid feature missingness, each under moderate and heavy severity --the patterns of which we indicatively illustrate in Figures \ref{fig.patterns.linear.block} and \ref{fig.patterns.linear.iid}, and compare four update laws: full SGD, masked SGD, scalar observability-aware adaptation (implemented via (\ref{eq.linear_regression_adaptive_step_alpha})-(\ref{eq.linear_regression_adaptive_step}) with $\epsilon = 1$), and directional observability-aware adaptation (implemented via (\ref{eq.adaptive_directional_control_law}) and (\ref{eq.directional.B_recursive}), with $\bm{A}_t = \tilde{\bm{x}}_t \tilde{\bm{x}}_t^\mathsf{T}$ as discussed in Section \ref{sec.special.linear_regression} for the linear regression case, and using the values $\epsilon = 0.05$ and $\beta = 0.99$ in (\ref{eq.directional.B_recursive})). Although we demonstrate our methods here in an offline learning pipeline, the underlying failure mode is not restricted to temporally blocked missingness. Even under iid masking, partial actuation remains: at each step, some coordinates are updated while others are not, yet all coordinates remain coupled through the same scalar residual. In sufficiently shuffled offline datasets this effect may appear more averaged over time, especially in simple problems, but it is not removed. Structured block missingness simply makes the same mechanism more persistent and therefore more visible. We therefore study both iid and block patterns, with the latter also serving as a natural proxy for online or streaming settings, where shuffling is generally unavailable. In fact, the present linear regression experiment is sample-wise, i.e., with a batch size of 1, so the learner directly experiences the temporal ordering of the missingness process.

\begin{table}[hb!]
\scriptsize
\centering
\begin{tabular}{cc|cc|cc|c|
>{\columncolor[HTML]{dcbe62}}c 
>{\columncolor[HTML]{FFCCC9}}c| 
>{\columncolor[HTML]{cdffcc}}c }
\multicolumn{2}{c|}{\textbf{missingness}} & \multicolumn{2}{c|}{\textbf{severity}} & \multicolumn{2}{c|}{\textbf{learning rate}} &  \textbf{method} & \multicolumn{2}{c|}{\cellcolor[HTML]{FFFFFF}\textbf{misalignment}} & {\cellcolor[HTML]{FFFFFF} \textbf{alignment}} \\
block & iid & moderate & heavy & $\eta = 0.020$ & $\eta = 0.008$ &  \textbf{used} & {\cellcolor[HTML]{FFFFFF} severe $(-\infty,-0.04]$} & {\cellcolor[HTML]{FFFFFF} $(-\infty,0]$} & {\cellcolor[HTML]{FFFFFF} $(0,\infty)$} \\ \hline
 &  &  &  &  &  & \cellcolor[HTML]{FF7F0E} masked & 25.7\% & 64.5\% & 35.5\% \\
$\bullet$ &  & $\bullet$ &  & $\bullet$ &  & \cellcolor[HTML]{2CA02C} \textbf{scalar}  & 12.9\% & \textbf{44.4\%} & \textbf{55.6\%} \\
 &  &  &  &  &  & \cellcolor[HTML]{e25e5f} \textbf{directional}  & \phantom{0}\textbf{3.6\%} & 45.7\% & 54.3\% \\ \hline
 &  &  &  &  &  & \cellcolor[HTML]{FF7F0E}masked & 16.3\% & 50.5\% & 49.5\% \\
$\bullet$ &  & $\bullet$ &  &  & $\bullet$ & \cellcolor[HTML]{2CA02C}\textbf{scalar}  & 12.9\% & \textbf{44.2\%} & \textbf{55.8\%} \\
 &  &  &  &  &  & \cellcolor[HTML]{e25e5f} \textbf{directional}  & \phantom{0}\textbf{3.6\%} & 44.4\% & 55.6\% \\ \hline
 &  &  &  &  &  & \cellcolor[HTML]{FF7F0E}masked & 26.3\% & 59.3\% & 40.7\% \\
$\bullet$ &  &  & $\bullet$ & $\bullet$ &  & \cellcolor[HTML]{2CA02C}\textbf{scalar}  & 16.7\% & \textbf{44.8\%} & \textbf{55.2\%} \\
 &  &  &  &  &  & \cellcolor[HTML]{e25e5f} \textbf{directional}  & \phantom{0}\textbf{7.7\%} & 48.0\% & 52.0\% \\ \hline
 &  &  &  &  &  & \cellcolor[HTML]{FF7F0E}masked & 20.9\% & 51.5\% & 48.5\% \\
$\bullet$ &  &  & $\bullet$ &  & $\bullet$ & \cellcolor[HTML]{2CA02C}\textbf{scalar}  & 16.9\% & \textbf{44.6\%} & \textbf{55.4\%} \\
 &  &  &  &  &  & \cellcolor[HTML]{e25e5f} \textbf{directional}  & \phantom{0}\textbf{7.5\%} & 47.3\% & 52.7\% \\ \hline
 &  &  &  &  &  & \cellcolor[HTML]{FF7F0E}masked & 18.6\% & 63.3\% & 36.7\% \\
 & $\bullet$ & $\bullet$ &  & $\bullet$ &  & \cellcolor[HTML]{2CA02C}\textbf{scalar}  & 10.0\% & \textbf{47.4\%} & \textbf{52.6\%} \\
 &  &  &  &  &  & \cellcolor[HTML]{e25e5f} \textbf{directional}  & \phantom{0}\textbf{4.0\%} & 48.6\% & 51.4\% \\ \hline
 &  &  &  &  &  & \cellcolor[HTML]{FF7F0E}masked & 13.8\% & 55.1\% & 44.9\% \\
 & $\bullet$ & $\bullet$ &  &  & $\bullet$ & \cellcolor[HTML]{2CA02C}\textbf{scalar}  & \phantom{0}8.9\% & \textbf{43.8\%} & \textbf{56.2\%} \\
 &  &  &  &  &  & \cellcolor[HTML]{e25e5f} \textbf{directional}  & \phantom{0}\textbf{3.7\%} & 47.6\% & 52.4\% \\ \hline
 &  &  &  &  &  & \cellcolor[HTML]{FF7F0E}masked & 10.9\% & 58.3\% & 41.7\% \\
 & $\bullet$ &  & $\bullet$ & $\bullet$ &  & \cellcolor[HTML]{2CA02C}\textbf{scalar}  & \phantom{0}7.6\% & \textbf{48.9\%} & \textbf{51.1\%} \\
 &  &  &  &  &  & \cellcolor[HTML]{e25e5f} \textbf{directional}  & \phantom{0}\textbf{4.4\%} & 50.3\% & 49.7\% \\ \hline
 &  &  &  &  &  & \cellcolor[HTML]{FF7F0E}masked & \phantom{0}8.7\% & 52.3\% & 47.7\% \\
 & $\bullet$ &  & $\bullet$ &  & $\bullet$ & \cellcolor[HTML]{2CA02C}\textbf{scalar}  & \phantom{0}6.7\% & \textbf{45.8\%} & \textbf{54.2\%} \\
 &  &  &  &  &  & \cellcolor[HTML]{e25e5f} \textbf{directional}  & \phantom{0}\textbf{4.2\%} & 49.6\% & 50.4\%
\end{tabular}
\caption{As introduced in Section \ref{sec.metrics}, our \textit{coherence} metric aims to quantify how self-aligned a method's trajectory is. Here we use our equation for $\mathrm{Coherence}(t)$ over a window of length $W = 50$, throughout the 15000 learning steps for all our runs (15 repetitions per indicated case) for \textbf{Experiment \ref{sec.exp1}}: \textit{linear regression under missing blocks and individual features}. The percentages tabulated above indicate how much effort and time each method is spending while aligned or misaligned with itself and its recent past. We note that these values are averaged over the recent past. However, alongside Figures \ref{fig.linear.block} and \ref{fig.linear.iid}, they help illuminate how the masked method is more often inconsistent with itself, as features come and go in the data stream, in contrast to the directional method (more self-aligned yet eager to exploit opportunities to move parameters) and even more so the scalar method (maximally self-aligned in comparison, yet, overall slower, and occasionally exhibiting severe misalignment more often when compared to directional as it cannot attenuate individual parameters without stopping learning at large). We omit the fully observed baseline from the table, since its update directions after reaching its effective steady state are dominated by small residual fluctuations rather than any missingness-induced alignment behavior.}
\label{table.linear.coherence}
\end{table}

While this example does not include any function approximation complexity, it is meant to expose, as transparently as possible, the dynamical consequences of data missingness on the learning trajectories. This linear example is particularly illuminating because the apparent simplicity of the model can hide a critical failure mode. If a coordinate is masked (i.e., in a real-world pipeline: missing, corrupt, or currently untrusted), then, under strict masking, that specific coordinate receives no instantaneous gradient update. However, regardless of whether masking is iid or blocked, the residual still depends on the unobserved part of the signal, so the observed coordinates should not move freely: the true but unavailable value of a masked feature may have implied a different motion for both unobserved and observed parameters. In this sense, the coupling does not originate in the forward map (which is trivial in linear regression), but in the loss-induced error dynamics under partial actuation. Standard masked SGD ignores this distinction and continues to move the currently actuated coordinates as if the whole parameter vector were being updated consistently, which leads to incoherent motion, rough trajectories, large loss spikes, and increased parameter hallucinations. By contrast, our scalar controller restores global caution under reduced observability, while our directional controller further exploits anisotropy in the recent excitation history to attenuate unsupported motion without throttling all directions equally. The result is a setting in which the mechanics of the proposed methods, and the implications of our coherence and smoothness metrics, are unusually easy to see and study. This is precisely the partial-actuation mechanism isolated in Section \ref{sec.prototypic_scalar}: masking does not destabilize learning by moving hidden coordinates, but by allowing visible coordinates to move as if the hidden ones were also evolving consistently.

We report our numerical results in Figure \ref{fig.linear.block} (block missingness), Figure \ref{fig.linear.iid} (iid missingness), and Table \ref{table.linear.coherence}. We plot the parameter error $\| \bm{w}_t - \bm{w}^* \|$ directly for $15$ runs for each case (block/iid missingness, moderate/heavy missingness, larger/smaller learning rate $\eta \in \{0.020, 0.008\}$) and method, and the pointwise-in-time mean. For one trajectory of the $15$, we plot our $\mathrm{Smoothness}(t)$ metric and tabulate $\mathrm{Coherence}(t)$ evaluated over a rolling window with $W = 50$ steps and binned/classified accordingly to indicate the method's occupation time in a state of averaged \textit{severe misalignment} ($\mathrm{Coherence}(t) \leq -0.04$), \textit{cumulative misalignment} ($\mathrm{Coherence}(t) \leq 0$) and \textit{alignment} ($\mathrm{Coherence}(t) > 0$).

Overall, the message of Experiment 6.1 is unambiguous: observability-aware adaptation wins consistently over naive masking. Across both iid and block missingness, and across both larger and smaller learning rate settings, the scalar and directional methods achieve lower parameter error, tighter run-to-run spread, substantially smoother trajectories, and better coherence statistics than their masked counterparts. The larger/smaller learning rate comparison is included specifically to rule out the trivial explanation that these gains arise merely from slowing learning down. Although reducing the learning rate does improve the masked baseline somewhat, it does not reproduce observability-aware behavior. Table \ref{table.linear.coherence} reinforces this point quantitatively: both observability-aware methods reduce severe misalignment relative to masked SGD in every reported regime, while also spending more time in aligned updates. Our gains therefore do not come from ``slower SGD'' alone, but from regulating parameter motion according to the current and recent observability structure.

\begin{figure*}[t!]
    \includegraphics[width=1.0\textwidth]{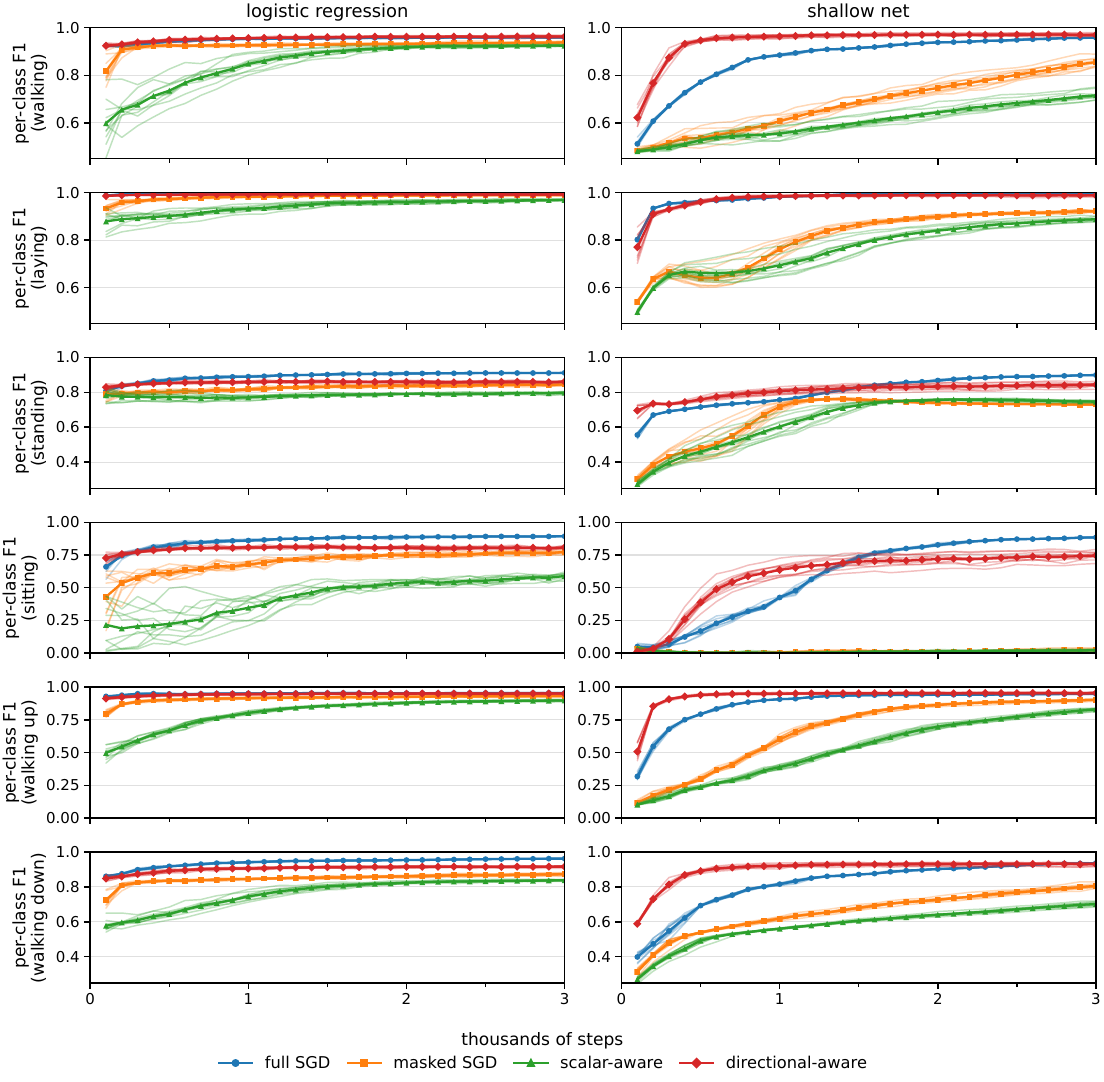}
    \caption{Numerical results from \textbf{Experiment \ref{sec.exp2}}: classification of human activities under multiple concurrently missing sensing modalities and missing sensor measurements. We train one-versus-rest classifiers based on logistic regression (left column) and a shallow neural network (right column), and report per-class F1 trajectories on the full unmasked test set for the six activity classes. In all plotted runs, training is performed under heavy coarse missingness: 4 out of 6 sensing groups are hidden in each training sample, and additional within-group feature drops are imposed on the remaining visible measurements. For each method, we show 10 seed trajectories faintly and their pointwise-in-time mean more heavily. Across both model classes, the directional observability-aware learner consistently improves over masked SGD, while the scalar-aware learner remains visibly more cautious and therefore slower, as expected from an isotropic controller. Table \ref{table.ex2} helps decongest nearby trajectories and highlight the directional method's strong gains.}
    \label{fig.exp2.trajectories}
\end{figure*}

\subsection{Learning to classify human activities under multiple concurrently missing sensing modalities and missing sensor measurements} \label{sec.exp2}

Progressing from the bare-bones context of Experiment \ref{sec.exp1}, which we leveraged to study in detail trajectory-level behaviors under masked SGD and our adaptation methodologies, we now consider a bona fide learning pipeline where multiple modalities and impactful learning outcomes can be more tangibly and organically conceptualized. We build our missing modalities and features context around the \textit{Human Activity Recognition} (HAR) dataset \citep{anguita2013public, human_activity_recognition_using_smartphones_240}, released under the CC BY 4.0 license, which includes 3D acceleration and gyroscopic signals in the time and frequency domain from a smartphone sensor, labels for the human subject's activity corresponding to the recorded features (i.e., {WALKING}, {WALKING UPSTAIRS}, {WALKING DOWNSTAIRS}, {SITTING}, {STANDING}, {LAYING}), and a hard split between training and evaluation data, with the latter obtained from different human subjects.

To induce heavy partial observability, we group the HAR features into six coarse sensing families and, for each training sample, hide exactly 4 out of the 6 groups while additionally retaining each feature within the remaining visible groups independently with probability 0.8. The realized mask is persistent per datapoint within a run, so repeated encounters with the same sample expose the same missingness pattern. We train six one-versus-rest classifiers and compare two model classes: logistic regression and a shallow neural network of architecture $561 \to 12 \to 1$. Full implementation details, including losses, parameter counts, and the numerically safeguarded full-order controller realization used here, are given in Appendix \ref{app.sec.exp2}.

\begin{figure}[ht!]
  \centering
    \includegraphics[width=1.0\textwidth]{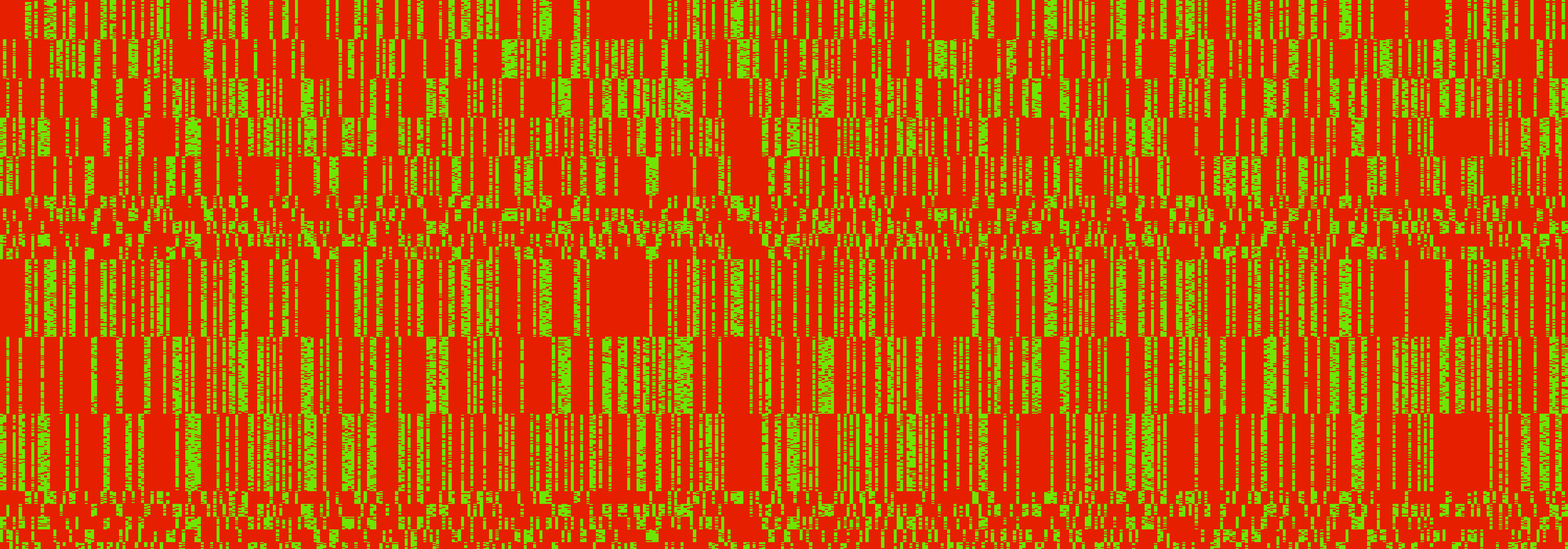}
    \caption{Representative realization of the heavy HAR missingness pattern used in \textbf{Experiment \ref{sec.exp2}}. The vertical axis is the canonical 561-feature HAR representation, whose ordering consists of contiguous signal-family ranges; under our coarse masking scheme, these fixed ranges are merged into six broader sensing families: body acceleration, gravity acceleration, body-acceleration jerk, body gyroscope, body-gyroscope jerk, and angle. For each training sample, the mask is first initialized by retaining each feature independently with probability 0.8, after which exactly 4 of the 6 coarse families are selected uniformly and fully hidden. Since each coarse family is itself a fixed union of several separated official signal-family ranges, the hidden regions appear as repeated structured horizontal bands rather than as a single contiguous slab; the additional fine-grained red/green texture comes from the within-visible Bernoulli thinning. The figure shows the first 500 training instances from one such fixed-mask realization.}\label{fig.patterns.har}
\end{figure}

In Figure \ref{fig.exp2.trajectories}, we plot per-class F1 trajectories on the full unmasked test set, evaluated every 100 optimization steps. Several messages are visible. First, our directional method consistently improves over naive masking in both model classes, and does so in the metric that matters here: full-view classification performance under heavy training-time modality loss. Second, the scalar controller remains useful but also reveals its basic limitation in this richer setting: because it throttles isotropically, it is often too conservative and therefore slower, especially in the shallow network case. We nevertheless view it as an important stepping stone, since it expresses the same control principle in its simplest form and helps bridge the path from the prototypic controller of Section \ref{sec.prototypic_scalar} to our directional machinery developed in Sections \ref{sec.directional}-\ref{sec.stability}. Third, we deliberately refrain from making model-selection claims from this experiment alone: logistic regression may simply be better matched to this particular HAR representation, or the chosen shallow network may be too small to fully exploit the data. Our goal here is not to optimize architecture design, but to stress the learning dynamics under heavy concurrent modality and feature loss. Finally, the shallow network results contain an especially interesting effect: in several classes and over substantial early-training intervals, the directional controller even outpaces the fully observed SGD baseline. We interpret this cautiously but positively: once embedded in a nonlinear model, the observability-aware geometry appears to do more than merely repair missing-data training, and can act as a beneficial conditioner of the optimization dynamics themselves.

In Table \ref{table.ex2}, we examine more closely the performance of the methods and the strong gains demonstrated by our directional observability-aware adaptation. 

\begin{table}[t] \centering
\scriptsize
\begin{tabular}{r c | rrrr|rrrr}
 &  & \multicolumn{4}{c|}{logistic regression} & \multicolumn{4}{c}{shallow net} \\
 &  & t = 100 & t = 200 & t = 500 & t = 1000 & t = 100 & t = 200 & t = 500 & t = 1000 \\ \hline
walking & full & 0.925 & 0.921 & 0.941 & 0.953 & 0.511 & 0.607 & 0.771 & 0.885 \\
 & masked & 0.817 & 0.907 & 0.925 & 0.926 & 0.484 & 0.496 & 0.537 & 0.607 \\
 & directional & 0.924 & 0.930 & 0.949 & 0.956 & 0.622 & 0.767 & 0.946 & 0.963 \\ \hline
laying & full & 0.987 & 0.992 & 0.994 & 0.995 & 0.802 & 0.933 & 0.964 & 0.982 \\
 & masked & 0.933 & 0.958 & 0.973 & 0.982 & 0.540 & 0.636 & 0.641 & 0.762 \\
 & directional & 0.986 & 0.989 & 0.991 & 0.992 & 0.770 & 0.910 & 0.961 & 0.984 \\ \hline
standing & full & 0.809 & 0.836 & 0.870 & 0.890 & 0.555 & 0.670 & 0.716 & 0.757 \\
 & masked & 0.789 & 0.793 & 0.801 & 0.817 & 0.301 & 0.385 & 0.478 & 0.718 \\
 & directional & 0.829 & 0.839 & 0.852 & 0.859 & 0.696 & 0.735 & 0.759 & 0.807 \\ \hline
sitting & full & 0.660 & 0.745 & 0.820 & 0.860 & 0.050 & 0.043 & 0.169 & 0.425 \\
 & masked & 0.430 & 0.535 & 0.609 & 0.679 & 0.036 & 0.015 & 0.004 & 0.005 \\
 & directional & 0.727 & 0.757 & 0.792 & 0.808 & 0.012 & 0.034 & 0.388 & 0.636 \\ \hline
walking up & full & 0.929 & 0.938 & 0.948 & 0.952 & 0.317 & 0.547 & 0.794 & 0.908 \\
 & masked & 0.796 & 0.869 & 0.902 & 0.917 & 0.118 & 0.171 & 0.298 & 0.601 \\
 & directional & 0.913 & 0.923 & 0.939 & 0.947 & 0.507 & 0.854 & 0.940 & 0.950 \\ \hline
walking down & full & 0.861 & 0.875 & 0.918 & 0.941 & 0.398 & 0.473 & 0.693 & 0.816 \\
 & masked & 0.724 & 0.808 & 0.836 & 0.845 & 0.314 & 0.408 & 0.539 & 0.618 \\
 & directional & 0.849 & 0.860 & 0.892 & 0.905 & 0.589 & 0.731 & 0.891 & 0.921
\end{tabular}
\caption{Mean per-class F1 scores at selected optimization steps for \textbf{Experiment \ref{sec.exp2}}, computed over the same runs used in Figure \ref{fig.exp2.trajectories} and evaluated on the \textbf{full unmasked test split}. Rows correspond to activity classes and methods; columns report checkpointed means for \textbf{logistic regression} and the \textbf{shallow neural network} at $t \in \{100, 200, 500, 1000\}$. We include this table to decongest the trajectory plots in Figure \ref{fig.exp2.trajectories}, especially on the \textbf{logistic regression} side. In that setting, our \textbf{directional observability-aware} method exceeds \textbf{masked SGD} in \textbf{all} 24 class-step entries, with an average gain of $0.0765$ F1 ($\approx 7.65$ percentage points), corresponding to about $9.4\%$ improvement in overall mean F1; it also exceeds \textbf{full SGD} in 7 of 24 entries, all at early or intermediate stages. The largest masked-to-directional improvement on the logistic-regression side occurs for \textbf{sitting} at $t = 100$, where the mean F1 rises from 0.430 to 0.727 (a gain of 0.297, or roughly $69\%$ relative improvement). Other especially strong logistic-regression gains appear for \textbf{walking down} at $t = 100$ ($0.724 \to 0.849$) and for \textbf{walking} at $t = 100$ ($0.817 \to 0.924$). Moreover, the cases in which \textbf{directional training surpasses even full SGD} are not isolated curiosities. On the logistic-regression side, they already appear in several early and intermediate checkpoints, most visibly for walking, and also for standing and sitting. In the \textbf{shallow neural network} case, the effect is markedly stronger: the directional method exceeds \textbf{masked SGD} in \textbf{23 of 24} class-step entries and exceeds \textbf{full SGD} in \textbf{19 of 24} entries. The lone shallow-net exception against masked occurs for \textbf{sitting} at the first checkpoint ($t = 100$), but this does not alter the qualitative message: masked training then effectively collapses for that class ($0.015, 0.004, 0.005$ at $t = 200, 500, 1000$), whereas the directional method recovers substantial performance ($0.034, 0.388, 0.636$). This suggests that, once embedded in a nonlinear model, the directional controller may do more than merely compensate for missingness: by reshaping motion according to the observed actuation geometry, it can also provide a beneficial conditioning effect on the optimization dynamics themselves.} \label{table.ex2}
\end{table}

\subsection{Learning to predict sea ice concentration with remote sensing under missing sensing modalities and missing sensor measurements} \label{sec.exp3}

In our search of a more complex model where intermittently or persistently missing sensing modalities arise naturally from real data acquisition and labeling constraints, we consider the AI4Arctic / ASIP Sea Ice Dataset, version 2 (ASID-v2) released by \citet{asip} under the CC BY 4.0 license, alongside the companion multisensor fusion study \citep{malmgren2020convolutional}. ASID-v2 contains Synthetic Aperture Radar (SAR) imagery obtained by the ESA Sentinel-1 satellite constellation, microwave radiometer measurements from the AMSR2 instrument onboard the JAXA GCOM-W satellite, and the corresponding ice charts manually produced by trained analysts at the Danish Meteorological Institute. As emphasized by the ASID-v2 contributors, sea ice charting is labor-intensive, depends on availability of multiple modalities of satellite data and human analysts, and is performed under operational constraints in which one ice chart is typically produced to correspond to one Sentinel-1 scene for the respective acquisition time. This makes the ASID-v2 dataset especially suitable for our context and learning objectives in this work: the supervision itself is expensive, the sensing modalities are heterogeneous and not equally reliable, and insisting that every labeled training instance must contain every sensing source in perfect form would discard a substantial amount of otherwise useful supervision. In such a setting, methods that can continue learning from partially observed examples are not merely algorithmically elegant; they directly relax a critical data requirement in a domain where obtaining labels and fully matched multisensor observations is costly.

Our learning environment and pipeline are focused on learning to predict neighborhood-wide sea ice concentration using SAR and AMSR data patches, and learning to do so under significant training data missingness. We describe our model and pipeline, as well as how they differ from the original pixelwise sea ice mapping objective of \citet{malmgren2020convolutional} in Appendix \ref{sec.app.exp3}. We sketch our convolutional model in Figure \ref{fig.ice.model}. As our model includes 368,753 trainable parameters, we build our controller machinery using a block-diagonal approximation of $\bm{A}_t$ to maintain $\bm{B}_t$ according to (\ref{eq.directional.B_recursive}), as we also describe in Appendix \ref{sec.app.exp3}.

\begin{table}[t]
\centering
\scriptsize
\begin{tabular}{cc|cccccc}
\renewcommand{\arraystretch}{1.5}
         &             & mean val.\ MAE        & c<m         & c best             & good region gain & bad excursion reduction & mean $\vert \Delta\text{MAE} \vert$ \\ 
rot.\ & window      & (f / m / c)             & (occ.) &  (occ.) & (c vs. m)          & (c vs. m)                 & (m / c)            \\ \hline
A        & 0-1k      & 14.08 / 22.80 / 8.65  & 100\%     & 90\%         & +40.0 pp         & +80.0 pp                & 12.25 / 2.52     \\
A        & 0-5k      & 9.97 / 12.13 / 7.59   & 90\%      & 82\%         & +34.0 pp         & +36.0 pp                & 4.47 / 1.39      \\
A        & 5k-20k  & 7.30 / 7.83 / 7.28    & 65\%      & 33\%         & +4.6 pp          & +7.9 pp                 & 1.09 / 0.75      \\ 
A        & 10k-20k & 6.96 / 7.70 / 7.27    & 59\%      & 20\%         & -4.0 pp          & +6.9 pp                 & 0.95 / 0.70      \\\hline
B        & 0-1k      & 11.19 / 21.03 / 6.10  & 100\%     & 70\%         & +10.0 pp         & +90.0 pp                & 11.62 / 2.09     \\
B        & 0-5k      & 6.68 / 10.99 / 8.20   & 58\%      & 34\%         & +10.0 pp         & +8.0 pp                 & 4.46 / 3.87      \\
B        & 5k-20k  & 4.09 / 6.37 / 5.14    & 68\%      & 17\%         & +14.6 pp         & +5.3 pp                 & 2.28 / 1.85      \\ 
B        & 10k-20k & 3.71 / 5.61 / 4.81    & 65\%      & 12\%         & +16.8 pp         & +2.0 pp                 & 1.88 / 1.57      \\\hline
C        & 0-1k      & 19.12 / 23.36 / 15.96 & 90\%      & 70\%         & +0.0 pp          & +30.0 pp                & 6.00 / 1.68      \\
C        & 0-5k      & 15.03 / 15.63 / 12.10 & 84\%      & 70\%         & +28.0 pp         & +22.0 pp                & 3.05 / 1.70      \\
C        & 5k-20k  & 10.77 / 11.71 / 10.57 & 68\%      & 44\%         & +19.2 pp         & +7.9 pp                 & 1.61 / 1.17  \\
C        & 10k-20k & 10.39 / 11.19 / 10.40 & 62\%      & 38\%         & +16.8 pp         & +5.0 pp                 & 1.40 / 1.19   
\end{tabular}
\caption{Validation occupation-time and reliability statistics for \textbf{Experiment \ref{sec.exp3}}, computed from the clean full-feature validation MAE trajectories for rotations A/B/C that we plotted in Figure \ref{fig.sea_ice_trajectories} over the training step windows 0--1000, 0--5000, 5000--20000, and 10000--20000. In \textbf{f/m/c}, we report the window-averaged validation MAE of \textbf{f} (full SGD), \textbf{m} (masked SGD), and \textbf{c} (our controlled learner); \textbf{c<m occ.} is the fraction of checkpoints at which controlled outperforms masked on validation MAE, while \textbf{c best occ.} is the fraction of checkpoints at which controlled is best among $\{$full, masked, controlled$\}$; \textbf{good region gain} (c vs.\ m) is the controlled-minus-masked difference, in percentage points, in occupation time spent below a rotation-specific good-region threshold; \textbf{bad-excursion reduction} (c vs.\ m) is the masked-minus-controlled difference, again in percentage points, in occupation time spent above a rotation-specific bad-excursion threshold. These thresholds are defined within each rotation from pooled validation MAE values across methods: \textit{good} = lower quartile and \textit{bad} = 90th percentile, yielding (A: good $\leq 6.81$, bad $\geq 9.30$; B: good $\leq 3.70$, bad $\geq 10.24$; C: good $\leq 10.08$, bad $\geq 14.62$). Finally, \textbf{mean} $\vert \bm{\Delta}\text{\textbf{MAE}} \vert$ (m/c) is the mean absolute checkpoint-to-checkpoint change in validation MAE for masked and controlled, with lower values indicating calmer and more reliable validation behavior. Since the rotations have materially different target compositions, absolute MAEs should be compared \textit{within}, not across, rotations.}
\label{table.ex3}
\end{table}

Figure \ref{fig.sea_ice_trajectories} shows the validation sea ice concentration mean absolute error trajectories for the full, masked, and controlled learners across the three split rotations described in Appendix \ref{sec.app.exp3}. Even before turning to explicit occupation-time summaries, a clear qualitative pattern is visible: the uncontrolled masked learner is the most excursion-prone, whereas the controlled learner enters useful validation regimes earlier, remains there more stably, and exhibits fewer severe rebounds. This behavior is consistent across all three rotations despite their materially different target compositions. Moreover, in several early portions of training, the controlled learner not only behaves more coherently than masked SGD, but also occasionally descends faster than the idealized full-information baseline, echoing the beneficial conditioning effect already observed in Experiment \ref{sec.exp2}.

In this experiment, we deliberately do not treat masked training loss as an indicator of learning quality. Under heavy modality-level and within-modality missingness, that loss is evaluated through a censored observation pattern and therefore need not reveal whether the learner is recovering the full underlying mapping. We instead judge the methods on the metric that matters for that question: MAE on the clean, fully present validation set. From that perspective, the primary and expected comparison is against masked SGD, and Table \ref{table.ex3} shows that the controlled learner wins strongly in exactly that sense: in the strongest windows it outperforms masked SGD for most validation checkpoints, spends substantially more time in good regimes, spends substantially less time in bad-excursion regimes, and exhibits markedly calmer validation dynamics. Our conclusion is therefore not that the controlled learner must always beat the idealized full-information SGD run, but that under substantial modality-level and within-modality missingness it continues to learn the underlying mapping credibly and much more reliably than the uncontrolled masked alternative. When it also rivals or exceeds the fully observed reference over parts of training, we view that as a more remarkable, beyond-mandate effect, suggesting that the observability-aware geometry may be doing beneficial optimization work in addition to stabilizing missing-data learning.

\begin{figure}
    \centering
    \includegraphics[width=0.99\linewidth]{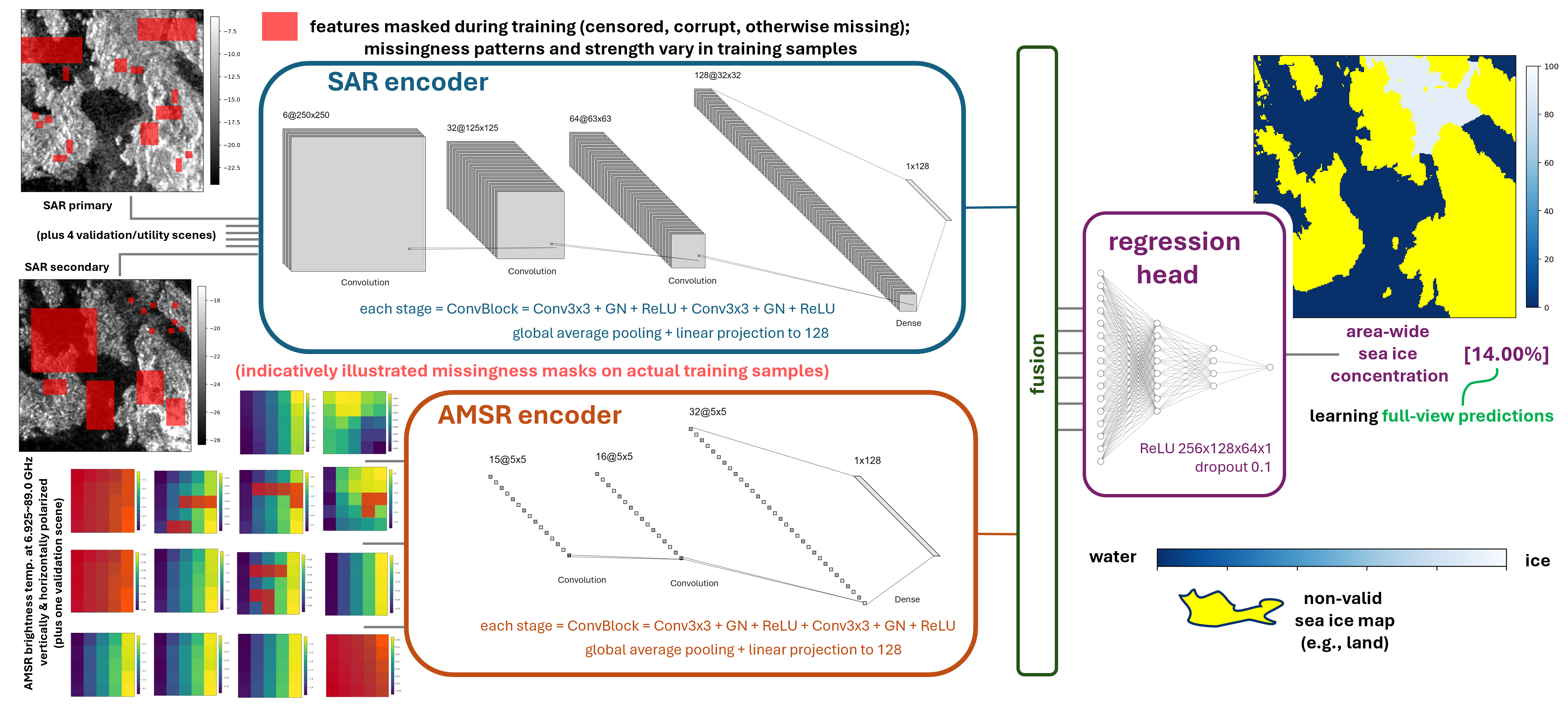}
    \caption{Two-branch convolutional regressor model employed in \textbf{Experiment \ref{sec.exp3}}: \textit{learning to predict sea ice concentration with remote sensing under missing sensing modalities and missing sensor measurements}.}
    \label{fig.ice.model}
\end{figure}

\begin{figure*}[b!]
    \includegraphics[width=\textwidth]{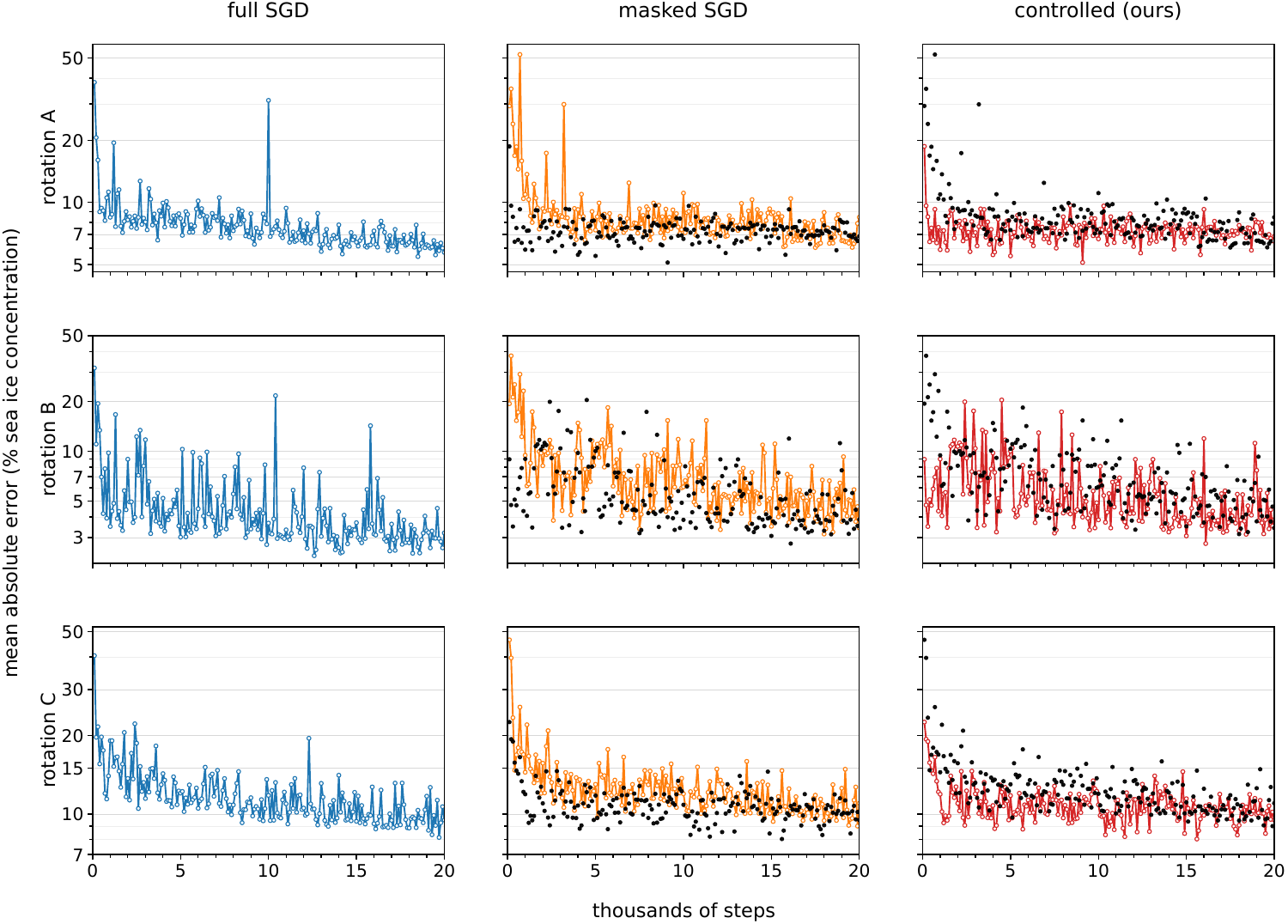}
    \caption{Validation trajectories from \textbf{Experiment \ref{sec.exp3}}: \textit{learning to predict sea ice concentration with remote sensing under missing sensing modalities and missing sensor measurements}, shown across three split rotations A/B/C for full SGD, masked SGD, and the controlled learner (ours). The controlled learner is consistently less excursion-prone than masked SGD, tends to enter useful error regimes earlier, and often remains there with fewer severe rebounds; Table \ref{table.ex3} quantifies this pattern through validation occupation-time and reliability statistics. \textbf{Note:} In the masked SGD column, the black dots correspond to the data points of the controlled run on the same rotation to facilitate an easier comparison. Conversely, in the controlled column, the black dots correspond to the data points of the masked run for the same rotation.}
    \label{fig.sea_ice_trajectories}
\end{figure*}

\section{Related work}

Existing approaches to (supervised) learning with missing data fall into five main categories, each of which we briefly contrast with our directional observability-aware framework in Sections \ref{sec.rw.imputation}--\ref{sec.rw.partial}. In summary, prior related work either reconstructs missing data, does not focus on the dynamical consequences of masking (i.e., from a trajectory stability viewpoint), or optimizes different robustness notions. By recasting learning as a closed-loop adaptive system whose controllability varies with the data mask, our directional observability-aware adaptation fills a critical gap: it provides provable trajectory stability and demonstrable coherence without imputation, distributional assumptions, or external regularization, which span the failure modes that remain unaddressed by existing techniques. In addition, we relate our developments to foundational control and estimation concepts in \ref{sec.rw.control}--\ref{sec.rw.estimation}.

\subsection{Imputation / EM / matrix completion ($\sim$ guess and fill the blanks, then learn)} \label{sec.rw.imputation}
These methods assume an underlying latent distribution and explicitly reconstruct missing entries (via generative models, expectation-maximization, or low-rank completion) before training \citep{emmanuel2021survey, bertsimas2018predictive, prakash2024benchmarking, rizvi2025analysis, racz2025comparison, duremasker}. While probabilistically elegant and sometimes statistically optimal under correct modeling assumptions, they focus exclusively on likelihood maximization and never examine the \textit{dynamic} evolution and stability of the parameter vector during intermittent missingness. In contrast, our approach never hallucinates or imputes missing values. Instead, it regulates parameter motion in an evidence-driven way, treating missingness as a structured loss of actuation authority rather than a statistical inference problem. 

It is also important to point out that in the missingness regimes that we have been using to illustrate our approach, with a very strong majority of features and the information they carry completely missing, even for prolonged blocks, we operate in a categorically different regime from a relatively mild, ``fill-a-few-blanks'' setting for which imputation methods are designed. In the consistently scarce data that may remain available, even from one iteration or batch to the next, there may be insufficient reliable information left to perform any meaningful imputation, without resorting to subjective guessing or outright hallucination, either of which can inject additional noise and exacerbate the very parameter drift and coherence loss we seek to prevent. In the context of our experiments, missing 4 out of 6 semantically distinct modalities plus intra-modality drops in Experiment \ref{sec.exp2}, or missing multiple entire SAR channels in Experiment \ref{sec.exp3} cannot be credibly remedied by imagining wearable signals or drawing hypothetical remote sensing patterns, respectively, just to learn a model of inevitably dubious characteristics and provenance. Our directional observability-aware controller sidesteps this issue entirely by never attempting to reconstruct anything missing, instead directly regulating parameter motion according to the empirically observed excitation.

\subsection{Masked losses / partial supervision (standard deep learning $\sim$ keep the blanks and learn)}
The standard practice in deep learning is to compute gradients solely on observed components (zeroing contributions from masked inputs or outputs), while leaving the optimizer unchanged \citep{ipsen2022deal, casella2024transformers, niu2025self}. This strategy is simple and scalable, yet offers zero safeguards against parameter drift in unobserved and empirically unsupported subspaces; it rather assumes that SGD will self-regularize. Our method \textit{explicitly} constrains, via a control-theoretic mechanism, the state evolution of the parameters, providing the missing closed-loop mechanism that attenuates incoherent motion.

Intentional masking can be beneficial in self-supervised pretraining, as demonstrated by masked autoencoders \citep{he2022masked}, where random patch masking can lead to learning robust representations. Nevertheless, this is fundamentally different from organically missing data during training, where missingness induces oscillations, stalling, or parameter hallucination. Our directional observability-aware controller directly mitigates these dynamical failure modes, in contrast to engineered masking designed and empirically observed to improve representation quality.

\subsection{Bayesian marginalization / uncertainty-aware models}
Bayesian and ensemble techniques maintain explicit posterior uncertainty over parameters and marginalize over missing data \citep{gal2016dropout, chen2025residual, zeng2026uncertainty, kendall2017uncertainties}. Although conceptually clean, they remain computationally heavy, model-dependent, and largely open-loop: posteriors can still drift arbitrarily when excitation weakens or entirely disappears for prolonged periods. Our framework delivers distribution-free boundedness guarantees on the parameter error trajectory without requiring a prior, offering a lightweight control-theoretic alternative that enforces learning coherence and stability.

\subsection{Robust optimization / DRO}
These approaches optimize for the worst-case data distribution within an ambiguity set, yielding strong adversarial guarantees \citep{kuhn2025distributionally, mohajerin2018data, ghosal2023distributionally}. However, they typically target distributional shifts or corrupted samples, not the structured loss of controllability induced by missing inputs or outputs. In contrast, we model missingness explicitly as a directional reduction in actuation authority, thereby protecting a fundamentally different axis of robustness -- one that standard DRO formulations do not capture.

\subsection{Online learning with partial feedback (bandits, etc.)} \label{sec.rw.partial}
Bandit-style and partial-observation frameworks (including contextual bandits and online convex optimization with bandit feedback) provide regret bounds when only partial loss information is available \citep{shao2026partial, alon2015online, raman2024multiclass, lattimore2020bandit, auer2002finite}. These results are elegant and long-standing, yet, they focus on cumulative performance and external regret, not on the internal stability or boundedness of the parameter vector itself from one iteration to the next. Our contribution complements this line of work by ensuring the learner's state remains sane under partial observability, a dynamical guarantee that remains distinct from asymptotic average-case analyses.

\subsection{Adaptive control and Lyapunov/ISS theory} \label{sec.rw.control} The stability tools we employ draw from classical nonlinear control theory, where the input-to-state stability (ISS) framework and the associated Lyapunov-based analyses are ubiquitous \citep{khalil2002nonlinear, haddad2008nonlinear}. More specifically, dissipation-mismatch inequalities and boundedness arguments that are routinely invoked in control-theoretic analyses form the foundations for our windowed residual-to-state bound (Proposition \ref{proposition.ISS}) and for the eigenvalue properties of our observability Gramian. Our key contribution is repurposing such control theoretic tools and workflows and aligning them with the specific setting of learning under partial observability. We close the loop around a data-driven observability metric that evolves online, yielding provable trajectory bounds for gradient-based learning under the stated finite-dimensional or active-subspace excitation and bounded-residual-mismatch conditions when inputs are intermittently missing.

\subsection{Parameter estimation under intermittent measurements in networked control systems} \label{sec.rw.estimation} A superficially related line of work studies parameter estimation and state observation under intermittent or packet-dropped measurements in networked control systems, where Gramian-type matrices are used to update observer gains against a known system model \citep{sinopoli2004kalman, schenato2007foundations, gupta2009data}. Although our observability matrix $\bm{B}_t$ shares a structural resemblance with these Gramians, their roles differ fundamentally. In adaptive observers, the system model is known a priori and the innovation term remains cleanly separated from the observability loss. In our framework, no such separation exists: the gradient $\nabla_{\bm{w}} \ell_t$ itself is an uncertain measurement whose conditional structure is induced by the same masked Jacobian channel that drives $\bm{B}_t$. Consequently, the inverse $\bm{B}_t^{-1}$ performs a dual role: simultaneously reshaping the geometry of the admissible parameter motion \textit{and} conditioning gradient information relative to recent excitation. This mechanism in our approach, which has no direct analog in the observer literature, is what enables the residual-to-state guarantees of Proposition \ref{proposition.ISS} in a gradient-based learning setting and distinguishes our contribution from classical parameter estimation under measurement dropout.

\section{Conclusion} \label{sec.conclusion}

Even though contemporary accomplishments in machine learning may suggest that models, \textit{somehow}, learn, and learn well, the assumptions behind such perspectives can be an unattainable luxury in many workflows of practical significance which do not enjoy data consistency, repeated exposure to informative features, or forgiving data geometry. We position our work and contributions not as a performance optimization but at the workflow-enabling level, as our directional observability-aware learning adaptation can enable learning  in pathologically sparse domains, which are abundant in the real world. Data may arrive intermittently, observations may be censored, delayed, or missing by design, and even interactions with real-world systems that yield such disappointing data may be expensive and otherwise scarce. Infrastructure, cyber-physical security, high-dimensional multimodal problems, financial or otherwise sensitive contexts routinely give rise to such complications. Our learning approach, by leveraging concepts and machinery from dynamical and control systems, provides the necessary evidence-driven adaptation that can render learning in such environments \textit{coherent}, \textit{dynamically stable}, \textit{economical} (in the sense of being able to credibly leverage partially available data streams without requiring persistent perfection and completeness) and ultimately \textit{viable}.

The explicit constants in Proposition \ref{proposition.ISS} are intended as analytical certificates rather than literal plug-in tuning prescriptions. Their role is to close the correctness chain and expose the expected qualitative dependencies: stronger effective excitation improves the guaranteed contraction margin, slower variation of the observability metric strengthens the transfer from raw to effective excitation, and overly aggressive step sizes erode the finite-window coercivity margin. As is typical of Lyapunov and ISS analyses, repeated worst-case norm and Young-inequality bounds make the resulting numerical values conservative. Practical choices of $\eta$, $\beta$, and $\epsilon$ should therefore be informed by these directional trends and validated empirically, rather than selected by literal substitution into the worst-case constants.

Although the present study focuses on supervised learning, the underlying mechanism operates at the level of gradient modulation under partial observability and is therefore not inherently tied to supervised objectives. This opens the door to extensions in streaming unsupervised and self-supervised settings, where temporally correlated missingness and local observability limitations are often even more central. Our individual technical results and the approach in its entirety can mechanistically extend to the multivariate output case, beyond the scalar case that we considered here for simplicity and brevity. Given the broader evolution of the learning frontier and the results that we observed in our experiments, we believe that an interesting and natural future direction pertains to the systematic empirical integration and evaluation of our directionally-aware adaptation in impactful model architectures, such as autoencoders and transformers, the construction of which in the real world is inherently driven by overall imperfect data streams.

\section*{Acknowledgements}

This work was funded by the U.S. Department of Energy through contract DE-AC07-05ID14517 and, in part, by the U.S. Department of Homeland Security Science and Technology Directorate and the ADAC-ARCTIC Center of Excellence managed by the University of Alaska Anchorage. Neither the U.S. Government nor any agency thereof, nor Contractor, nor any of their employees, make any warranty, express or implied, or assume any legal liability or responsibility for the accuracy, completeness, or usefulness of any information, apparatus, product, or process disclosed. The views and opinions of authors expressed herein do not necessarily state or reflect those of the U.S. Government or any agency or Contractor thereof.

\section*{Impact Statement}

We show theoretically and empirically how observability-aware, control-theoretic adaptation can promote coherent and dynamically stable learning under severe and intermittent data missingness.

\bibliography{references}

@article{emmanuel2021survey,
  title={A survey on missing data in machine learning},
  author={Emmanuel, Tlamelo and Maupong, Thabiso and Mpoeleng, Dimane and Semong, Thabo and Mphago, Banyatsang and Tabona, Oteng},
  journal={Journal of Big Data},
  volume={8},
  number={1},
  pages={140},
  year={2021},
  publisher={Springer}
}

@article{bertsimas2018predictive,
  title={From predictive methods to missing data imputation: an optimization approach},
  author={Bertsimas, Dimitris and Pawlowski, Colin and Zhuo, Ying Daisy},
  journal={Journal of Machine Learning Research},
  volume={18},
  number={196},
  pages={1--39},
  year={2018}
}

@article{prakash2024benchmarking,
  title={Benchmarking machine learning missing data imputation methods in large-scale mental health survey databases},
  author={Prakash, Preethi and Street, Kelly and Narayanan, Shrikanth and Fernandez, Bridget A and Shen, Yufeng and Shu, Chang},
  journal={Artificial Intelligence in Health},
  volume={2},
  number={1},
  pages={81--92},
  year={2024},
  publisher={AccScience Publishing}
}

@article{rizvi2025analysis,
  title={Analysis of machine learning based imputation of missing data},
  author={Rizvi, Syed Tahir Hussain and Latif, Muhammad Yasir and Amin, Muhammad Saad and Telmoudi, Achraf Jabeur and Shah, Nasir Ali},
  journal={Cybernetics and Systems},
  volume={56},
  number={6},
  pages={818--832},
  year={2025},
  publisher={Taylor \& Francis}
}

@article{racz2025comparison,
  title={Comparison of missing value imputation tools for machine learning models based on product development cases studies},
  author={R{\'a}cz, Anita and Gere, Attila},
  journal={LWT - Food Science and Technology},
  volume={221},
  pages={117585},
  year={2025},
  publisher={Elsevier}
}

@article{casella2024transformers,
  title={Transformers deep learning models for missing data imputation: an application of the {ReMasker} model on a psychometric scale},
  author={Casella, Monica and Milano, Nicola and Dolce, Pasquale and Marocco, Davide},
  journal={Frontiers in Psychology},
  volume={15},
  pages={1449272},
  year={2024},
  publisher={Frontiers Media SA}
}

@article{kuhn2025distributionally,
  title={Distributionally robust optimization},
  author={Kuhn, Daniel and Shafiee, Soroosh and Wiesemann, Wolfram},
  journal={Acta Numerica},
  volume={34},
  pages={579--804},
  year={2025},
  publisher={Cambridge University Press}
}

@article{shao2026partial,
  title={Partial Feedback Online Learning},
  author={Shao, Shihao and Fang, Cong and Lin, Zhouchen and Tao, Dacheng},
  journal={arXiv preprint arXiv:2601.21462},
  year={2026}
}

@inproceedings{alon2015online,
  title={Online learning with feedback graphs: Beyond bandits},
  author={Alon, Noga and Cesa-Bianchi, Nicolo and Dekel, Ofer and Koren, Tomer},
  booktitle={Conference on Learning Theory},
  pages={23--35},
  year={2015},
  organization={PMLR}
}

@inproceedings{raman2024multiclass,
  title={Multiclass online learnability under bandit feedback},
  author={Raman, Ananth and Raman, Vinod and Subedi, Unique and Mehalel, Idan and Tewari, Ambuj},
  booktitle={International Conference on Algorithmic Learning Theory},
  pages={997--1012},
  year={2024},
  organization={PMLR}
}

@book{lattimore2020bandit,
  title={Bandit algorithms},
  author={Lattimore, Tor and Szepesv{\'a}ri, Csaba},
  year={2020},
  publisher={Cambridge University Press}
}

@article{auer2002finite,
  title={Finite-time analysis of the multiarmed bandit problem},
  author={Auer, Peter and Cesa-Bianchi, Nicolo and Fischer, Paul},
  journal={Machine learning},
  volume={47},
  number={2},
  pages={235--256},
  year={2002},
  publisher={Springer}
}

@article{mohajerin2018data,
  title={Data-driven distributionally robust optimization using the Wasserstein metric: Performance guarantees and tractable reformulations},
  author={Mohajerin Esfahani, Peyman and Kuhn, Daniel},
  journal={Mathematical Programming},
  volume={171},
  number={1},
  pages={115--166},
  year={2018},
  publisher={Springer}
}

@inproceedings{ghosal2023distributionally,
  title={Distributionally robust optimization with probabilistic group},
  author={Ghosal, Soumya Suvra and Li, Yixuan},
  booktitle={Proceedings of the AAAI Conference on Artificial Intelligence},
  volume={37},
  number={10},
  pages={11809--11817},
  year={2023}
}

@inproceedings{gal2016dropout,
  title={Dropout as a Bayesian approximation: Representing model uncertainty in deep learning},
  author={Gal, Yarin and Ghahramani, Zoubin},
  booktitle={International Conference on Machine Learning},
  pages={1050--1059},
  year={2016},
  organization={PMLR}
}

@article{chen2025residual,
  title={Residual bayesian attention networks for uncertainty quantification in regression tasks},
  author={Chen, Youliang and Guan, Wencan and Azzam, Rafig},
  journal={Scientific Reports},
  volume={15},
  number={1},
  pages={38279},
  year={2025},
  publisher={Nature Publishing Group UK London}
}

@article{zeng2026uncertainty,
  title={Uncertainty-Aware Bayesian Time Series framework for probabilistic imputation},
  author={Zeng, Zi-Xuan and Tan, Yan-Ke and Deng, E and Ni, Yi-Qing and Zhang, Qi-Lin and Yang, Bin},
  journal={Advanced Engineering Informatics},
  volume={70},
  pages={104153},
  year={2026},
  publisher={Elsevier}
}

@article{kendall2017uncertainties,
  title={What uncertainties do we need in {B}ayesian deep learning for computer vision?},
  author={Kendall, Alex and Gal, Yarin},
  journal={Advances in Neural Information Processing Systems},
  volume={30},
  year={2017}
}

@inproceedings{he2022masked,
  title={Masked autoencoders are scalable vision learners},
  author={He, Kaiming and Chen, Xinlei and Xie, Saining and Li, Yanghao and Doll{\'a}r, Piotr and Girshick, Ross},
  booktitle={Proceedings of the IEEE/CVF Conference on Computer Vision and Pattern Recognition},
  pages={16000--16009},
  year={2022}
}

@inproceedings{ipsen2022deal,
  title={How to deal with missing data in supervised deep learning?},
  author={Ipsen, Niels Bruun and Mattei, Pierre-Alexandre and Frellsen, Jes},
  booktitle={10th International Conference on Learning Representations},
  year={2022}
}

@article{niu2025self,
  title={A self-supervised masked spatial distribution learning method for predicting machinery remaining useful life with missing data reconstruction},
  author={Niu, Ben and Xiao, Yi and Xiao, Qinge and Liu, Yang and Peng, Tao and Yang, Zhile},
  journal={Advanced Engineering Informatics},
  volume={64},
  pages={102938},
  year={2025},
  publisher={Elsevier}
}

@book{haddad2008nonlinear,
  title={Nonlinear dynamical systems and control: a Lyapunov-based approach},
  author={Haddad, Wassim M and Chellaboina, VijaySekhar},
  year={2008},
  publisher={Princeton University Press}
}

@book{khalil2002nonlinear,
  title={Nonlinear systems},
  author={Khalil, Hassan K},
  edition={3},
  year={2002},
  publisher={Prentice Hall},
  address = {Upper Saddle River, NJ}
}

@inproceedings{duremasker,
  title={ReMasker: Imputing Tabular Data with Masked Autoencoding},
  author={Du, Tianyu and Melis, Luca and Wang, Ting},
  booktitle={12th International Conference on Learning Representations},
  year={2024}
}

@article{sinopoli2004kalman,
  title={Kalman filtering with intermittent observations},
  author={Sinopoli, Bruno and Schenato, Luca and Franceschetti, Massimo and Poolla, Kameshwar and Jordan, Michael I and Sastry, Shankar S},
  journal={IEEE Transactions on Automatic Control},
  volume={49},
  number={9},
  pages={1453--1464},
  year={2004},
  publisher={IEEE}
}

@article{schenato2007foundations,
  title={Foundations of control and estimation over lossy networks},
  author={Schenato, Luca and Sinopoli, Bruno and Franceschetti, Massimo and Poolla, Kameshwar and Sastry, S Shankar},
  journal={Proceedings of the IEEE},
  volume={95},
  number={1},
  pages={163--187},
  year={2007},
  publisher={IEEE}
}

@article{gupta2009data,
  title={Data transmission over networks for estimation and control},
  author={Gupta, Vijay and Dana, Amir F and Hespanha, Joao P and Murray, Richard M and Hassibi, Babak},
  journal={IEEE Transactions on Automatic Control},
  volume={54},
  number={8},
  pages={1807--1819},
  year={2009},
  publisher={IEEE}
}

@article{malmgren2020convolutional,
  title={A convolutional neural network architecture for {Sentinel}-1 and {AMSR2} data fusion},
  author={Malmgren-Hansen, David and Pedersen, Leif Toudal and Nielsen, Allan Aasbjerg and Kreiner, Matilde Brandt and Saldo, Roberto and Skriver, Henning and Lavelle, John and Buus-Hinkler, J{\o}rgen and Krane, Klaus Harnvig},
  journal={IEEE Transactions on Geoscience and Remote Sensing},
  volume={59},
  number={3},
  pages={1890--1902},
  year={2020},
  publisher={IEEE}
}

@misc{human_activity_recognition_using_smartphones_240,
  author       = {Reyes-Ortiz, Jorge and Anguita, Davide and Ghio, Alessandro and Oneto, Luca and Parra, Xavier},
  title        = {{Human Activity Recognition Using Smartphones}},
  year         = {2013},
  howpublished = {UCI Machine Learning Repository},
  note         = {{DOI}: https://doi.org/10.24432/C54S4K}
}

@inproceedings{anguita2013public,
  title={A public domain dataset for human activity recognition using smartphones.},
  author={Anguita, Davide and Ghio, Alessandro and Oneto, Luca and Parra, Xavier and Reyes-Ortiz, Jorge Luis and others},
  booktitle={Esann},
  volume={3},
  number={1},
  pages={3--4},
  year={2013}
}

@misc{asip,
    author = "Roberto Saldo and Matilde Brandt Kreiner and Jørgen Buus-Hinkler and Leif Toudal Pedersen and David Malmgren-Hansen and Allan Aasbjerg Nielsen and Henning Skriver",
    title = "{AI4Arctic / ASIP Sea Ice Dataset - version 2}",
    year = "2020",
    month = "10",
    url = "https://data.dtu.dk/articles/dataset/AI4Arctic_ASIP_Sea_Ice_Dataset_-_version_2/13011134",
    doi = "10.11583/DTU.13011134.v3",
}
\bibliographystyle{tmlr}

\appendix
\onecolumn
\section*{APPENDIX}

\section{Pedagogical numerical illustration of the partial-actuation failure mode}
\label{sec.app.demo}

We present a minimal, $d = 2$, blackboard-style linear regression example that makes the learning instability under missing data concrete. In the linear regression case, the failure mode arises as the scalar residual couples coordinates, both observed and masked, indiscriminately, and moves the observed parameters regardless of the true values of the masked features and the update direction and strength that the full feature vector would have led to. We configure our example as follows:
\begin{align*}
 \text{\textbf{features:}}& \quad \bm{x}_t \in \mathbb{R}^2 \;\text{\textbf{(full)}}, \quad \tilde{\bm{x}}_t = \left\{ \left[ \begin{array}{c} x_{1,t} \\ x_{2,t} \end{array} \right] , \left[ \begin{array}{c} 0 \\ x_{2,t} \end{array} \right],\left[ \begin{array}{c} x_{1,t} \\ 0 \end{array} \right] \right\} \; \text{\textbf{(masked)}} \\
 \text{\textbf{true model:}}& \quad y_t = \langle \bm{w}^*, \bm{x}_t \rangle, \text{where } \bm{w}^* = \left[ \begin{array}{c} w_1^* \\ w_2^* \end{array} \right] = \left[ \begin{array}{c} 1 \\ 1 \end{array} \right] \\
 \text{\textbf{method:}}& \quad \bm{w}_{t+1} = \bm{w}_t + \eta \tilde{\bm{x}}_t \underbrace{ \left( y_t - \langle \bm{w}_t, \tilde{\bm{x}}_t \rangle \right)}_{\text{residual }
 {r}_t} \\
 \text{\textbf{learning rate:}}& \quad \eta = 0.05 \\
 \text{\textbf{initialization:}}& \quad \bm{w}_0 =  \left[ \begin{array}{c} 0 \\ 0 \end{array} \right] \\
 \text{\textbf{progression:}}& \quad \text{[step 0] }\bm{x}_0 = [4,4]^\mathsf{T}, y_0 = 8, \text{\textbf{feature 2 missing}} \to \tilde{\bm{x}}_0 = [4,0]^\mathsf{T} \\
 & \quad \text{[step 1] } \bm{x}_1 = [3,5]^\mathsf{T}, y_1 = 8, \text{\textbf{feature 1 missing}} \to \tilde{\bm{x}}_1 = [0,5]^\mathsf{T} \\
 & \quad \text{[step 2] } \bm{x}_2 = [5,3]^\mathsf{T}, y_2 = 8, \text{\textbf{feature 1 missing} (again)} \to \tilde{\bm{x}}_2 = [0,3]^\mathsf{T}
\end{align*}
To illustrate how the considered failure mode can emerge, we study the evolution of a \textbf{full SGD} method starting from $\bm{w}_0$ and operating on the above full-feature samples, a \textbf{masked SGD} method operating on the above masked samples, and a \textbf{counterfactual} step that picks up from wherever the masked SGD method is and applies a full feature visibility update as if it were the full SGD method.

\renewcommand{\arraystretch}{1.2}
\begin{tabular}{|r|c|c|c|}
\hline
\cellcolor[HTML]{000000}{\color[HTML]{FFFFFF} \textbf{STEP 0}} & \textbf{full SGD} & \textbf{masked SGD} & \textbf{counterfactual} \\ \hline
starting $\bm{w}$ & $\bm{w}_0 = [0,0]^\mathsf{T}$ & $\bm{w}_0 = [0,0]^\mathsf{T}$ &  $\bm{w}_0 = [0,0]^\mathsf{T}$ \\ \hline
sample $\bm{x}$ used & $\bm{x}_0 = [4,4]^\mathsf{T}$ & $\tilde{\bm{x}}_0 = [4,0]^\mathsf{T}$ & $\bm{x}_0 = [4,4]^\mathsf{T}$ \\ \hline
actual $y$ & 8 & 8 & 8 \\ \hline
predicted $\hat{y}$ & 0 & 0 & 0 \\ \hline
residual $r$ & 8 & 8 & 8 \\ \hline
$\Delta \bm{w}$  & $[1.6, 1.6]^\mathsf{T}$  & $[1.6,0]^\mathsf{T}$ & $[1.6, 1.6]^\mathsf{T}$ \\ \hline
new parameter $\bm{w}_1$ &  $\bm{w}_1^\mathrm{f} = [1.6, 1.6]^\mathsf{T}$ & $\bm{w}_1^\mathrm{m} = [1.6, 0]^\mathsf{T}$  &  $[1.6, 1.6]^\mathsf{T}$ \\ \hline
new parameter error $\bm{e}_1$ & $[0.6, 0.6]^\mathsf{T}$ & $[0.6, -1.0]^\mathsf{T}$  & $[0.6, 0.6]^\mathsf{T}$ \\ \hline
new parameter-error norm $\| \bm{e}_1 \|$ & $0.849$ & $1.166$ & $0.849$ \\ \hline
\end{tabular}

All three methods start from $[0,0]^\mathsf{T}$, so the counterfactual update is identical to full SGD here. Coincidentally, their predicted $\hat{y}$ is identical. Full SGD, with both features, immediately corrects both components. Masked SGD correctly updates $\bm{w}_1$, but leaves $w_{2,1}$ frozen at $0$, one full unit below $w_2^* = 1$. The divergence already emerges: $w_{1,1}$ moderately overshot at $1.6$ while $w_{2,1}$ is left entirely unactuated. However, this is the canonical behavior of masked SGD in the linear regression case: $w_{2,1}$ expectedly stays put. A masked feature and the part(s) of the model it influences (which, in the case of linear regression, have a very direct, one-to-one correspondence) cannot help the convergence or stability of the method. Yet, as we will see next, a masked feature/component can yield a destabilizing failure that is realized by the motion of currently visible and actuated components. Moreover, effects such as the larger parameter error seen above in the masked method can accumulate and propagate in subsequent iterations.

\begin{tabular}{|r|c|c|c|}
\hline
\cellcolor[HTML]{000000}{\color[HTML]{FFFFFF} \textbf{STEP 1}} & \textbf{full SGD} & \textbf{masked SGD} & \textbf{counterfactual} \\ \hline
starting $\bm{w}$ & $\bm{w}_1^\mathrm{f} = [1.6,1.6]^\mathsf{T}$ & $\bm{w}_1^\mathrm{m} = [1.6,0]^\mathsf{T}$ &  $\bm{w}_1^\mathrm{m} = [1.6,0]^\mathsf{T}$ \\ \hline
sample $\bm{x}$ used & $\bm{x}_1 = [3,5]^\mathsf{T}$ & $\tilde{\bm{x}}_1 = [0,5]^\mathsf{T}$ & $\bm{x}_1 = [3,5]^\mathsf{T}$ \\ \hline
actual $y$ & 8 & 8 & 8 \\ \hline
predicted $\hat{y}$ & 12.8 & 0 & 4.8 \\ \hline
residual $r$ & $\bm{-4.8}$ & $\bm{+8.0}$ & $\bm{+3.2}$ \\ \hline
$\Delta \bm{w}$  & $[-0.72, -1.2]^\mathsf{T}$  & $[0,+2.0]^\mathsf{T}$ & $[+0.48, +0.80]^\mathsf{T}$ \\ \hline
new parameter $\bm{w}_2$ &  $\bm{w}_2^\mathrm{f} = [0.88, 0.40]^\mathsf{T}$ & $\bm{w}_2^\mathrm{m} = [1.6, 2.0]^\mathsf{T}$  &  $[2.08, 0.80]^\mathsf{T}$ \\ \hline
new parameter error $\bm{e}_2$ & $[-0.12, -0.60]^\mathsf{T}$ & $[0.6, 1.0]^\mathsf{T}$  & $[1.08, -0.20]^\mathsf{T}$ \\ \hline
new parameter-error norm $\| \bm{e}_2 \|$ & $\bm{0.612}$ & $\bm{1.166}$ & $\bm{1.098}$ \\ \hline
\end{tabular}

We notice above a significant difference between the \textit{masked residual} ($+8.0$) and the\textit{ counterfactual residual} ($+3.2$): the $4.8$ unit gap equals $w_{1,1}^\mathrm{masked} \cdot x_{1,1} = 1.6 \cdot 3 = 4.8$, which corresponds to the invisible prediction contribution from the previously overshot $w_{1,1}^\mathrm{masked}$ flowing through the missing feature channel. This gap itself decomposes as:
\begin{equation*}
    4.8 = w_{1,1}^\mathrm{m} \cdot x_{1,1} = \underbrace{w_1^* \cdot x_{1,1}}_{=3} + \underbrace{\left( w_{1,1}^\mathrm{m} - w_1^* \right) \cdot x_{1,1}}_{=1.8}
\end{equation*}
The term $w_1^* \cdot x_{1,1}$ is exogenous and present even with a perfect model, whereas the term $\left( w_{1,1} - w_1^* \right) \cdot x_{1,1}$ is endogenous and represents the learner's own past error in $w_{1,1}$ re-entering, yet in a rather detrimental way through the hidden channel. The masked learner uses the inflated $8$ units of residual to drive the update of the observable subspace, which corresponds to the component $w_{2,1}$, whereas the unobservable subspace, $w_{1,1}$, stays put, by definition. The inflated residual is then driving $w_{2,1}=0$ to $w_{2,2} = 2.0$ in a single step, overshooting $w_{2}^*$ by a factor of two. 

A striking observation in this step is that the error norm $\|\bm{e}_2 \| = 1.166$ for masked remains exactly the same as $\|\bm{e}_1 \|$ at the conclusion of step 0. The error vector did not shrink; instead, it rotated from $[0.6,-1.0]$ to $[0.6, +1.0]^\mathsf{T}$. The magnitude of the error in $w_2$ is unchanged at $1.0$, but its sign flipped: the learner moved from being $1.0$ unit below $w_2^*$ to being $1.0$ unit above it, with complete confidence. This rotation at constant norm is a rather clean signature of learning incoherence: one full step of gradient descent, in the simplest context of 2-dimensional linear regression, nothing was gained, yet the representation in the view of the learner, as described by $\bm{w}_2$, changes monumentally.

\begin{tabular}{|r|c|c|c|}
\hline
\cellcolor[HTML]{000000}{\color[HTML]{FFFFFF} \textbf{STEP 2}} & \textbf{full SGD} & \textbf{masked SGD} & \textbf{counterfactual} \\ \hline
starting $\bm{w}$ & $\bm{w}_2^\mathrm{f} = [0.88,0.40]^\mathsf{T}$ & $\bm{w}_2^\mathrm{m} = [1.6,2.0]^\mathsf{T}$ &  $\bm{w}_2^\mathrm{m} = [1.6,2.0]^\mathsf{T}$ \\ \hline
sample $\bm{x}$ used & $\bm{x}_2 = [5,3]^\mathsf{T}$ & $\tilde{\bm{x}}_2 = [0,3]^\mathsf{T}$ & $\bm{x}_2 = [5,3]^\mathsf{T}$ \\ \hline
actual $y$ & 8 & 8 & 8 \\ \hline
predicted $\hat{y}$ & 5.6 & 6.0 & 14.0 \\ \hline
residual $r$ & $\bm{+2.4}$ & $\bm{+2.0}$ & $\bm{-6.0}$ \\ \hline
$\Delta \bm{w}$  & $[+0.60, +0.36]^\mathsf{T}$  & $[0,+0.30]^\mathsf{T}$ & $[-1.50,-0.90]^\mathsf{T}$ \\ \hline
new parameter $\bm{w}_3$ &  $\bm{w}_3^\mathrm{f} = [1.48, 0.76]^\mathsf{T}$ & $\bm{w}_3^\mathrm{m} = [1.6, 2.30]^\mathsf{T}$  &  $[0.10, 1.10]^\mathsf{T}$ \\ \hline
new parameter error $\bm{e}_3$ & $[0.48, -0.24]^\mathsf{T}$ & $[0.6, 1.30]^\mathsf{T}$  & $[-0.90, 0.10]^\mathsf{T}$ \\ \hline
new parameter-error norm $\| \bm{e}_3 \|$ & $\bm{0.537}$ & $\bm{1.432}$ & $\bm{0.906}$ \\ \hline
\end{tabular}

In the final step that we examine above, the error norm for masked SGD grows from $1.166$ to $1.432$. The masked update machinery constructs a severely inflated residual, which is misinforming and misleading the learner. The inflation is so high that the residual flips sign from what the counterfactual update would work with. As a result, the masked method moves in the opposite direction than what would help regulate the error. The contrast between the different residuals also illuminates another point: the masked method, seeing a relatively smaller residual yet of the same sign as that in the preceding step, may ``think'' that it is close to converging to the sought-after $\bm{w}^*$, which could not be further from the truth.

\paragraph{A scalar observability-aware replay of the critical step.}
The scalar controller of Section \ref{sec.prototypic_update} does not recover the missing contribution or correct the sign of the residual. It instead limits how strongly the learner may act on that residual. To make this distinction concrete, consider the same starting state and masked sample as in the final column above:
\begin{equation*}
\bm{w}_2^{\mathrm{m}}=[1.6,2.0]^\mathsf{T},\qquad \tilde{\bm{x}}_2=[0,3]^\mathsf{T},\qquad \tilde r_2=2.
\end{equation*}
With $\epsilon=1$, the scalar gain is
\begin{equation*}
\alpha_2 = \frac{\eta}{\epsilon+\|\tilde{\bm{x}}_2\|^2} =\frac{0.05}{1+9} = 0.005.
\end{equation*}
The resulting update is therefore
\begin{equation*}
\Delta\bm{w}_2^{\mathrm{scalar}} = \alpha_2\tilde{\bm{x}}_2\tilde r_2 =[0,0.03]^\mathsf{T},
\end{equation*}
rather than the masked-SGD update $[0,0.30]^\mathsf{T}$. Thus,
\begin{equation*}
\bm{w}_3^{\mathrm{scalar}}=[1.6,2.03]^\mathsf{T} \qquad \|\bm{e}_3^{\mathrm{scalar}}\|=1.191, 
\end{equation*}
compared with $\|\bm{e}_3^{\mathrm{masked}}\|=1.432$. The scalar controller has not inferred the unavailable feature contribution: the proposed motion remains imperfect and the error norm still grows slightly from $1.166$. It has, however, prevented the misleading residual from producing a full-strength update, substantially limiting the resulting excursion.

\paragraph{Directional observability-aware learning.} We do \textit{not} replay the directional controller of Section \ref{sec.directional} in this three-step illustration in $\mathbb{R}^2$: its distinguishing mechanism depends on a sufficiently rich history of heterogeneous sensitivity directions, as well as a nontrivial parameter space, whereas the axis-aligned masked samples used here produce only a short, diagonal actuation history. That mechanism is therefore studied over sustained
learning runs in Section \ref{sec.exp1}.

\section{Proofs}

\subsection{Proof of Lemma \ref{lemma.Bt_eigenvalue_bounds}} \label{app.proof.lemma_Bt_eigenvaluePbounds}

Unrolling the empirical EMA state in (\ref{eq.directional.B_recursive}) gives
\begin{equation}
    \bm{S}_t=(1-\beta)\sum_{j=0}^{t}\beta^j\bm{g}_{t-j}\bm{g}_{t-j}^\mathsf{T}.
\end{equation}
Since $\|\bm{g}_{t-j}\|\leq G$, we have $\bm{g}_{t-j}\bm{g}_{t-j}^\mathsf{T}\preceq G^2I$, and therefore
\begin{align}
    \bm{B}_t &=\epsilon I+\bm{S}_t, \nonumber\\
    &\preceq \epsilon I+(1-\beta)\sum_{j=0}^{t}\beta^jG^2I \preceq (\epsilon+G^2)I =\overline{b}I.
\end{align}
Because $\bm{S}_t\succeq0$, we also have
\begin{equation}
    \bm{B}_t\succeq\epsilon I=\underline{b}I.
\end{equation}
For $t\geq T-1$, Assumption \ref{assumption_window} yields the sharper lower bound
\begin{align}
    \bm{B}_t &\succeq \epsilon I+(1-\beta)\sum_{j=0}^{T-1}\beta^j\bm{g}_{t-j}\bm{g}_{t-j}^\mathsf{T}, \nonumber\\
    &\succeq \epsilon I+(1-\beta)\beta^{T-1} \sum_{j=0}^{T-1}\bm{g}_{t-j}\bm{g}_{t-j}^\mathsf{T}, \nonumber\\
    &\succeq \left(\epsilon+(1-\beta)\beta^{T-1}\lambda_g\right)I.
\end{align}
The inverse bounds follow from order reversal under inversion. Finally,
\begin{equation}
    \|\bm{\psi}_t\| = \|\bm{B}_t^{-1}\bm{g}_t\| \leq \|\bm{B}_t^{-1}\|\|\bm{g}_t\|  \leq \frac{G}{\underline{b}},
\end{equation}
which concludes the proof.

\subsection{Proof of Lemma \ref{lemma.raw_to_effective_excitation}} \label{app.proof.raw_to_effective_excitation}

Fix $t\geq T-1$ and define the raw and effective sensitivity matrices
\begin{equation}
    \bm{\mathcal{G}}_t  := \left[ \bm{g}_{t-T+1}\;\cdots\;\bm{g}_{t} \right],  \qquad  \bm{\Psi}_t  :=  \left[ \bm{P}_{t-T+1}\bm{g}_{t-T+1}\;\cdots\;\bm{P}_{t}\bm{g}_{t}  \right].
\end{equation}
Assumption \ref{assumption_window} gives
\begin{equation}
    \bm{\mathcal{G}}_t\bm{\mathcal{G}}_t^\mathsf{T}\succeq\lambda_gI,
\end{equation}
and therefore
\begin{equation}
    \sigma_{\min}(\bm{\mathcal{G}}_t)\geq\sqrt{\lambda_g}.
\end{equation}
Write
\begin{equation}
    \bm{\Psi}_t = \bm{P}_t\bm{\mathcal{G}}_t+\bm{E}_t,
\end{equation}
where the $k$th column of $\bm{E}_t$ is
\begin{equation}
    \left(\bm{P}_{t-T+k}-\bm{P}_t\right)\bm{g}_{t-T+k}, \qquad k=1,\dots,T.
\end{equation}
From Lemma \ref{lemma.Bt_eigenvalue_bounds},
\begin{equation}
    \sigma_{\min}(\bm{P}_t)\geq\frac{1}{\overline{b}},
\end{equation}
so
\begin{equation}
    \sigma_{\min}\left(\bm{P}_t\bm{\mathcal{G}}_t\right)  \geq \frac{\sqrt{\lambda_g}}{\overline{b}}.
\end{equation}
Moreover,
\begin{align}
    \|\bm{E}_t\|  &\leq \|\bm{E}_t\|_{\mathrm{F}}, \nonumber\\
    &\leq G \left( \sum_{k=0}^{T-1}  \|\bm{P}_{t-k}-\bm{P}_t\|^2  \right)^{1/2}, \nonumber\\
    &\leq G\nu_T.
\end{align}
The singular-value perturbation inequality now yields
\begin{equation}
    \sigma_{\min}(\bm{\Psi}_t)  \geq \frac{\sqrt{\lambda_g}}{\overline{b}}-G\nu_T.
\end{equation}
Under condition (\ref{eq.metric_variation_condition}), the right-hand side is positive. Hence
\begin{align}
    \sum_{k=0}^{T-1}\bm{\psi}_{t-k}\bm{\psi}_{t-k}^\mathsf{T}  &=  \bm{\Psi}_t\bm{\Psi}_t^\mathsf{T}, \nonumber\\
    &\succeq \sigma_{\min}^2(\bm{\Psi}_t)I, \nonumber\\
    &\succeq \left( \frac{\sqrt{\lambda_g}}{\overline{b}}-G\nu_T
    \right)^2I,
\end{align}
which establishes Assumption \ref{assumption_effective_window} with the stated $\mu$.

For completeness, we also derive the EMA variation bound used after the lemma. The resolvent identity gives
\begin{equation}
    \bm{P}_{j+1}-\bm{P}_j = \bm{P}_{j+1}\left(\bm{B}_j-\bm{B}_{j+1}\right)\bm{P}_j.
\end{equation}
Also,
\begin{equation}
    \bm{S}_{j+1}-\bm{S}_j  =  (1-\beta)\left(\bm{A}_{j+1}-\bm{S}_j\right).
\end{equation}
Since $\|\bm{A}_{j+1}\|\leq G^2$ and $\|\bm{S}_j\|\leq G^2$,
\begin{equation}
    \|\bm{B}_{j+1}-\bm{B}_j\| = \|\bm{S}_{j+1}-\bm{S}_j\| \leq 2(1-\beta)G^2.
\end{equation}
Using $\|\bm{P}_j\|\leq1/\underline{b}$,
\begin{equation}
    \|\bm{P}_{j+1}-\bm{P}_j\|  \leq  \frac{2(1-\beta)G^2}{\underline{b}^2} =:\delta_P.
\end{equation}
By telescoping,
\begin{equation}
    \|\bm{P}_{t-k}-\bm{P}_t\| \leq k\delta_P,
\end{equation}
and consequently
\begin{align}
    \nu_T  &\leq  \delta_P  \left(  \sum_{k=0}^{T-1}k^2  \right)^{1/2}, \nonumber\\
    &= \frac{2(1-\beta)G^2}{\underline{b}^2} \sqrt{\frac{T(T-1)(2T-1)}{6}}.
\end{align}
This concludes the proof.

\subsection{Proof of Lemma \ref{lemma.dissipation}} \label{app.proof.one_step_dissipation}

We begin by expanding $V_{t+1}$ using $\bm{\psi}_t=\bm{B}_t^{-1}\bm{g}_t$ and $\chi_t=\bm{\psi}_t^\mathsf{T}\bm{e}_t$:
\begin{align}
    V_{t+1}  &=  \|\bm{e}_t-\eta\bm{\psi}_ts_t\|^2, \nonumber\\
    &= V_t-2\eta s_t\chi_t+\eta^2s_t^2\|\bm{\psi}_t\|^2. \label{eq.proof.dissipation_expand}
\end{align}
By Assumption \ref{assumption_d}, $\chi_t=\tilde{r}_t-\delta_t$, and by Assumption \ref{assumption.sectors}, $s_t=\phi(\tilde{r}_t)$. Therefore,
\begin{align}
    -2\eta s_t\chi_t   &=  -2\eta\phi(\tilde{r}_t)(\tilde{r}_t-\delta_t), \nonumber\\
    &=   -2\eta\tilde{r}_t\phi(\tilde{r}_t)   +   2\eta \delta_t\phi(\tilde{r}_t), \nonumber\\
    &\leq  -2\eta m\tilde{r}_t^2   +  2\eta M|\delta_t||\tilde{r}_t|, \nonumber\\
    &\leq  -\eta m\tilde{r}_t^2  +  \frac{\eta M^2}{m}\delta_t^2, \label{eq.proof.dissipation_cross}
\end{align}
where the last line follows from Young's inequality.

From Lemma \ref{lemma.Bt_eigenvalue_bounds},
\begin{equation}
    \|\bm{\psi}_t\|  \leq  \frac{G}{\underline{b}},
\end{equation}
and the growth bound in Assumption \ref{assumption.sectors} gives
\begin{equation}
    \eta^2s_t^2\|\bm{\psi}_t\|^2  \leq  \eta^2M^2\frac{G^2}{\underline{b}^2}\tilde{r}_t^2.
\end{equation}
Introducing these bounds in (\ref{eq.proof.dissipation_expand}) yields
\begin{equation}
    V_{t+1}  \leq  V_t  -  \eta\left( m-\eta M^2\frac{G^2}{\underline{b}^2}  \right)\tilde{r}_t^2   +     \frac{\eta M^2}{m}\delta_t^2.
\end{equation}
Under (\ref{eq.Lyapunov_lemma_step_size}),
\begin{equation}
    m-\eta M^2\frac{G^2}{\underline{b}^2}\geq\frac{m}{2},
\end{equation}
so
\begin{equation}
    V_{t+1}     \leq     V_t-\alpha_0\tilde{r}_t^2+\gamma_0\delta_t^2,
\end{equation}
with $\alpha_0=\eta m/2$ and $\gamma_0=\eta M^2/m$.

Next, from $\tilde{r}_t=\chi_t+\delta_t$ and the inequality $(x+y)^2\geq \frac{1}{2}x^2-y^2$,
\begin{equation}
    \tilde{r}_t^2  \geq \frac{1}{2}\chi_t^2-\delta_t^2.
\end{equation}
Thus,
\begin{align}
    V_{t+1}     &\leq  V_t  -  \frac{\alpha_0}{2}\chi_t^2   +   (\alpha_0+\gamma_0)\delta_t^2, \nonumber\\     &=
    V_t-\alpha\chi_t^2+\gamma \delta_t^2. \end{align}
Finally,
\begin{align}
    \|\bm{e}_{t+1}-\bm{e}_t\|     &=     \eta\|\bm{\psi}_t\||s_t|, \nonumber\\ 
    &\leq   \eta\frac{GM}{\underline{b}}|\tilde{r}_t|,
\end{align}
which establishes all three claims.

\subsection{Proof of Lemma \ref{lemma.windowed_obs_in_update_geometry}} \label{app.proof.update_geometry_lemma}

Fix $t\geq T-1$ and let $s:=t-T+1$. For every $k=0,\dots,T-1$,
\begin{align}
    \chi_{t-k}  &= \bm{\psi}_{t-k}^\mathsf{T}\bm{e}_{t-k}, \nonumber\\
    &= \bm{\psi}_{t-k}^\mathsf{T}\bm{e}_{s}  +  \bm{\psi}_{t-k}^\mathsf{T} \left( \bm{e}_{t-k}-\bm{e}_{s} \right).
\end{align}
Using $(x+y)^2\geq\frac{1}{2}x^2-y^2$ and summing over the window,
\begin{align}
    \sum_{k=0}^{T-1}\chi_{t-k}^2 &\geq \frac{1}{2} \sum_{k=0}^{T-1} \left( \bm{\psi}_{t-k}^\mathsf{T}\bm{e}_{s} \right)^2 \nonumber\\
    &\quad - \sum_{k=0}^{T-1} \left( \bm{\psi}_{t-k}^\mathsf{T} \left( \bm{e}_{t-k}-\bm{e}_{s} \right) \right)^2. \label{eq.proof.window_split}
\end{align}
Assumption \ref{assumption_effective_window} gives
\begin{equation}
    \sum_{k=0}^{T-1} \left( \bm{\psi}_{t-k}^\mathsf{T}\bm{e}_{s} \right)^2 \geq \mu\|\bm{e}_{s}\|^2 = \mu V_s. \label{eq.proof.window_effective_term}
\end{equation}

Define
\begin{equation}
    \Delta_t := \max_{0\leq k\leq T-1} \|\bm{e}_{t-k}-\bm{e}_{s}\|.
\end{equation}
Since $\|\bm{\psi}_{t-k}\|\leq G/\underline{b}$,
\begin{equation}
    \sum_{k=0}^{T-1} \left( \bm{\psi}_{t-k}^\mathsf{T} \left( \bm{e}_{t-k}-\bm{e}_{s}  \right) \right)^2 \leq T\frac{G^2}{\underline{b}^2}\Delta_t^2. \label{eq.proof.window_drift_term}
\end{equation}
The motion bound from Lemma \ref{lemma.dissipation} implies
\begin{align}
    \|\bm{e}_{t-k}-\bm{e}_{s}\| &\leq \sum_{j=s}^{t-1} \|\bm{e}_{j+1}-\bm{e}_{j}\|, \nonumber\\
    &\leq \eta\frac{GM}{\underline{b}} \sum_{j=s}^{t}|\tilde{r}_j|.
\end{align}
Therefore, by Cauchy--Schwarz,
\begin{equation}
    \Delta_t^2  \leq \eta^2\frac{G^2M^2}{\underline{b}^2}  T
    \sum_{j=s}^{t}\tilde{r}_j^2. \label{eq.proof.window_delta}
\end{equation}
Summing the sharper dissipation inequality (\ref{eq.one_step_Lyapunov_decr}) over $j=s,\dots,t$ gives
\begin{equation}
    V_{t+1}-V_s  \leq -\alpha_0\sum_{j=s}^{t}\tilde{r}_j^2  + \gamma_0\sum_{j=s}^{t}\delta_j^2.
\end{equation}
Using $V_{t+1}\geq0$ and $\sum_{j=s}^{t}\delta_j^2\leq T\overline{\delta}^2$,
\begin{equation}
    \sum_{j=s}^{t}\tilde{r}_j^2 \leq \frac{V_s}{\alpha_0}  + \frac{\gamma_0T}{\alpha_0}\overline{\delta}^2. \label{eq.proof.window_residual_sum}
\end{equation}
Combining (\ref{eq.proof.window_delta}) and (\ref{eq.proof.window_residual_sum}) yields
\begin{equation}
    \Delta_t^2  \leq  \eta^2\frac{G^2M^2T}{\underline{b}^2}  \left( \frac{V_s}{\alpha_0} + \frac{\gamma_0T}{\alpha_0}\overline{\delta}^2     \right). \label{eq.proof.window_delta_final}
\end{equation}
Finally, introducing (\ref{eq.proof.window_effective_term}), (\ref{eq.proof.window_drift_term}), and (\ref{eq.proof.window_delta_final}) into (\ref{eq.proof.window_split}) gives
\begin{align}
    \sum_{k=0}^{T-1}\chi_{t-k}^2  &\geq \frac{\mu}{2}V_s  - T\frac{G^2}{\underline{b}^2}\Delta_t^2, \nonumber\\
    &\geq  \left( \frac{\mu}{2}  -  \frac{\eta^2G^4M^2T^2}{\alpha_0\underline{b}^4}\right)V_s \nonumber\\
    &\quad  - \frac{\eta^2G^4M^2\gamma_0T^3}{\alpha_0\underline{b}^4} \overline{\delta}^2, \nonumber\\
    &= \theta V_{t-T+1}-\zeta\overline{\delta}^2.
\end{align}
The explicit condition (\ref{eq.lemma_windowed_step_size_cond}) is equivalent to $\theta>0$, completing the proof.

\subsection{Proof of Proposition \ref{proposition.ISS}} \label{app.proof.ISS}

We develop the proof in three phases / steps. In Step 1, we look at the residual-to-state behavior of the error at the end of $T$-length windows. In Step 2, we examine and bound the growth of the error within such windows. In Step 3, we compile the final bound for $\|\bm{e}_t\|$.

\subsubsection*{Step 1: window-to-window error behavior}

From Lemma \ref{lemma.dissipation},
\begin{equation}
    V_{t+1}\leq V_t-\alpha\chi_t^2+\gamma \delta_t^2. \label{eq.proof.iss.Lyap_dissipation}
\end{equation}
Fix any integer $n\geq0$ and sum over $t=nT,\dots,nT+T-1$:
\begin{equation}
    V_{(n+1)T}  \leq   V_{nT}  -  \alpha\sum_{k=0}^{T-1}\chi_{nT+k}^2  +  \gamma\sum_{k=0}^{T-1}\delta_{nT+k}^2. \label{eq.proof.iss.Lyap_dissipation_window}
\end{equation}
Using $\sum \delta_{nT+k}^2\leq T\overline{\delta}^2$ and Lemma \ref{lemma.windowed_obs_in_update_geometry} with $t=nT+T-1$,
\begin{equation}
    \sum_{k=0}^{T-1}\chi_{nT+k}^2  \geq  \theta V_{nT}-\zeta\overline{\delta}^2.
\end{equation}
Hence
\begin{equation}
    V_{(n+1)T}  \leq  (1-\alpha\theta)V_{nT}  +  (\alpha\zeta+\gamma T)\overline{\delta}^2. \label{eq.proof.iss.window_recursion}
\end{equation}
Define
\begin{equation}
    \rho:=1-\alpha\theta.
\end{equation}
Under (\ref{eq.proposition_step_size}), $\theta>0$ and
\begin{equation}
    0<\alpha\theta     \leq  \frac{\eta m\mu}{8}   <1,
\end{equation}
so $\rho\in(0,1)$. Iterating (\ref{eq.proof.iss.window_recursion}) gives
\begin{equation}
    V_{nT}  \leq \rho^nV_0  + \frac{\alpha\zeta+\gamma T}{\alpha\theta} \overline{\delta}^2.
\end{equation}
Taking square roots,
\begin{equation}
    \|\bm{e}_{nT}\|  \leq  q^n\|\bm{e}_0\|   + C\overline{\delta}, \label{eq.proof.iss.endpoint_bound}
\end{equation}
where
\begin{equation}
    q:=\sqrt{\rho}=\sqrt{1-\alpha\theta},  \qquad  C:=  \sqrt{\frac{\alpha\zeta+\gamma T}{\alpha\theta}}.
\end{equation}

\subsubsection*{Step 2: within-window error behavior}

Fix $n\geq0$ and take any $t$ such that $nT\leq t\leq(n+1)T$. Then
\begin{equation}
    \|\bm{e}_t\|  \leq  \|\bm{e}_{nT}\|  +  \sum_{j=nT}^{t-1}  \|\bm{e}_{j+1}-\bm{e}_j\|.
\end{equation}
Using (\ref{eq.one_step_motion_bound}), enlarging the sum to the full window, and applying Cauchy--Schwarz,
\begin{align}
    \|\bm{e}_t\|   &\leq  \|\bm{e}_{nT}\|   +  \eta\frac{GM}{\underline{b}}  \sum_{j=nT}^{(n+1)T-1}  |\tilde{r}_j|, \nonumber\\ 
    &\leq   \|\bm{e}_{nT}\|   +   \eta\frac{GM}{\underline{b}}  \sqrt{T  \sum_{j=nT}^{(n+1)T-1}   \tilde{r}_j^2  }. \label{eq.proof.iss.within_residual}
\end{align}
Summing (\ref{eq.one_step_Lyapunov_decr}) over the same window gives
\begin{equation}
    \sum_{j=nT}^{(n+1)T-1}\tilde{r}_j^2 \leq \frac{V_{nT}}{\alpha_0} + \frac{\gamma_0T}{\alpha_0}\overline{\delta}^2.
\end{equation}
Introducing this bound in (\ref{eq.proof.iss.within_residual}) and using $\sqrt{x+y}\leq\sqrt{x}+\sqrt{y}$ yields
\begin{equation}
    \|\bm{e}_t\| \leq A\|\bm{e}_{nT}\| + D\overline{\delta}, \label{eq.proof.iss.within_window}
\end{equation}
where
\begin{equation}
    A :=  1+ \eta\frac{GM}{\underline{b}} \sqrt{\frac{T}{\alpha_0}}, \qquad  D :=  \eta\frac{GM}{\underline{b}} T \sqrt{\frac{\gamma_0} {\alpha_0}}. \label{eq.proof.iss.AD}
\end{equation}

\subsubsection*{Step 3: final residual-to-state bound}

For any $t\geq0$, let $n=\lfloor t/T\rfloor$. Combining (\ref{eq.proof.iss.endpoint_bound}) and (\ref{eq.proof.iss.within_window}) gives
\begin{align}
    \|\bm{e}_t\|   &\leq  A\left( q^{\lfloor t/T\rfloor}\|\bm{e}_0\|  +  C\overline{\delta} \right)   +  D\overline{\delta}, \nonumber\\     &=     A q^{\lfloor t/T\rfloor}\|\bm{e}_0\|     +  L\overline{\delta}, 
\end{align}
where
\begin{equation}
    L:=AC+D.
\end{equation}
Thus the constants $q,A,L$ are explicit functions of the problem parameters, and the stated estimate establishes a residual-to-state, ISS-type bound for all $t$. When $\delta_t$ is independently specified as an exogenous input, the same estimate has the conventional input-to-state-stability interpretation.

\section{Implementation details for numerical experiments}

In this appendix, we describe the learning pipelines around Experiments \ref{sec.exp2} and \ref{sec.exp3}, as well as the experiment-specific realizations of our observability-aware controllers. Sections \ref{sec.directional}-\ref{sec.stability} develop and prove the idealized control law used in our theory, whereas the implementations below use numerically stabilized and architecture-specific realizations that preserve the same observability-aware geometry, while promoting robustness and tractability in finite-precision training, especially in practically realistic and meaningful parameter spaces that range between $O\left(10^3 \right) \sim O\left(10^6 \right)$.

\subsection{Experiment \ref{sec.exp2}} \label{app.sec.exp2}

The Human Activity Recognition (HAR) dataset includes 561 input features and six activity labels: WALKING, WALKING UPSTAIRS, WALKING DOWNSTAIRS, SITTING, STANDING, and LAYING, with an explicit and intentional split between training and test sets, with the latter obtained from different human subjects. We load the data with this standard train/test split, remap labels to $\{1, \dots, 6\}$, and standardize the features with a StandardScaler fit on the training split, which we also apply unchanged to the test split. 

We formulate the classification task as six one-versus-rest (OVR) binary classification problems, one for each $k \in \{1, \dots 6\}$. For a sample $\bm{x}_i \in \mathbb{R}^{561}$ with multiclass label $y_i$, we define the binary target for head $k$ as $y_{i,k}^{\mathrm{bin}} := \bm{1}\{y_i = k\}$. At inference time, the multiclass prediction is obtained by taking the largest OVR logit across the six heads. Without this being required, for computational efficiency our implementation updates and evaluates the six binary heads in parallel, while the independent random seed runs presented in our evaluation are executed serially. 

\subsubsection{Missingness construction}

To induce meaningful modality-level missingness in HAR, we hard-coded the stated HAR feature ordering and grouped the originally continuous features into six coarse sensing groups: body acceleration, gravity acceleration, body-acceleration jerk, body gyroscope, body-gyroscope jerk, and angle features. 

For each training example, we first sample a binary feature mask initialized by keeping each feature independently with probability 0.8. We then select exactly 4 out of the 6 coarse modality groups uniformly without replacement and set all features in those groups to zero. The resulting mask is applied once to the standardized datapoint and then reused whenever that datapoint is revisited during training. Thus, masks are persistent within a run and vary only with the random seed. Our evaluation reported in Section \ref{sec.exp2} uses a full (unmasked) test split as the validation set, allowing us to judge how each method learns the underlying mapping despite severe data missingness. 

\subsubsection{Models and losses}

We consider two model classes: 

\begin{itemize}
    \item \textbf{Logistic regression.} For head $k$, let $a_{i,k} := \langle \bm{\theta}_k , \bm{x}_i \rangle + b_k$ and  $p_{i,k} := \sigma (a_{i,k}) =1/\left(1 + e^{-a_{i,k}}\right)$, where $\bm{\theta}_k \in \mathbb{R}^{561}$, $b_k \in \mathbb{R}$, and $\displaystyle \bm{w}_k = \left[ \begin{array}{c} \bm{\theta}_k \\ b_k \end{array} \right] \in \mathbb{R}^{562}$ is the model's parameter vector for head $k$.

    The loss for head $k$ is 
    \begin{equation}
        \ell_k^{\mathrm{logistic}} = - \frac{1}{n} \sum_{i = 1}^{n} \left( y_{i,k}^{\mathrm{bin}} \log p_{i,k} + \left( 1 - y_{i,k}^{\mathrm{bin}} \right) \log (1 - p_{i,k}) \right) + \frac{\lambda}{2} \| \bm{\theta}_k \|_2^2,
    \end{equation} where $n$ is the minibatch size. The bias is not regularized. Our logistic regression model has 562 learnable parameters per class, totaling $6 \times 562 = 3372$ learnable parameters overall.

    \item \textbf{Shallow neural network.} For head $k$, we use a single-hidden-layer MLP with width 12, and learnable unknowns: $\bm{\Theta}_k \in \mathbb{R}^{12 \times 561}$, $\bm{b}_{1,k} \in \mathbb{R}^{12}$, $\bm{v}_k \in \mathbb{R}^{12}$, and $b_{2,k} \in \mathbb{R}$.
    For sample $i$, the forward map is:
    \begin{subequations}
    \begin{align}
        \bm{u}_{i,k} &= \bm{\Theta}_{k} \bm{x}_i + \bm{b}_{1,k}, \label{eq.app.ex2.fwd1} \\
        \bm{h}_{i,k} &= \tanh ( \bm{u}_{i,k}), \label{eq.app.ex2.fwd2} \\
        a_{i,k} &= \langle \bm{v}_k, \bm{h}_{i,k} \rangle + b_{2,k}, \label{eq.app.ex2.fwd3} \\
        p_{i,k} &= \sigma (a_{i,k}) = \frac{1}{1 + e^{-a_{i,k}}}, \label{eq.app.ex2.fwd4}
    \end{align}
    \end{subequations}
    and the minibatch loss is:
    \begin{equation}
        \ell_k^{\mathrm{MLP}} = -\frac{1}{n} \sum_{i = 1}^{n} \left( y_{i,k}^{\mathrm{bin}} \log p_{i,k} + \left( 1 - y_{i,k}^{\mathrm{bin}} \right) \log (1 - p_{i,k}) \right) + \frac{\lambda}{2} \left(  \| \bm{\Theta}_k \|_F^2 + \| \bm{v}_k \|_2^2 \right).
    \end{equation}
    Each per-class head contains $12 \cdot 561 = 6732$ input-to-hidden weights, 12 hidden biases, 12 output weights, and one output bias, for a total of $12 \cdot 561 + 12 + 12 +1 = 6757$ learnable parameters per head. The weights are initialized from $0.05 \, \mathcal{N}(0,1)$ with class-dependent random seeds, while all biases start at zero.

\end{itemize}

\subsubsection{Full order controller implementation}
In Section \ref{sec.exp2} we compare four update laws: full SGD, masked SGD, scalar observability-aware adaptation, and directional observability-aware adaptation. For the full and masked baselines, the update is ordinary gradient descent with the common step size $\eta = 0.005$. 

At every optimization step, we draw a minibatch of $n = 256$ training samples. For each one-versus rest head $k \in \{1, \dots, 6\}$, we form a full-order observability matrix,
\begin{equation}
    \bm{A}_{k,t} \in \mathbb{R}^{p_k \times p_k}
\end{equation}
where $p_k$ is the number of learnable parameters of that head:
\begin{equation}
    p_k = 562 \;\; \text{for logistic regression}, \qquad \qquad p_k = 6757 \;\; \text{for the shallow neural network}.
\end{equation}

In view of Experiment \ref{sec.exp3} and its significantly larger parameter space, we emphasize that here, in Experiment \ref{sec.exp2}, the directional controller is implemented without any block or other approximation: each binary head maintains its own \textit{dense} controller state.

For each head $k$, the instantaneous full-order actuation matrix is constructed from the minibatch Jacobian rows. Specifically, if $\bm{J}_{k,t}^{(i)} \in \mathbb{R}^{1 \times p_k}$ denotes the Jacobian row of the scalar output of head $k$ for minibatch sample $i$, then the controller uses
\begin{equation}
    \bm{A}_{k,t} = \frac{1}{n} \sum_{i=1}^{n} \left(\bm{J}_{k,t}^{(i)} \right)^\mathsf{T} \bm{J}_{k,t}^{(i)} \in \mathbb{R}^{p_k \times p_k}.
\end{equation}
Equivalently, if we stack the minibatch Jacobian rows into the matrix $\bm{G}_{k,t} \in \mathbb{R}^{n \times p_k}$, then $\bm{A}_{k,t} = \frac{1}{n} \bm{G}_{k,t}^\mathsf{T} \bm{G}_{k,t}$. 

For the \textbf{logistic regression} head, with learnable parameters $\bm{w}_k = \left[ \bm{\theta}_k^\mathsf{T} \; b_k \right]^\mathsf{T} \in \mathbb{R}^{562}$, let the minibatch inputs be augmented with a bias coordinate:
\begin{equation}
    \overline{\bm{x}}_i = \left[ \begin{array}{c} \bm{x}_i \\ 1 \end{array} \right] \in \mathbb{R}^{562}.
\end{equation}
For sample $i$, the scalar output is $p_{i,k} = \sigma \left( \bm{\theta}_k^\mathsf{T} \overline{\bm{x}}_i \right) \in \mathbb{R}$, and the corresponding Jacobian row of the scalar output with respect to the parameter vector $\bm{w}_k$ is:
\begin{equation}
    \bm{J}_{k,t}^{(i)} = p_{i,k} (1 - p_{i,k}) \bm{x}_i^\mathsf{T}.
\end{equation}
Therefore, the minibatch Jacobian matrix is
\begin{equation}
    \bm{G}_{k,t} = \left[ \begin{array}{c} p_{1,k}(1-p_{1,k}) \overline{\bm{x}}_1^\mathsf{T} \\ \vdots \\ p_{n,k}(1-p_{n,k}) \overline{\bm{x}}_n^\mathsf{T} \end{array} \right],
\end{equation}
and the full-order controller matrix is:
\begin{equation}
    \bm{A}_{k,t} = \frac{1}{n} \sum_{i=1}^{n} \left( p_{i,k}(1-p_{i,k}) \right)^2  \overline{\bm{x}}_i  \overline{\bm{x}}_i^\mathsf{T}.
\end{equation}
The corresponding logistic loss gradient used in the update is:
\begin{equation}
    \nabla \ell_{k,t}^{\mathrm{logistic}} = \frac{1}{n} \sum_{i=1}^n \left( p_{i,k} - y_{i,k}^\mathrm{bin} \right)  \overline{\bm{x}}_i + \lambda \left[ \begin{array}{c} \bm{\theta}_k \\ 0 \end{array} \right].
\end{equation}

For the \textbf{shallow neural network} head, the learnable parameter is the vector:
\begin{equation}
    \bm{w}_k = \mathrm{vec} \left( \bm{\Theta}_k, \bm{b}_{1,k}, \bm{v}_k, b_{2,k} \right) \in \mathbb{R}^{6757},
\end{equation}
where $\bm{\Theta}_k \in \mathbb{R}^{12 \times 561}$, $\bm{b}_{1,k} \in \mathbb{R}^{12}$, $\bm{v}_k \in \mathbb{R}^{12}$, and $b_{2,k} \in \mathbb{R}$. For sample $i$, take the forward map to be given by (\ref{eq.app.ex2.fwd1})-(\ref{eq.app.ex2.fwd4}). Let
\begin{equation}
    s_{i,k} = p_{i,k} (1 - p_{i,k}), \qquad \bm{d}_{i,k} = (1 - h_{i,k}^2) \odot \bm{v}_k \in \mathbb{R}^{12},
\end{equation}
where the square and the Hadamard product are taken elementwise. Then the Jacobian row of the scalar output $p_{i,k}$ with respect to the packed parameter vector $\bm{w}_k$ is:
\begin{equation}
    \bm{J}_{k,t}^{(i)} = \mathrm{pack} \left( s_{i,k} \bm{d}_{i,k} \bm{x}_i^\mathsf{T},\; s_{i,k} \bm{d}_{i,k} , \; s_{i,k} \bm{h}_{i,k}, \; s_{i,k}\right) \in \mathbb{R}^{1 \times 6757}. \label{eq.app.ex2.mlp_rows}
\end{equation}
Hence
\begin{equation}
    \bm{A}_{k,t} = \frac{1}{n} \sum_{i=1}^{n} \left( \bm{J}_{k,t}^{(i)} \right)^\mathsf{T} \bm{J}_{k,t}^{(i)} = \frac{1}{n} \bm{G}_{k,t}^\mathsf{T} \bm{G}_{k,t},
\end{equation}
with $\bm{G}_{k,t} \in \mathbb{R}^{n \times 6757}$ formed by stacking the rows given by (\ref{eq.app.ex2.mlp_rows}). The corresponding minibatch loss gradient is the packed vector:
\begin{align}
    \nabla \ell_{k,t}^{\mathrm{MLP}} &= \mathrm{pack} \left( \frac{1}{n} \sum_{i=1}^{n} \left(p_{i,k} - y_{i,k}^\mathrm{bin} \right) \bm{d}_{i,k} \bm{x}_i^\mathsf{T} + \lambda \bm{\Theta}_k , \;\;\frac{1}{n} \sum_{i=1}^n \left(p_{i,k} - y_{i,k}^\mathrm{bin} \right) \bm{d}_{i,k}, \right. \nonumber \\
    & \left. \qquad  \qquad  \qquad  \qquad \qquad \frac{1}{n} \sum_{i=1}^n \left(p_{i,k} - y_{i,k}^\mathrm{bin} \right) \bm{h}_{i,k} + \lambda \bm{v}_k , \;\; \frac{1}{n} \sum_{i=1}^n \left(p_{i,k} - y_{i,k}^\mathrm{bin} \right) \right).
\end{align}

Given $\bm{A}_{k,t}$, the \textbf{scalar-aware update} for head $k$ is:
\begin{equation}
    \gamma_{k,t} = \mathrm{tr}(\bm{A}_{k,t}), \quad\alpha_{k,t} = \min \left( \frac{\eta}{\epsilon + \gamma_{k,t}} , \; \alpha_{\mathrm{max}} \right), \qquad \bm{w}_{k+1} = \bm{w}_k - \alpha_{k,t} \nabla \ell_{k,t},
\end{equation}
where $\ell_{k,t} \in \left\{ \ell_{k,t}^\mathrm{logistic}\;, \; \ell_{k,t}^{\mathrm{MLP}} \right\}$, according to the model employed. In the reported runs of the scalar-aware controller, we used $\epsilon = 1$ and $\alpha_\mathrm{max} = 1$.

Given $\bm{A}_{k,t}$, to realize the \textbf{directional-aware update} for head $k$ we use
\begin{align}
    \bm{S}_{k,0} = 0, \quad \bm{S}_{k,t} &= \beta \bm{S}_{k,t-1} + (1 - \beta) \bm{A}_{k,t}, \\ \bm{R}_{k,t} &= \bm{S}_{k,t} + \epsilon I, \\
    \bm{q}_{k,t} &= \bm{R}_{k,t}^{-1} \nabla \ell_{k,t},
\end{align}
where $\ell_{k,t} \in \left\{ \ell_{k,t}^\mathrm{logistic}\;, \; \ell_{k,t}^{\mathrm{MLP}} \right\}$, according to the model employed. Then, we form the state-dependent gain
\begin{equation}
    \alpha_{k,t} = \min \left( \frac{\eta}{1 + \langle \nabla \ell_{k,t}, \bm{q}_{k,t} \rangle} , \; \alpha_{\mathrm{max}}\right), \label{eq.app.ex2.aug_alpha}
\end{equation}
and perform the parameter update according to $\bm{w}_{k+1} = \bm{w}_k - \alpha_{k,t} \bm{q}_{k,t}$. 

If $\|\alpha_{k,t} \bm{q}_{k,t} \| > 5$, we scale it to have norm exactly $5$. Each of the six OVR heads carries its own empirical controller state $\bm{S}_{k,t}$, and the six heads are updated in parallel inside each training iteration. In the reported runs of the directional-aware controller, we used $\epsilon = 0.05$, $\beta = 0.99$, and $\alpha_\mathrm{max} = 1$.

\emph{\textbf{Numerical safeguards for high-dimensional parameter spaces}} 

We emphasize one design choice shared with the ideal controller of Sections \ref{sec.directional}-\ref{sec.stability}, followed by one practical departure used in this experiment. 

First, consistently with (\ref{eq.directional.B_recursive}), the implementation does not recursively accumulate the ridge term inside the stored controller state. Instead, it maintains the unregularized EMA state $\bm{S}_{k,t}$ and adds $\epsilon I$ only when forming the regularized preconditioning matrix $\bm{R}_{k,t}$. This keeps the recent actuation history and the numerical regularization conceptually distinct. This separation is particularly important in a full order realization and especially in large spaces where such a full order realization is still practically feasible. In this particular experiment, each shallow network OVR head carries a dense $6757 \times 6757$ controller state. If the ridge term were accumulated recursively inside the stored EMA state, then it would no longer act merely as a small numerical floor: with $\beta = 0.99$, the isotropic contribution would build up geometrically toward $\epsilon / (1-\beta) I$, thereby injecting a large persistent isotropic background into a controller whose main purpose is to retain recent, data-driven directional actuation structure. In high-dimensional spaces, such information would be numerically dominated by the isotropic background. We therefore keep the EMA state purely empirical and add $\epsilon I$ only when forming the matrix to be inverted. This preserves a fixed conditioning floor, while preventing the numerical regularization from overwhelming the anisotropy that the controller is meant to exploit.

The practical departure is that the implementation above does not apply the raw preconditioned step $-\eta \bm{R}_{k,t}^{-1} \nabla \ell_{k,t}$. Rather, it rescales the preconditioned direction using the state-dependent factor $\alpha_{k,t}$, as given by (\ref{eq.app.ex2.aug_alpha}), and then clips the resulting increment if necessary. This extra factor acts as a trust-region-like safeguard: it preserves the observability-aware direction while shrinking updates whose preconditioned energy is large, thereby improving robustness under dense finite-batch full-order control. Given that $\bm{q}_{k,t} = \bm{R}_{k,t}^{-1} \nabla \ell_{k,t}$, the denominator in $\alpha_{k,t}$ is $1 + \langle \nabla \ell_{k,t}, \bm{q}_{k,t} \rangle$, that is, one plus the squared gradient norm measured in the inverse-observability geometry induced by $\bm{R}_{k,t}^{-1}$. Thus, our safeguard is not altering the controller's direction (as it enters as a scalar scaling mechanism applied uniformly to all controller components), but is regulating its magnitude precisely when the dense full-order machinery indicates that the proposed motion is too energetic relative to the current observability structure. This is especially relevant in the present experiment, where each OVR head uses a dense controller of dimension 562 (in logistic regression) or 6757 (in the shallow network), and where $\bm{A}_k$ is built from finite-batch Jacobian outer products rather than from any population-level quantity. In this setting, the additional scalar scaling and final norm clipping provide a practical finite-precision and finite-sample safeguard, while respecting and preserving the directional geometry that the ideal controller is meant to exploit.

\subsection{Experiment \ref{sec.exp3}} \label{sec.app.exp3}

Our pipeline around ASID-v2 differs from the original pixelwise sea ice mapping objective of \citet{malmgren2020convolutional} to perform dense sea ice prediction over an image patch, with a convolutional neural network designed to fuse high resolution Sentinel-1 imagery with much coarser AMSR2 measurements by injecting the low resolution information deeper in the network rather than naively resampling everything to the SAR grid. We instead formulate a patch-level scalar regression problem: given a localized multisensor patch, we aim to predict the mean sea ice concentration in that patch. This reformulation is deliberately simpler than full semantic segmentation, but it remains operationally meaningful. A regional estimate of sea ice concentration over a 10 km $\times$ 10 km area is already useful for situational awareness and downstream decision support, and it gives us a clean setting in which to introduce and study the effect of missing modalities and missing measurements without conflating them with the many additional complications of dense per-pixel decoding. In particular, it lets us ask whether our learner can still move coherently towards a credible concentration estimate when one sensing stream disappears altogether or when only fragments of the remaining streams are available.

\subsubsection{Convolutional model}

Our predictor model is a compact two-branch convolutional regressor built around the strong resolution asymmetry of the underlying sensing problem. Each training example consists of a 250 $\times$ 250 SAR patch, corresponding to a 10 km $\times$ 10 km footprint at the dataset's 40 m SAR spacing, and a co-located 5 $\times$ 5 AMSR patch, with the scalar target defined as the mean charted sea ice concentration over the same area (in \%). The SAR branch receives 6 channels: the two NERSC-corrected SAR backscatter images, two validity-mask channels, incidence angle, and distance-to-land. The AMSR branch receives 15 channels: 14 brightness-temperature channels, alongside one AMSR validity-mask channel. The SAR stream is encoded by a deeper convolutional tower with three successive downsampling blocks, reflecting the fact that the SAR input carries the fine spatial texture and geometric detail of the scene. The AMSR stream is encoded by a shallower convolutional tower with no aggressive downsampling, reflecting the tiny native grid and its role as coarse radiometric context rather than fine image structure. Each branch is then globally pooled and projected to a learned feature vector, the two embeddings are concatenated, and a small multilayer regression head maps the fused representation to a single scalar prediction of patch-level sea ice concentration. When a whole modality is absent during training, the corresponding branch embedding is explicitly gated by its present/absent flag before fusion. In the reported configuration, the model has 368,753 trainable parameters; we use GroupNorm throughout, SGD with step size 0.005, batch size $8$, and evaluate on unseen validation scenes every 100 training steps, up to 20,000 training total steps. Because the ASID-v2 labels are heavily skewed toward low-concentration and open-water conditions, we also use a weighted random sampler with concentration bins $[0, 5, 25, 75, 100]$ to avoid having the learner dominated by the rather frequently seen near-zero patches.

\subsubsection{Data missingness}
During training, we introduce synthetic feature data missingness, while keeping validation data clean and fully present. At each training example contained in the training set, an entire sensing modality is removed with probability 0.60; conditional on such a modality-level removal, SAR is the missing modality with probability 0.60, and AMSR with probability 0.40. In addition, partial within-modality missingness is triggered with probability 0.15, in which case 20\% of the respective measurements are zeroed entrywise. Thus, the learner must operate in a regime where it often sees only one sensing family and, even then, may see an internally incomplete version of it.

\subsubsection{Controller implementation and approximations}
For a model with $d=368{,}753$ trainable parameters, maintaining a full dense observability matrix $\mathbf{B}_t$ is impractical. We therefore use a block-diagonal approximation in which the controller state is maintained over contiguous parameter chunks of size 128. Within each learnable tensor, the parameters are flattened in their native tensor order and partitioned into consecutive chunks of 128 parameters. For each chunk, we maintain a local block of $\mathbf{B}_t$ by exponentially averaging the outer products of the chunk-restricted loss gradient $\nabla_{\bm{w}}\ell_t$, which serves here as a practical surrogate for the ideal actuation object, with $\beta=0.99$; the resulting block states are regularized with $\epsilon=0.005$ \textit{when solving for the preconditioned update} (as in Experiment \ref{sec.exp2}, we maintain an unregularized EMA state for each block and add the $\epsilon{I}$ regularization only when forming the local matrix used for the preconditioned solve, rather than recursively accumulating the ridge term inside the stored EMA state).

The chunk size of 128 is chosen as a deliberate engineering trade-off. Smaller neighborhoods (e.g., size 10) would generate tens of thousands of blocks, incurring prohibitive bookkeeping overhead, while collapsing the approximation toward a nearly diagonal (per-parameter) controller that captures almost no local directional structure. Conversely, substantially larger chunks (e.g., size 1000) would produce matrices whose storage footprint and $O(k^3)$ inversion cost become impractical. At 128 parameters per block, the local matrices remain modest in size ($\approx$16k elements), inexpensive to store and invert, and yet large enough to encode meaningful local actuation geometry.

Although this block-diagonal approximation discards explicit cross-chunk couplings and the local neighborhoods are not necessarily semantically contiguous and directly adjacent in the model architecture, the controller nevertheless promotes coherent multimodal cooperation. Because the SAR and AMSR2 encoders are fused downstream through concatenated embeddings that feed a shared scalar regression head, as are noncontiguous parts within either encoder, parameters belonging to distant, seemingly unrelated chunks remain repeatedly co-actuated by the same downstream signals and common loss. Moreover, the actuation geometry encoded in $\mathbf{B}_t$ is itself shaped by the controller in a closed-loop fashion: since each update is preconditioned by the current $\mathbf{B}_t^{-1}$, the controller participates in creating the very directional structure it later exploits. Consequently, even in a high-dimensional model where parameters can evolve in many directions, the directional observability-aware controller can still bias the network toward more coherent multimodal \textit{and} intra-model cooperation by reshaping updates according to the empirically observed actuation geometry revealed through shared downstream fusion, rather than allowing uncontrolled masked SGD to establish arbitrary or inconsistent relationships.

In this experiment we focus only on our directional controller, as we already observed the scalar controller’s isotropic throttling to be too indiscriminate for high-missingness and multimodal settings.

\subsubsection{Dataset split rotations}
We constructed patch datasets from a local subset of 25 ASID-v2 netCDF scenes by scanning the SAR grid, retaining informative patches, and pairing each retained SAR patch with its co-located AMSR crop and scalar ice concentration target. To assess sensitivity to train/validation composition, we created three cyclic split rotations (A, B, and C) of the same eligible scene pool, with rotation A serving as the baseline. These rotations produce materially different patch-level target distributions: rotation A has training patches with mean 27.22 (std 39.22) and validation patches with mean 27.13 (std 39.44); rotation B has training mean 33.23 (std 40.98) and validation mean 16.09 (std 33.17); rotation C has training mean 20.80 (std 36.39) and validation mean 48.73 (std 40.92). Absolute MAE values are therefore not directly comparable across rotations. Instead, the rotations test whether the directional controller consistently reaches useful error regimes early and suppresses large validation excursions under substantially varying train/validation target compositions.

\end{document}